\documentclass[a4paper,english]{article}
\usepackage{graphicx}

\usepackage{amsmath,amsfonts,amssymb,latexsym,epsfig}
\usepackage{mathrsfs}
\usepackage{verbatim}
\usepackage{latexsym}
\usepackage{amsthm}
\usepackage{amssymb}
\usepackage{subfig}
\usepackage{graphics}
\usepackage{amsbsy}
\usepackage{fullpage}
\usepackage{enumerate}
\usepackage{times}
\usepackage{multirow}
\usepackage{soul}
\usepackage[normalem]{ulem}

\newtheorem{theorem}{Theorem}[section]
\newtheorem{lemma}[theorem]{Lemma}

\newtheorem{remark}[theorem]{Remark}

\newtheorem{corollary}[theorem]{Corollary}
\newtheorem{prop}[theorem]{Proposition}
\newtheorem{assumptions}[theorem]{Assumptions}

\numberwithin{equation}{section}
\usepackage{amscd}
\usepackage{tikz}
\usepackage{xcolor}
\usepackage[english]{babel}
\usepackage[latin1]{inputenc}
\usepackage{times}
\usepackage[T1]{fontenc}
\usepackage{graphicx}
\usepackage{amsmath,amssymb}
\usepackage{slidesec}
\usepackage{verbatim}
\usepackage{url}
\usepackage{epsf}
\usepackage{amsmath}
\usepackage{graphics}
\usepackage{amsfonts}
\usepackage{amsbsy}
\usepackage{lscape}
\usepackage{enumerate}
\usepackage{amsthm}
\usepackage{amssymb}
\usepackage{exscale}

\newcommand{\IE}{\mathbb{E}}
\newcommand{\IR}{\mathbb{R}}
\newcommand{\bigo}{{\mathcal O}}

\newcommand{\cL}{{\mathcal L}}

\DeclareMathOperator*{\Ddiv}{div}
\begin{document}
\setlength{\baselineskip}{12pt}
% \title{Lie-Trotter Splitting Algorithms for Simulating Langevin Dynamics and Applications to Nonreversible MCMC}
\title{Nonreversible Langevin Samplers: Splitting Schemes, Analysis and Implementation}
\author{A.B. Duncan\\
        School of Mathematical and Physical Sciences\\
        University of Sussex\\
        Falmer\\
        and\\
G.A. Pavliotis\\
        Department of Mathematics\\
    Imperial College London \\
        London SW7 2AZ, UK \\
        and \\
        K.C. Zygalakis \\
        School of Mathematics \\
        University of Edinburgh \\
        Edinburgh, EH9 3FD, UK
                    }

\maketitle

\begin{abstract}
For a given target density $\pi:\mathbb{R}^{d}\mapsto
\mathbb{R}$, there exist an infinite number of diffusion processes
which have unique invariant density $\pi$. As observed in a number of papers
\cite{duncan2016variance,rey2014irreversible,rey2014variance} samplers based on nonreversible diffusion processes can significantly outperform their reversible counterparts both in terms of asymptotic variance and rate of convergence to equilibrium. In this paper, we take advantage of this in order to construct efficient sampling algorithms based on
the Lie-Trotter decomposition of a nonreversible diffusion process into reversible and nonreversible components.  We show that samplers based on this scheme can significantly outperform standard MCMC methods, at the cost of introducing some controlled bias.  In particular, we prove that numerical integrators constructed according to this decomposition are geometrically ergodic and characterize fully their asymptotic bias and variance, showing that the sampler inherits the good mixing properties of the underlying nonreversible diffusion.  This is illustrated further with a number of numerical examples ranging from highly correlated low dimensional distributions, to logistic regression problems in high dimensions as well as inference for spatial
models with many latent variables.
\end{abstract}

%\section*{What's left to do}
%\begin{enumerate}
%    \item Discussion section needs to be written.
%    \item \sout{Spatial example figures.}
%    \item \sout{Redo banana vanity plot.}
%    \item \sout{Fill in missing details in Numerical Experiments section (banana, log regression and spatial mala).}
%    \item \sout{Asymptotic variance proof, some details have been glossed over, add them.}
%    \item Fix overful equation lines, restructure some equations to save space.
%    \item \sout{Bias plots are misleading...redo.}
%\end{enumerate}

\section{Introduction}
\label{sec:intro}

Consider the problem of computing expectations with respect to a probability distribution with smooth density $\pi(x)$, known only up to a normalization constant, i.e. we wish to evaluate
\begin{equation}
\label{e:sample}
\pi(f) = \int_{\mathbb{R}^d}f(x)\pi(x)\,dx.
\end{equation}
For high dimensional distributions, deterministic techniques are no longer tractable.  On the other hand, probabilistic methods do not suffer the same curse of dimensionality and thus are often the method of choice.  One such approach is \emph{Markov Chain Monte Carlo} (MCMC) which is based on the construction of a Markov process on $\mathbb{R}^d$ whose unique invariant distribution is $\pi(x)$.  Due to their simplicity and wide applicability, Markov chains based on Metropolis-Hastings (MH) transition kernels \cite{hastings1970monte,metropolis1953equation} and their numerous variants remain the most widely used scheme for sampling from a general target probability distribution, despite having been introduced over 60 years ago.  As there are infinitely many Markov processes which are ergodic with respect to a given target distribution $\pi$, a natural question is whether a Markov process can be chosen which is more efficient, in terms of convergence to equilibrium and mixing.   Metropolized schemes are reversible Markov chains by construction, i.e. they satisfy \emph{detailed balance}.  It is a well documented fact that nonreversible chains might convergence to equilibrium faster than reversible ones \cite{neal2004improving,diaconis2000analysis,mira2000non}.   Various MCMC schemes have been proposed which are based on the general idea of breaking reversibility by introducing an augmented target measure on an extended state space, along with dynamics which are invariant with respect to the augmented target measure.  For discrete state spaces, the lifting method \cite{diaconis2000analysis,hukushima2013irreversible,turitsyn2011irreversible} is one such approach, where the Markov chain is ``lifted'' from the state space $E$ to $E \times \lbrace 1, -1 \rbrace$.  The transition probabilities in each copy of $E$ are modified to introduce transitions between the copies to preserve the invariant distribution but now promote the sampler to generate long trajectories.  For continuous state spaces, analogous approaches involve augmenting the state space with a velocity/momentum variable and constructing Makovian dynamics which are able to mix more rapidly in the augmented state space.  Such methods include Hybrid Monte Carlo (HMC) methods, inspired by Hamiltonian dynamics.   While the standard construction of HMC \cite{duane1987hybrid,neal2011mcmc} is reversible, it is straightforward to construct dynamics based on the Generalized HMC scheme \cite{horowitz1991generalized} which will not be reversible, see also \cite{ottobre2016function} and more recently \cite{ma2016unifying}.
\\\\
Deferring issues of simulation until later, another candidate Markov process for sampling from $\pi$ is the diffusion  $(X_t)_{t\geq 0}$ defined by the following It\^{o} stochastic differential equation (SDE):
\begin{equation}
\label{eq:sde1}
	dX_t = b(X_t)\,dt + \sqrt{2}\,dW_t,
\end{equation}
where $W_t$ is a standard $\mathbb{R}^d$--valued Brownian motion and $b:\mathbb{R}^d \rightarrow \mathbb{R}^d$ is a smooth vector field which satisfies
\begin{equation} \label{eq:drift}
	b(x) = \nabla \log\pi(x) + \gamma(x),\quad \nabla\cdot(\pi(x)\gamma(x)) = 0,
\end{equation}
for some smooth vector field $\gamma$ on $\mathbb{R}^d$ satisfying some mild assumptions (c.f. Proposition \ref{prop:main}).  The process $X_t$ is nonreversible if and only if $\gamma \neq 0$.
By the Birkhoff ergodic theorem,
$$
	\lim_{T\rightarrow \infty}\frac{1}{T}\int_0^T f(X_s)\,ds = \mathbb{E}_{\pi}[f],\quad f \in L^1(\pi),
$$
and thus one can use
$$
	\pi_T(f) := \frac{1}{T}\int_0^T f(X_s)\,ds
$$
as an estimator for $\mathbb{E}_{\pi}[f]$, for $T$ sufficiently large.  A natural way to measure the efficiency of such estimator is the mean square error (MSE) given by
\begin{equation} \label{eq:MSE}
\text{MSE}(T):= \IE |\pi_{T}(f) -\pi(f) |^{2}.
\end{equation}
Under appropriate conditions on $X_t$ and $f$, the estimator $\pi_T(f)$ will satisfy a \emph{central limit theorem}, i.e.
\begin{equation} \label{eq:CLT}
	\lim_{T\rightarrow+\infty} \sqrt{T}\left(\pi_T(f) - \mathbb{E}_{\pi}[f]\right) = \mathcal{N}(0, 2\sigma^2(f)),
\end{equation}
where $\sigma^2(f)$ is the \emph{asymptotic variance} of the estimator $\pi_T(f)$ which can be expressed by
\begin{equation} \label{eq:var_ass}
\sigma^{2}(f):= \left \langle \phi ,(-\mathcal{L})\phi\right \rangle_{\pi},
\end{equation}
where $\mathcal{L}$ is the infinitesimal generator of \eqref{eq:sde1} and $\phi$ is the mean zero solution of the following Poisson equation on $\mathbb{R}^d$,
\begin{equation} \label{eq:Poisson}
-\mathcal{L}\phi=f-\pi(f).
\end{equation}
This relationship can be used to simplify the expression for the
MSE \eqref{eq:MSE} and decompose it in terms of {bias} $\mu_T(f)$ and \emph{variance} $\sigma^2_T(f)$ as follows
 \[
\IE |\pi_{T}(f) -\pi(f) |^{2}= (\IE\pi_{T}(f)-\pi(f) )^{2}+\IE(\pi_{T}(f)-\IE\pi_{T}(f))^{2}=\left(\mu_{T}(f)\right)^2+\sigma^{2}_{T}(f).
\]
For large $T$, the variance satisfies $\sigma^{2}_{T}(f) \simeq T^{-1}\sigma^{2}(f)$, while $\mu_T(f)^2 = o(T^{-1})$.  Since $\gamma(x)$ is not uniquely defined in \eqref{eq:drift}, a natural question is how it  should be
chosen  to ensure that for a given time $T$, the MSE in \eqref{eq:MSE} is as small as possible. This can be achieved in two manners, the first by
maximising the $L^2(\pi)$-spectral gap associated with \eqref{eq:sde1} as studied in \cite{lelievre2013optimal,wu2014} and hence increasing the speed with which $\mu_{T}$ converges to zero. In general, maximising the $L^2(\pi)$-spectral gap is challenging.  An alternative is to choose $\gamma(x)$ in such a way so as to reduce the asymptotic variance $\sigma^{2}(f)$.  It should be emphasised that the optimal choice will be different for each case.  In particular in \cite{duncan2016variance,rey2014irreversible,rey2014variance}, it was shown that the choice $\gamma(x)=0$, which
corresponds to using reversible dynamics, gives the maximum value of asymptotic variance for a given choice of diffusion tensor.  In particular, introducing a nonreversible perturbation will never decrease the performance of an estimator based on Langevin dynamics, both in terms of convergence to equilibrium and asymptotic variance.
\\\\
In general (\ref{eq:sde1}) cannot be simulated exactly, and one typically resorts to a discretisation of the SDE, denoted by $\widehat{X}^{\Delta t}_{n}$, in order to approximate $\pi(f)$. In particular, the following ergodic average is used
\begin{equation} \label{eq:ergodic_num}
\widehat{\pi}^{\Delta t}_{T}(f):=\frac{1}{N}\sum_{k=0}^{N}f(\widehat{X}^{\Delta t}_{k}), \quad  N \Delta t=T.
\end{equation}
Extra caution has to be taken in order to ensure that the above quantity converges in the limit of $T \rightarrow \infty$  since even  if (\ref{eq:sde1}) is ergodic (or even exponentially ergodic), this will not necessarily be the case for its numerical discretisation  \cite{roberts2002langevin, stramer1999langevin, stramer1999langevin2}.  In addition, even when the numerical discretization is ergodic and thus
\begin{equation}
\lim_{T \rightarrow \infty} \widehat{\pi}_{T}(f)=  \widehat{\pi}^{\Delta t}(f)=\int_{\mathbb{R}^{d}}f(x)\widehat{\pi}^{\Delta t}(x)dx,
\end{equation}
it is not true in general that $\widehat{\pi}^{\Delta t}=\pi$, since the underlying numerical discretization introduces bias in the estimation of $\pi(f)$ (see
\cite{TaT90,AVZ13,AVZ15}). One way to eliminate such bias is through Metropolization \cite{smith1993bayesian,tierney1994markov}, \emph{i.e.} the introduction of an accept-reject step  that ensures that the corresponding Markov chain is ergodic with respect to the target distribution $\pi$. However, such bias elimination might not be
advantageous in practice since the Metropolised chain will be reversible by construction, thus eliminating any benefit introduced by the nonreversible perturbation $\gamma$.
\\\\
When computing expectations of distributions with expensive likelihoods,   it might be too costly to sample a long Markov chain trajectory.  If an appropriate nonreversible Langevin dynamics \eqref{eq:sde1} can be introduced which does give rise to a dramatic reduction in asymptotic variance,  then it might be advantageous  to permit a controlled amount of bias in exchange for needing to sample fare less.  This bias-variance
tradeoff, in the context of numerical discretisations of \eqref{eq:sde1} is the subject of study of this paper.  In particular, we will consider
discretizations based on a Lie-Trotter splitting between the reversible and the nonreversible part of the dynamics. More specifically, we consider
integrators of the form
\begin{equation} \label{eq:Lie}
\widehat{X}^{\Delta t}_{n+1}= \Theta_{\Delta t} \circ \Phi_{\Delta t}(\widehat{X}^{\Delta t}_{n}),
\end{equation}
where $\Phi_{\Delta t}(x)$ is a integrator that approximates the flow map corresponding to the deterministic dynamics
\begin{equation}
\label{eq:non_reversible}
\frac{d x_t}{d t} = \gamma (x_t),
\end{equation}
and $\Theta_{\Delta t}(x)$  which approximates the reversible dynamics
\begin{equation}
\label{eq:reversible}
d x_t = \nabla \log{\pi}(x_{t})dt+\sqrt{2}dW_{t}.
\end{equation}
The choice of $\Phi_{\Delta t},\Theta_{\Delta t}$ has a fundamental influence on the bias, asymptotic variance and stability of the resulting
sampler. In particular, if one chooses  $\Phi_{\Delta t}$  to be a Metropolised integrator \cite{bou2012nonasymptotic} then, similarly to the result in \cite{AVZ15}, the order of convergence of the deterministic integrator $\Phi_{\Delta t}$ provides a lower bound for the difference between expectations with respect to $\widehat{\pi}^{\Delta t}$ and $\pi$.  However,
this is not the case for the numerical asymptotic variance $\widehat{\sigma}_{\Delta t}^{2}(f)$, since even though we  can show that it is a perturbation of $
\sigma^{2}(f)$ the difference will depend crucially on the choice of $\Theta_{\Delta t}$. These results are important as they allow  to choose the
correct combination of dynamics and numerical scheme that drastically reduces the computational cost required to achieve a given tolerance of
error.
\\\\
In summary, the main of the contributions of this paper are
\begin{enumerate}
\item proving  geometric ergodicity for the Markov chain given by \eqref{eq:Lie} for a variety of different numerical integrators applied to the reversible part;
\item a complete characterisation of the asymptotic bias of \eqref{eq:Lie};
\item showing that, by completely characterising the asymptotic variance, numerical integrators of the type \eqref{eq:Lie} inherit the asymptotic variance benefits of the non reversible SDE \eqref{eq:sde1};
\item exhibiting the potential of using nonreversible integrators for sampling as illustrated from a number of different numerical experiments on inference for spatial models as well as real data sets.
\end{enumerate}

The rest of the paper is organised as follows. In Section \ref{sec:set_up} we describe some known theoretical results for the SDE \eqref{eq:sde1} which are necessary for the development of this paper.  In Section \ref{subsec:geo_num} we identifity sufficient conditions to guarantee geometric ergodicity of the Lie-Trotter splitting scheme \eqref{eq:Lie} on $\mathbb{R}^d$.    In Section \ref{sec:set_up1} we study the asymptotic properties of a class of numerical integrators for \eqref{eq:sde1} for which the Lie-Trotter scheme is a special case.  In particular we derive perturbative expansions for the asymptotic bias and variance.  In Section \ref{sec:main} we apply these results to characterise the asymptotic bias and variance of the Lie-Trotter scheme on the bounded domain $\mathbb{T}^d$.   In  Section \ref{sec:gaussian}, we focus on the case where the target distribution is Gaussian and study analytically the trade-off between the asymptotic bias and asymptotic variance in this case. To demonstrate the efficacy of these schemes, in Section \ref{sec:num} we present a number of numerical experiments on inference for spatial models as well as on Bayesian logistic regression.  Proofs of the main results of this paper are deferred to Section \ref{sec:proofs_main} as well as the Appendices.  Finally, a discussion of the results presented in this paper and potential future research directions can be found in Section \ref{sec:discussion}.

\section{Properties of Overdamped Langevin Diffusions} \label{sec:set_up}
In this section we discuss different known theoretical results that are useful for understanding the main results of the paper. We start by listing the assumptions we shall make on $\pi$ and the SDE (\ref{eq:sde1}) to ensure ergodicity.

\begin{assumptions}
\label{ass:ergodicity_sde}
\item[1.] The measure $\pi$ possesses a positive smooth density $\pi(x) > 0$, known up to a normalizing constant, such that $\pi \in L^1(\mathbb{R}^d)$.
\item[2.] The drift vector $b:\mathbb{R}^d\rightarrow \mathbb{R}^d$ of \eqref{eq:sde1} is smooth and satisfies
\begin{equation}
\label{eq:drift_condition}
    b(x) = \nabla \log \pi(x) + \gamma(x),
\end{equation}
where $\gamma:\mathbb{R}^d \rightarrow \mathbb{R}^d$ is a smooth vector field with components in $L^1(\pi)$ such that
\begin{equation}
\label{eq:flow_condition}
    \nabla\cdot\left(\pi(x)\gamma(x)\right) = 0.
\end{equation}
\end{assumptions}

The following result provides necessary and sufficient conditions on the coefficients of (\ref{eq:sde1}) to ensure that  $X_t$ possesses a unique stationary distribution $\pi$.

\begin{prop}\label{prop:main}
Suppose that Assumptions \ref{ass:ergodicity_sde} hold.  Then the diffusion process $X_t$ defined by (\ref{eq:sde1}) possesses a strongly continuous semigroup $(P_t)_{t\geq 0}$ on $L^2(\pi)$ defined by
\begin{equation}
\label{eq:semigroup}
    P_t f(x) = \mathbb{E}[f(X_t)\, | X_0 = x].
\end{equation}
The associated infinitesimal generator is an an extension of
\begin{equation}
\label{eq:generator1}
\cL =  \frac{1}{\pi}\nabla\cdot\left(\pi \nabla\cdot\right) + \gamma\cdot\nabla
\end{equation}
with core $C^\infty_c(\mathbb{R}^d)$.  Moreover, $P_t$ has unique invariant distribution $\pi$.  Conversely, given a diffusion process of the form (\ref{eq:sde1}) which is invariant with respect to $\pi$, then the drift $b$ necessarily satisfies \eqref{eq:drift_condition} and \eqref{eq:flow_condition}.
\end{prop}
\begin{proof}
The first part of this result is a direct application of \cite[Thm 8.1.26]{lorenzi2006analytical}.  The converse implication can be checked using integration by parts.
\end{proof}
% \begin{prop} \label{prop:main}
% Suppose that the coefficient $b$ is smooth.   Then, the process $X_t$ is ergodic with respect to $\pi$ if and only if
% \eqref{eq:drift} holds.
% \end{prop}
% \begin{proof}
% This characterisation follows immediately from the results presented in~\cite[Sec. 4.6]{pavliotis2014stochastic} or \cite[Prop. 3, Prop. 5]{villani2009hypocoercivity}.
% \end{proof}

%  assuming  that $\pi$ possesses a smooth density $\pi(x) > 0$, known up to a normalizing constant, such that $\int_{\mathbb{R}^d} \pi(x)\,dx < \infty$.  The infinitesimal generator corresponding to (\ref{eq:sde1}) is given by an extension of
% \noindent It should be noted that more general forms of the Langevin equation could be considered which will be ergodic with respect to $\pi$. In particular,  one has that the SDE
% \begin{equation} \label{eq:sde2}
% dX_{t}=\left[\Sigma(X_{t})\nabla \log{\pi(X_{t})}+\nabla \cdot \Sigma(X_{t})+\gamma(X_{t})\right]dt+\sqrt{2\Sigma(X_{t})}dW_{t}
% \end{equation}
% will be have unique stationary distribution $\pi$ provided that $\Sigma(x)$ is a smooth positive definite matrix and
% \begin{equation} \label{eq:div_free}
% \nabla \cdot(\gamma(x) \pi(x))=0.
% \end{equation}
\noindent
While many choices for $\gamma$ are possible (see \cite{ma2015complete} for a more complete recipe) a natural family of vector fields is given by $\gamma(x) = J\nabla \Phi(\pi(x))$, where $\Phi$ is a smooth function satisfying $\nabla \Phi(\pi(\cdot)) \in L^1(\pi)$ and $J$ is $d\times d$ skew-symmetric matrix.   We shall focus specifically on the following three choices:
\begin{enumerate}
\item If $\pi$ satisfies $\int_{\mathbb{R}^d}|\nabla \log\pi(x)|\pi(dx)<\infty$, then the vector field
\begin{equation} \label{eq:gamma_choice}
\gamma(x)= J \nabla \log{\pi(x)},\quad J = -J^\top,
\end{equation}
satisfies condition \eqref{eq:flow_condition}.  This was the choice which was studied specifically in \cite{duncan2016variance}.
\item If $\int_{\mathbb{R}^d} |\nabla \log \pi(x)|\pi^{1+\alpha}(dx) < \infty$ for some $\alpha > 0$ then another natural choice for the vector field is given by
\begin{equation}\label{eq:gamma_choice_weighted}
\gamma(x) =  J\nabla \pi^{\alpha}(x),\quad J = -J^\top.
\end{equation}
Although (\ref{eq:gamma_choice_weighted}) introduces an additional tuning parameter $\alpha$, one might prefer this choice as it coincides with the intuition that when far away from the modes the sampler should move towards the modes as quickly as possible, and should only undergo these deterministic meanders in regions of high probability.

\item Let $\Psi:\mathbb{R}\rightarrow \mathbb{R}$ be a smooth, compactly supported function.  Then
\begin{equation}
\label{eq:finitely_supported_gamma}
    \gamma(x) =  J\nabla \log \pi(x) \Psi(\pi(x)),\quad J = -J^\top, \mbox{ and } \beta \in \mathbb{R},
\end{equation}
will always satisfy \eqref{eq:flow_condition}.  Moreover, if $\pi$ has compact level sets, then $\gamma$ will also be compactly supported on $\mathbb{R}^d$.
\end{enumerate}

 Applying the results detailed in \cite{glynn1996liapounov,meyn1993survey}, we shall assume that the process $X_{t}$ possesses a Lyapunov function,  which is sufficient to ensure the exponential ergodicity of $X_{t}$, as detailed in the subsequent proposition.

\begin{assumptions}[Foster--Lyapunov Criterion]
\label{ass:lyapunov}
There exists a function $V:\mathbb{R}^d\rightarrow \mathbb{R}$ and constants $c > 0$ and $b \in \mathbb{R}$ such that
\begin{equation}
\label{eq:lyapunov_condition}
  \mathcal{L} V(x) \leq -c V(x) + b \mathbf{1}_C, \mbox{ and } V(x) \geq 1, \quad x \in \mathbb{R}^d,
\end{equation}
where $\mathbf{1}_C$ is the indicator function over a \emph{petite set}.
\end{assumptions}

 For the definition of a petite set we refer the reader to  \cite{meyn1993stability}.   For the generator $\mathcal{L}$ corresponding to the process (\ref{eq:sde1}) compact sets are always petite.  The exponential ergodicity of $X_t$ follows from the following proposition (see also \cite{mattingly2002ergodicity,meyn1993stability}).
\begin{prop}
Suppose that Assumption \ref{ass:lyapunov} holds,  then there exist constants $C >0$ and  $\lambda > 0$ such that:
\begin{equation}
\left| P_t f(x) - \pi(f) \right| \leq C U(x)e^{-\lambda t}, \quad x \in \mathbb{R}^d,
\end{equation}
for all $f$ satisfying $|f| \leq U$.
\end{prop}

Moreover, the Foster-Lyapunov criterion also provides a sufficient condition for the Poisson equation (\ref{eq:Poisson}) to be well-posed, and thus for the central limit theorem (\ref{eq:CLT}) to hold.

\begin{prop}
\label{prop:CLT}
Suppose that Assumption \ref{ass:lyapunov} holds and that $\pi(U^2) < \infty$, then for any function $f$ such that $|f|\leq U$, the central limit theorem (\ref{eq:CLT}) holds, i.e. $\sqrt{T}(\pi_T(f) - \pi(f))$ converges weakly to a $\mathcal{N}(0, \sigma^2(f))$--distributed random variable, with
$$
    \sigma^2(f) = \int_{\mathbb{R}^d} \phi(x)(-\mathcal{L})\phi(x) \pi(x)\,dx,
$$
where $\phi$ is the unique mean zero solution to the Poisson equation \eqref{eq:Poisson}.  Moreover the solution $\phi$ can be expressed as
$$
    \phi = \int_0^\infty \left[P_t f - \pi(f)\right]\,dt.
$$
\end{prop}
\noindent
The following lemma  provides a sufficient condition on $\pi$  for (\ref{eq:sde1}) to possess a Lyapunov function. It is a slight generalisation of a similar result from \cite{roberts1996exponential}, extended to apply also in the case of nonreversible diffusion processes.

% \begin{assumptions}
% \label{ass:drift_condition}
% {The density $\pi$ is bounded  and} for some $0 < \delta < 1$:
% \begin{equation}
% \label{eq:lyapunov_condition_nonrev}
% \liminf_{|x|\rightarrow \infty} \left((1 - \delta)|\nabla \log \pi(x)|^2 + \Delta \log \pi(x)\right) > 0.
% \end{equation}
% \end{assumptions}

\begin{lemma}{{\cite[Theorem 2  .3]{roberts1996exponential}}}
Consider the process $X_t$ defined by (\ref{eq:sde1}) with drift coefficient $b$ satisfying (\ref{eq:drift_condition}).   Suppose that $\pi$ is bounded, there exists $0 < \delta < 1$ such that,
\begin{equation}
\label{eq:lyapunov_condition_nonrev}
\liminf_{|x|\rightarrow \infty} \left((1 - \delta)|\nabla \log \pi(x)|^2 + \Delta \log \pi(x)\right) > 0,
\end{equation}
and the vector field $\gamma$ satisfies
\begin{equation}
  \label{eq:nonreversible_lyapunov_condition}
    \nabla\cdot \gamma(x) = 0,\quad x \in \mathbb{R}^d,
  \end{equation}
then the Foster--Lyapunov criterion  holds for (\ref{eq:sde1}) with $U(x) = \pi^{-\delta}(x)$ and moreover $\pi(U) < \infty$.
\end{lemma}
\begin{remark}
Note that when $\gamma(x) = J\nabla \Phi(\pi(x))$ equation \eqref{eq:nonreversible_lyapunov_condition} is automatically satisfied. Hence the choices of
 choices of $\gamma$ specified by \eqref{eq:gamma_choice}, \eqref{eq:gamma_choice_weighted} and  \eqref{eq:finitely_supported_gamma} all satisfy \eqref{eq:nonreversible_lyapunov_condition}.
\end{remark}
% \textcolor{red}{[I might remove this proof]}
% \begin{proof}
% For this choice of $\gamma$ the generator of (\ref{eq:sde1}) has the form
% $$
%   \cL = \left(\nabla \log\pi(x) + \gamma(x)\right)\cdot \nabla + \Delta, \quad \beta \in \mathbb{R},
% $$
% For $U(x) = \pi^{-\delta}(x)$ we obtain:
% \begin{align*}
%   \cL U(x) &= -\delta \pi^{-\delta}(x)\left|\nabla \log\pi(x)\right|^2  - \delta\nabla\cdot\left(\pi^{-\delta}(x)\nabla\log\pi(x)\right) \\
%           &= -\delta \left[\left(1 - \delta\right)\norm{\nabla \log\pi (x)}^2 + \Delta \log\pi(x)\right]\pi^{-\delta}(x).
% \end{align*}
% Thus, by assumption (\ref{eq:lyapunov_condition_nonrev}), {there exists $\epsilon > 0$ and $M$ such that for $|x| > M$}:
% $$
%   \left(1 - \delta \right)\norm{\nabla \log\pi (x)}^2 + \Delta \log\pi(x) > \epsilon,
% $$
% and so
% $$
%   \cL U(x) \leq -\delta\epsilon U(x) + b\mathbf{1}_{C_M},
% $$
% where $C_M = \lbrace x \in \mathbb{R}^d \, : \, |x| \leq M \rbrace$ and $b$ is a positive constant.
% \\\\
% Finally, we note that since $\pi$ is bounded, then $U$ is bounded above from zero uniformly.  Thus, $U$ can be rescaled to satisfy the condition $U \geq 1$, as is required by the Foster--Lyapunov criterion.
% \end{proof}

\section{Geometric ergodicity of the splitting scheme on $\mathbb{R}^d$} \label{subsec:geo_num}
In this section we identify sufficient conditions under which the Lie-Trotter scheme on $\mathbb{R}^d$ is geometrically ergodic with respect to an invariant distribution $\widehat{\pi}^{\Delta t}$ which will be a perturbation of $\pi$.   In general, a discretization of the ergodic diffusion process (\ref{eq:sde1}) need not to be ergodic, geometric or otherwise, see \cite{roberts1996exponential}.  For the splitting scheme we shall show that provided  the approximate nonreversible flow $\Phi_{\Delta t}$ is sufficiently weak away from the origin, the process (\ref{eq:Lie}) will inherit the geometric ergodicity from the reversible dynamics.
\\\\
We follow the Meyn and Tweedie \cite{meyn1993stability} recipe to demonstrate geometric ergodicity of $\left(\widehat{X}^{\Delta t}_n\right)_{n \in \mathbb{N}}$.   Consider the reversible process defined by
\begin{equation}
\label{eq:reversible_part}
    Z^{\Delta t}_{n+1} = \Theta_{\Delta t} Z^{\Delta t}_n,
\end{equation}
and $\widetilde{P}_{\Delta t}$ be the corresponding transition semigroup.  We shall assume that the reversible dynamics are a Metropolis-Hastings chain, with proposal kernel $q_{\Delta t}(\cdot | x)$, more specifically, given $x \in \mathbb{R}^d$, $\Theta_{\Delta t}(x)$ is constructed as follows
\begin{enumerate}
    \item Sample $y \sim q_{\Delta t}(\cdot \, | \, x)$.
    \item With probability
                $$
                    \alpha(x,y) = \min\left(1, \frac{\pi(y)q_{\Delta t}(x | y)}{\pi(x)q_{\Delta t}(y  | x)}\right),
                $$
                set $\Theta_{\Delta t}x := y$ otherwise $\Theta_{\Delta t}x := x$.
\end{enumerate}
It is well known that the target distribution $\pi$ is invariant under the map $\Theta_{\Delta t}$ \cite{metropolis1953equation,hastings1970monte}.\\\\
Denote by $\widehat{P}_{\Delta t}(x, \cdot)$ and  $\widetilde{P}_{\Delta t}(x, \cdot)$ the transition distribution functions of the splitting scheme (\ref{eq:Lie}) and (\ref{eq:reversible_part}) respectively.  Then clearly
    $$\widehat{P}_{\Delta t}f(x, A) = (\widetilde{P}_{\Delta t}f)(\Phi_{\Delta t}(x), A), \quad A \in \mathcal{B}(\mathbb{R}^d).$$
Following the approach of \cite{mengersen1996rates} we first show that (\ref{eq:Lie}) is a $\pi$-irreducible, aperiodic Markov chain. Moreover, we will show that all compact sets are small, i.e. for every compact set $C$, there exists a $\delta > 0$ and $n > 0$ such that
$$
    \widehat{P}_{\Delta t}^n(x, \cdot)\geq \delta \nu(\cdot),\quad x\in C.
$$
Finally, we will show that if a Foster-Lyapunov condition holds for the reversible dynamics $\widetilde{P}_{\Delta t}$,  then it also holds for $\widehat{P}_{\Delta t}$.  To this end, we shall make the following assumptions.
%\textcolor{red}{[AD]:  Originally I had intended to follow the approach of \cite{Roberts1998} in which the authors analyse the properties of perturbations of geometrically ergodic Markov chains.   While their approach is very natural,  I am convinced that there are serious issues with their proof.  In particular, the proof of their Lemma 3, which is the crucial part that we require, is completely wrong.  Instead, I shall follow the approach of \cite{mengersen1996rates} which will provide a significantly weaker result,  but I see no way around this.}

\begin{assumptions}
\label{ass:nonreversible}
For $\Delta t$ sufficiently small, we assume that
\item[1] The reversible chain (\ref{eq:reversible_part}) satisfies a Foster-Lyapunov condition, i.e. there exists a continuous function $V \geq 1$, a compact set $C\subset \mathbb{R}^d$ and constants $\lambda \in (0,1)$ and $b \geq 0$ such that
\begin{equation}
\label{eq:foster_rev}
    \widetilde{P}_{\Delta t} V(x) \leq \lambda V(x) + b \mathbf{1}_C(x),   \quad x \in \mathbb{R}^d.
\end{equation}

\item[2] The nonreversible flow map $\Phi_{\Delta t}$ satisfies the following condition,
\begin{equation}
        \lim\sup_{|x|\rightarrow \infty} \frac{V(\Phi_{\Delta t}(x)) - V(x)}{V(x)}< \frac{1}{\lambda}-1.
\end{equation}
\item[3] The preimage $\Phi_{\Delta t}^{-1}(C)$ is bounded.
\end{assumptions}

The main theorem of this section establishes the geometric ergodicity of (\ref{eq:Lie}).
\begin{theorem}
\label{thm:geom_ergodic}
Suppose that Assumptions \ref{ass:nonreversible} hold, and that $\pi$ and $q_{\Delta t}(y | x)$ are positive and continuous for all $x, y \in \mathbb{R}^d$.  Then for $\Delta t$ sufficiently small, the process $\widehat{X}_n^{\Delta t}$ is geometrically ergodic, i.e. there exists $\rho \in (0,1)$ and $K > 0$ such that
$$
    \sup_{|g|\leq V}\left\lvert \int_{\mathbb{R}^d} g(y)\left(P_{\Delta t}(x, y) - \pi(y)\right)\,dy\right\rvert \leq KV(x) \rho^n, \quad n \in \mathbb{N}.
$$
\end{theorem}
\noindent
We now focus on the case when the reversible dynamics are simulated using MALA (Metropolis-Adjusted Langevin Algorithm), i.e. using a proposal of the form
\begin{equation}
\label{eq:mala}
    q_{\Delta t}(\cdot | \, x) \sim \mathcal{N}(x + \nabla \log\pi(x)\Delta t, 2\Delta t),
\end{equation}
for a stepsize $\Delta t>0$.  The following result is an application of Theorem \ref{thm:geom_ergodic} for the proposal (\ref{eq:mala}).

\begin{corollary}[Geometric Ergodicity of Lie-Trotter scheme with MALA dynamics]
\label{cor:mala}
Consider the Lie-Trotter splitting scheme $\widehat{X}_{n}^{\Delta t}$ where the reversible dynamics (\ref{eq:reversible}) are simulated using a MALA scheme with proposal defined by (\ref{eq:mala}).  Suppose that the conditions on  $\pi$ and  $q_{\Delta t}$ specified in \cite[Theorem 4.1]{roberts1996exponential} hold and moreover that
\begin{equation}
\label{eq:mala_condition}
\lim_{|x|\rightarrow \infty} \left(|\Phi_{\Delta t}(x)| - |x|\right) = 0,
\end{equation}
for $\Delta t$ sufficiently small.  Then $\widehat{X}_{n}^{\Delta t}$ is geometrically ergodic.
\\\\
In particular, suppose that $\lim_{|x|\rightarrow \infty}\pi(x)\rightarrow 0$, and that, given $\alpha > 0$, there exist positive constants $\alpha'$, $K_1$ and $K_2$ such that
\begin{equation}
\label{eq:pi_decay_ass}
    \left|\nabla \pi^\alpha(x) \right|\leq K_1 \pi^{\alpha'}(x), \quad \left|\nabla\nabla\pi^\alpha(x)\right|_{max} \leq K_2,\quad x \in \mathbb{R}^d,
\end{equation}
where $|\cdot|_{max}$ denotes the  max norm. If  $\gamma = J\nabla \pi^{\alpha}$ for $J$ antisymmetric, then  condition (\ref{eq:mala_condition}) will hold if  $\Phi_{\Delta t}(x)$ is simulated using an explicit Euler or Runge-Kutta scheme.  A similar result holds for  $\gamma$  given by (\ref{eq:finitely_supported_gamma}).
\end{corollary}

\section{Asymptotic Bias and Variance Estimates for general integrators}  \label{sec:set_up1}
In this section we consider the asymptotic behaviour of the estimator (\ref{eq:ergodic_num}) for $\pi(f)$, obtained for a general numerical scheme $(\widehat{X}_k^{\Delta t})_{k\geq 0}$.   In particular, we shall derive estimates for the asymptotic bias and asymptotic variance of the estimator $\widehat{\pi}_{\Delta t}(f)$.   For simplicity we shall focus on the case where the domain is $\mathbb{T}^d$, i.e. the unit hypercube with periodic boundary conditions.  As in \cite{mattingly2010convergence} this set-up greatly simplifies the derivation of expressions for bias and variance, particularly since remainder terms arising from Taylor expansions can be easily controlled.   We expect that extending these results to unbounded domains should be possible by following analogous approaches in \cite{kopec2014weak}.   Throughout this section, we shall assume that the numerical integrator $\widehat{X}_k^{\Delta t}$ is ergodic, with unique invariant distribution $\widehat{\pi}^{\Delta t}$.

\subsection{Notation}
We first introduce the notation which will be used in this section and the remainder of the paper.  Given a probability measure $\mu$ on $(\mathbb{T}^d, \mathcal{B}(\mathbb{T}^d))$ define  $L^2(\mu)$ to be the Hilbert space of square integrable functions on $\mathbb{T}^d$, equipped with inner product $\langle \cdot, \cdot \rangle_{\mu}$ and norm $\lVert \cdot \rVert_{L^2(\pi)}$.  The subspace $L^2_0(\mu)$ of $L^2(\mu)$ is defined to be
\begin{equation}
\label{eq:L20}
    L^2_0(\mu) = \lbrace f \in L^2(\mu) \, : \, \mu(f) = 0 \rbrace,
\end{equation}
We define $L^\infty(\mu)$ (also denoted $L^\infty(\mathbb{T}^d)$) to be the Banach space of essentially bounded functions on $\mathbb{T}^d$ equipped with norm $\lVert \cdot \rVert_{L^\infty(\mathbb{T}^d)}$.  The subspace $L^\infty_0(\mu)$ of $L^\infty(\mu)$ is defined analogously to (\ref{eq:L20}).  Finally, given a (signed) measure $\nu$ on $(\mathbb{T}^d, \mathcal{B}(\mathbb{T}^d))$ we denote the total variation norm of $\nu$ by $\lVert \nu \rVert_{TV}$.

\subsection{Backward error analysis for ODEs} \label{subsec:back}
Backward error analysis is a powerful tool for the analysis of numerical integrators for differential equations \cite{SaC94,LeR04,HLW06}.
In particular, it is the main ingredient for the proof of the good energy conservation (without drift) of symplectic Runge-Kutta methods when applied to deterministic Hamiltonian systems over exponentially long time intervals \cite{HLW06}.
In our context it is useful to characterize the infinitesimal generator of the  numerical flow  $\Phi_{\Delta t}$ approximating the solution of the ODE (\ref{eq:non_reversible}).  Indeed, given a consistent integrator $z_{n+1}=\Phi_{\Delta t}(z_n)$ for the ODE
\begin{equation} \label{eq:ode}
\frac{dz(t)}{dt}=f(z(t)),
\end{equation} the idea of backward error analysis is to search for a modified differential equation written as a formal series in powers of the stepsize ${\Delta t}$,
\begin{equation} \label{eq:bea}
\frac{d \widetilde z}{dt} = f(\widetilde z) + {\Delta t} f_1(\widetilde z) + {\Delta t}^2f_2(\widetilde z) + \ldots, \quad \widetilde z(0)=z_0
\end{equation}
such that (formally) $z_n=\widetilde z(t_n)$, where $t_n=n{\Delta t}$ (in the above differential equation, we omit the time variable for brevity).   The numerical solution can this be interpreted as a higher order approximation of the exact solution of a modified ODE. For all reasonable integrators, the vector fields $f_j$  can be constructed inductively \cite{LeR04,HLW06}, starting from $f_0=f$. In general, the series in (\ref{eq:bea}) will diverge for nonlinear systems, and thus needs to be truncated.  We thus consider the truncated modified ODE at order $s$
\begin{equation} \label{eq:beatr}
\frac{d \widetilde z}{dt} = f(\widetilde z) + {\Delta t}f_1(\widetilde z) + {\Delta t}^2f_2(\widetilde z) + \ldots + {\Delta t}^sf_{s}(\widetilde z), \quad \widetilde z(0)=z_0
\end{equation}
we have $z_n = \widetilde z(t_n) + \bigo({\Delta t}^{s+1})$ for ${\Delta t}\rightarrow 0$ for bounded times $t_n=n{\Delta t}\leq T$.
We note that the flow $\widetilde\Phi_{\Delta t}(z)$ of the modified differential equation  \eqref{eq:beatr} satisfies
\begin{equation}
\label{flow:bel}
\phi\circ \widetilde\Phi_{\Delta t}=\left({\sum_{k=0}^{M}} \frac{{\Delta t}^k\widetilde{{\mathcal{L}}}_D^k}{k!} \right)\phi{+ \bigo(\Delta t^{M+1})},\qquad \widetilde{\mathcal{L}}_{D}=F_0+{\Delta t}F_1+{\Delta t}^2F_2
+\ldots+{\Delta t}^sF_s,
\end{equation}
for all $M\geq 0$, and smooth test functions $\phi$, and where $F_j\phi=f_j\cdot\nabla\phi,~j=0,\ldots,s$ and $f_0=f$.  Note that the $\bigo({\Delta t}^{M+1})$ terms in (\ref{flow:bel}) are independent of ${\Delta t}\rightarrow 0$ but depend on $M,s$ and $\phi$%
\footnote{{For all ${\Delta t}$ small enough, the sum in \eqref{flow:bel}  can be shown to converge for $M\rightarrow \infty$ in the case
of analytic vector fields $f_j$ (and analytic test functions $\phi$), which permits to remove the $\bigo$ remainder.}}.

\subsection{Asymptotic bias of numerical integrators} \label{subsec:frame}
 The aim of this subsection is to describe the conditions on a numerical integrator for \eqref{eq:sde1} which are sufficient for the numerical invariant distribution $\widehat{\pi}^{\Delta t}$ to approximate $\pi$ to order $r$ in the weak sense.  These conditions relate directly to the expansion of one-step numerical expectations in powers of $\Delta t$. In particular, denote by $\widehat{P}_{\Delta t}$ the transition semigroup associated with $\widehat{X}^{\Delta t}$, i.e.
\[
\widehat{P}_{\Delta t} f:=\IE\left[f(\widehat{X}^{\Delta t}_{1})|X_{0}=x\right].
\]
and assume that the following expansion holds
\begin{equation}
\label{eq:semigroup_expansion}
    \widehat{P}_{\Delta t} f = f + \Delta t A_0f + \ldots + \Delta t^{k}A_{k-1}f +  \Delta t^{k+1}A_{k}f + \Delta t^{q}Q_{f,\Delta t}, \quad q>k+1
\end{equation}
where $A_{i}, i=0,1,\cdots k$ are linear differential operators with coefficients depending smoothly on $\pi(x)$ and its derivatives, as well as on the choice of the numerical integrator. In addition $Q_{f, \Delta t}$ is a smooth remainder term depending both on  $f$ and  $\Delta t$ while being uniformly bounded with respect to $\Delta t$.  The following theorem provides sufficient conditions for expectations with respect to $\widehat{\pi}^{\Delta t}$  to approximate expectations with respect to $\pi$ to order $r$.

\begin{theorem} \label{th:general}
Consider equation \eqref{eq:sde1} solved by an numerical scheme which is ergodic with respect to some probability measure $\widehat{\pi}_{\Delta t}$ and such that
\begin{equation}
\label{eq:bias_condition}
A^{*}_{j}\pi=0, \quad \text{for} \quad j=1, \cdots, r-1,
\end{equation}
where $q > r$, then one obtains
\begin{equation}
\int_{\mathbb{T}^{d}} f(x)\widehat{\pi}^{\Delta t}(dx)=\int_{\mathbb{T}^{d}} f(x)\pi(dx)+\Delta t^{r} \int_{\mathbb{T}^{d}} A_r(-\mathcal{L})^{-1}(f-\pi(f))\pi(dx) +\Delta t^{q}R_{f,\Delta t},
\end{equation}
% \begin{equation}
% \int_{\mathbb{T}^{d}} f(x)\widehat{\pi}^{\Delta t}(dx)=\int_{\mathbb{T}^{d}} f(x)\pi(dx)+\Delta t^{r} \int_{0}^{\infty}\int_{\mathbb{T}^{d}} A_{r}u(x,t)\pi(dx)dt +\Delta t^{q}Q_{f,\Delta t}, \quad q>r
% \end{equation}
where the remainder term $R_{f,\Delta t}$ is uniformly bounded with respect to $\Delta t$, for $\Delta t$ sufficiently small.
\end{theorem}
% $u(x,t)$ solves the backward Kolmogorov equation
% \begin{subequations}\label{eq:kolm_back}
% \begin{eqnarray}
% \frac{\partial u}{ \partial t }  &=& \mathcal{L}u, \\
% u(x,0) &=& f(x).
% \end{eqnarray}
% \end{subequations}
% Equivalently, the coefficient of the $O(\Delta t^r)$ term can be expressed as
% $$
% \int_{0}^{\infty}\int_{\mathbb{T}^{d}} A_{r}u(x,t)\pi(dx)dt = \int_{\mathbb{T}^{d}} A_{r}\psi(x)\pi(dx),
% $$
\begin{proof}
The proof can be found in \cite{AVZ13}.
\end{proof}

\begin{remark}
Integrators $\widehat{X}_n^{\Delta t}$ which have weak error order $r$ will automatically satisfy condition (\ref{eq:bias_condition}) for $j=0, \ldots, r-1$.  However, the converse is not necessarily true, see  \cite{AVZ13} for further discussion.
\end{remark}
An immediate corollary of Theorem (\ref{th:general}) is that, if (\ref{eq:bias_condition}) holds, then for $\Delta t$ sufficiently small, the estimator $\widehat{\pi}_T$ given by (\ref{eq:ergodic_num}) satisfies
$$
    \lim_{N\rightarrow \infty} \widehat{\pi}_{N\Delta t}(f) = \pi(f) + \Delta t^{r} \int_{\mathbb{T}^d} A_r(-\mathcal{L})^{-1}(f-\pi(f)) \pi(dx).
$$

\subsection{Asymptotic variance of numerical integrators} \label{subsec:frame1}

The aim of this subsection is to derive a perturbation expansion in the small timestep regime for the asymptotic variance of an arbitrary ergodic numerical integrator for the dynamics
\eqref{eq:sde1}.  To this end, we consider a diffusion $X_t$ for which the central limit theorem (\ref{eq:CLT}) holds.  Moreover, we shall make the following assumption, which implies that the corresponding numerical scheme $\widehat{X}^{\Delta t}_k$  converges to equilibrium exponentially fast in $L^\infty(\mathbb{T}^d)$, with rate which is uniform with respect to $\Delta t$.

 \begin{assumptions} \label{ass:un_erg}
 There exist constants $C>0$ and $\lambda > 0$ independent of $\Delta t$ such that, for $\Delta t$ sufficiently small,
 $$
	\left\lVert \widehat{P}_{\Delta t}^{k} f  - \widehat{\pi}^{\Delta t}(f)\right\rVert_{L^\infty(\mathbb{T}^d)} \leq C e^{-\lambda k \Delta t}\left\lVert f - \widehat{\pi}^{\Delta t}(f) \right\rVert_{L^{\infty}(\mathbb{T}^d)}, \quad f \in L^{\infty}(\mathbb{T}^d).
$$
\end{assumptions}
\begin{remark}
This condition is nontrivial to verify in general.  For the specific case of the Lie-Trotter integrator (\ref{eq:Lie}), when the reversible component of the dynamics is integrated using MALA, in Theorem \ref{thm:unif_ergodic} we prove that Assumption \ref{ass:un_erg} holds.
\end{remark}
Given an observable $f \in C^\infty(\mathbb{T}^d)$ we consider  $\widehat{\pi}^{\Delta t}_{T}$ as in \eqref{eq:ergodic_num}.  We define the rescaled asymptotic variance of the estimator $\widehat{\pi}^{\Delta t}_{T}$ as follows
\begin{equation} \label{eq:discrete_as_var}
\widehat{\sigma}_{\Delta t}^2(f)=\Delta t \lim_{N\rightarrow \infty} N\mbox{Var}_{\widehat{\pi}^{\Delta t}}\left[\frac{1}{N}\sum_{k=0}^{N-1} f(\widehat{X}^{\Delta t}_k)\right].
\end{equation}
Note here that we  rescale the asymptotic variance  with $\Delta t$, to guarantee a well--defined limit when $\Delta t \rightarrow 0$. Assumption \ref{ass:un_erg} implies that there exists a constant $K > 0$, independent of $\Delta t$ such that
\begin{equation}
\label{eq:discrete_poisson_bound}
    \left\lVert \left[\frac{I - \widehat{P}_{\Delta t}}{\Delta t}\right]^{-1} \right\rVert_{L^\infty_0(\widehat{\pi}_{\Delta t})} < K,
\end{equation}
for $\Delta t$ sufficiently small. In particular, we can express  \eqref{eq:discrete_as_var} as
\begin{equation}
\label{eq:discrete_as_var_generator_form}
\widehat{\sigma}_{\Delta t}^2(f)= 2\Delta t \left \langle \left(f - \widehat{\pi}^{\Delta t}(f)\right), \left({I - \widehat{P}_{\Delta t}}\right)^{-1}\left( f - \widehat{\pi}^{\Delta t}(f)\right)\right\rangle_{\widehat{\pi}^{\Delta t}} - \Delta t\mbox{Var}_{\widehat{\pi}^{\Delta t}}[f].
\end{equation}
It should be clear from (\ref{eq:discrete_as_var_generator_form}) that there will be two contributions to the error between $\widehat{\sigma}_{\Delta t}^2(f)$ and $\sigma^2(f)$:  one arising from the order of weak convergence of the numerical method, and one from the time discreteness of the process $\widehat{X}_k^{\Delta t}$.  Indeed, even when one considers  the exact discrete time dynamics defined by
$$
    X_{n}^{\Delta t} = X(n \Delta t),\quad n \in \mathbb{N},
$$
the error between the corresponding asymptotic variance $\sigma^2_{\Delta t}(f)$ and $\sigma^2(f)$ will be non-zero, despite the fact that both
discrete and continuous time Markov processes have the same invariant distribution.   To isolate the different sources of error, we present first
Proposition \ref{prop:correlation_expansion_discrete} which quantifies the effect of the time-discreteness on the asymptotic variance. In Theorem
\ref{thm:correlation_expansion_discrete}  we then quantify the error between the asymptotic variances $\sigma^2_{\Delta t}(f)$ and $\widehat{\sigma}^2_{\Delta
t}(f)$ of $X_n^{\Delta t}$ and $\widehat{X}_n^{\Delta t}$, respectively.
% $\sigma_{\Delta t}^2(f)$ and $\sigma^{2}(f)$

\begin{prop}
\label{prop:correlation_expansion_discrete}
For all $\phi \in C^\infty(\mathbb{T}^d)$, such that $\pi(\phi) = 0$ there exists a smooth function $R_{\phi}$ such that for $\Delta t$ sufficiently small,
\begin{equation}
\label{eq:generator_expansion_discrete1}
\left(\frac{I - P_{\Delta t}}{\Delta t}\right)^{-1}\phi(x)  =\left(-\mathcal{L}\right)^{-1} \phi(x) + \frac{\Delta t}{2} \left(\frac{I - P_{\Delta t}}{\Delta t}\right)^{-1} (-\mathcal{L})\phi(x) - \frac{\Delta t^2}{6} \left(\frac{I - P_{\Delta t}}{\Delta t}\right)^{-1}(-\mathcal{L})^2\phi  + \Delta t^3 R_\phi,
\end{equation}
where $R_\phi$ is bounded, independent of $\Delta t$.  In particular, for $f \in C^\infty(\mathbb{T}^d)$,
$$
\sigma^2_{\Delta t}(f) = \sigma^2(f) + \frac{\Delta t^2}{6} \left\langle (-\mathcal{L}) \left(f- \pi(f)\right), f-\pi(f)\right\rangle_{\pi} + o(\Delta t^2).
$$
% For $\Delta t$ sufficiently small we have
% \begin{equation}
% \sigma_{\Delta t}^2(f)=\sigma^{2}(f)+3\Delta t(\mbox{Var}_{\pi}[f])+\mathcal{O}(\Delta t^{2})
% \end{equation}
\end{prop}

\begin{proof}
The proof can be found in Section \ref{subsec:as_var_num_proofs}.
\end{proof}
%\st{From the equation above it should be clear that even if the exact solution $X^{\Delta t}_{k}$ is used as the numerical method of choice the difference between $\widehat{\sigma}_{\Delta t}^2(f)$ and $\sigma(f)$ cannot be zero despite the fact that the discrete and the continuous Markov chain processes  have the same invariant distribution. To isolate the effect of the time discreteness  on the autocorrelation, we in fact compare  $\widehat{\sigma}_{\Delta t}^2(f)$ with the asymptotic variance  $\sigma_{\Delta t}^2(f)$ of the discrete exact dynamics}
% \begin{equation}
% \sigma_{\Delta t}^2(f)= 2\Delta t \left \langle \left(f - \pi(f)\right), \left({I - P_{\Delta t}}\right)^{-1}\left( f - \pi(f)\right)\right\rangle_{\pi} - \Delta t\mbox{Var}_{\pi}[f],
% \end{equation}
% where $P_{\Delta t}$ is the transition semigroup for the exact dynamics.
% In particular, we have the following Theorem for the difference between  $\widehat{\sigma}_{\Delta t}^2(f)$ and  $\sigma_{\Delta t}^2(f)$
\noindent
Define the operator $M_{\Delta t}$ to be the projector onto functions with mean zero with respect to $\widehat{\pi}_{\Delta t}$, i.e.
$$
    M_{\Delta t}\phi(x) = \phi(x) - \int_{\mathbb{T}^d} \phi(y)\widehat{\pi}_{\Delta t}(y)\,dy.
$$
The following theorem characterises the difference between the asymptotic variance arising from the exact discrete time dynamics $X^{\Delta t}_n$ and the numerical integrator $\widehat{X}^{\Delta t}_n$.
\begin{theorem}
\label{thm:correlation_expansion_discrete}
Suppose that, for some $k\in \mathbb{N}$, $k\geq 1$, there exist operators $A_{0}, \ldots, A_{k}$ on $C^{\infty}(\mathbb{T}^d)$, bounded uniformly with respect to $\Delta t$, where $A_i = \frac{\mathcal{L}^{i+1}}{(i+1)!}, i=0,\cdots, k-1$ and such that for all $\psi \in C^\infty(\mathbb{T}^d)$ the semigroup $\widehat{P}_{\Delta t}$ satisfies \eqref{eq:semigroup_expansion}.
Suppose that the corresponding invariant distribution $\widehat{\pi}^{\Delta t}$ satisfies
$$
    \int_{\mathbb{T}^d} \psi(x) \widehat{\pi}^{\Delta t}(x)\,dx = \int_{\mathbb{T}^d}\psi(x)\pi(x)\,dx + \Delta t^r R_{\psi},
$$
where $r > k$ and $R_{\psi}$ is a smooth remainder term, uniformly bounded with respect to $\Delta t$.  Moreover, suppose that $\widehat{P}_{\Delta t}$ satisfies  \eqref{eq:discrete_poisson_bound}. Then for all $f , g \in C^\infty(\mathbb{T}^d)$ such that $\pi(f)=\pi(g) = 0$, we have the expansion
\begin{equation}
\label{eq:covariance_expansion1}
\begin{aligned}
\left\langle g, \left(\frac{I - P_{\Delta t}}{\Delta t}\right)^{-1} f \right\rangle_{\pi} &= \left\langle M_{\Delta t}g, \left(\frac{I - \widehat{P}_{\Delta t}}{\Delta t}\right)^{-1} M_{\Delta t} f \right\rangle_{\widehat{\pi}_{\Delta t}} + \Delta t^k R_1(f,g) + o(\Delta t^k),
\end{aligned}
\end{equation}
where
\begin{equation}
\label{eq:covariance_remainder_term}
R_1(f,g) = \left\langle \left(\frac{I - \widehat{P}_{\Delta t}}{\Delta t}\right)^{-1}M_{\Delta t}\left(\frac{\mathcal{L}^{k+1}}{(k+1)!} - A_{k}\right)\left(\frac{I - {P}_{\Delta t}}{\Delta t}\right)^{-1}f, M_{\Delta t}g\right\rangle_{\widehat{\pi}_{\Delta t}}.
\end{equation}
In particular
\begin{equation}
\label{eq:covariance_expansion}
\widehat{\sigma}^2_{\Delta t}(f) = \sigma^2_{\Delta t}(f) + 2\Delta t^k R_1(f,f) + o(\Delta t^{k}).
\end{equation}
If moreover
\begin{equation}
\label{ass:mean_zero_q}
    \pi\left({A}_{k}\psi\right) = 0,
\end{equation}
holds for for all $\psi \in C^\infty(\mathbb{T}^d)$, then we can write
\begin{equation}
\label{eq:improved_remainder_term}
R_1(f,g) = \left\langle \left(-\mathcal{L}\right)^{-1}\left(\frac{\mathcal{L}^{k+1}}{(k+1)!} - A_{k}\right)\left(-\mathcal{L}\right)^{-1}f, g\right\rangle_{{\pi}} + o(\Delta t^k).
\end{equation}
\end{theorem}

\begin{proof}
The proof can be found at Section \ref{subsec:as_var_num_proofs}.
\end{proof}

\begin{remark}
It is interesting to note that contrary to the case of the asymptotic bias in Theorem \ref{th:general}, the order of error for the discrete asymptotic variance in Theorem \ref{thm:correlation_expansion_discrete} depends crucially on the order of the weak convergence of the underlying numerical integrator. Furthermore, we see that if the weak order of the integrator  is than two then the leading order error term between $\widehat{\sigma}^{2}_{\Delta t}(f)$ and the asymptotic variance of the continuous process $\sigma^{2}(f)$ equals to the leading order term of difference between $\sigma^{2}_{\Delta t}(f)$ and $\sigma^{2}(f)$.
\end{remark}

To complete this analysis we shall consider the asymptotic variance arising from a perturbed diffusion process $\widetilde{X}_t$ having infinitesimal generator $\widetilde{\mathcal{L}}_{\Delta t}$ such that, for $\Delta t$ sufficiently small
\begin{equation}
\label{eq:generator_perturbed}
    \widetilde{\mathcal{L}}_{\Delta t}f = {\mathcal{L}}f + \Delta t^{k}\mathcal{L}_{k}f + \Delta t^{q-1} R_{f}, \quad f \in C^\infty(\mathbb{T}^d),
\end{equation}
where $q > k+1$.  We shall also assume that $(\widetilde{L}_{\Delta t})^{-1}$ is bounded in $L^\infty_0(\widehat{\pi}_{\Delta t})$ uniformly with respect to $\Delta t$.  More specifically there exists $K > 0$, independent of $\Delta t$ such that
\begin{equation}
\label{eq:continuous_poisson_bound}
    \left\lVert \left(-\widetilde{\mathcal{L}}_{\Delta t}\right)^{-1} \right\rVert_{L^\infty_0(\widehat{\pi}_{\Delta t})} < K,
\end{equation}
for $\Delta t$ sufficiently small.  The following result characterises the influence of this perturbation on the asymptotic variance for small $\Delta t$.   For numerical approximations of $X_t$ for which a modified SDE \cite{Zyg11} is known, the following result combined with Proposition \ref{prop:correlation_expansion_discrete} provide a convenient means of obtaining an expression for the asymptotic variance $\widetilde{\sigma}^2_{\Delta t}$ of the numerical scheme in terms of $\sigma^2(f)$.

\begin{prop}
\label{prop:continuous_perturbation}
Consider a diffusion process $\widetilde{X}_t$ on $\mathbb{T}^d$ with smooth coefficients and generator $\widetilde{\mathcal{L}}_{\Delta t}$ which satisfies (\ref{eq:generator_perturbed}) and (\ref{eq:continuous_poisson_bound}).  Suppose that $\widetilde{X}_t$ has unique invariant distribution $\widehat{\pi}_{\Delta t}$ which satisfies
\begin{equation}
    \int \psi(x)\widehat{\pi}_{\Delta t}(x)\,dx = \int \psi(x)\pi(x)\,dx + \Delta t^r R_{\psi},
\end{equation}
where $r > k$, and $R_{\psi}$ is a smooth remainder term, uniformly bounded with respect to $\Delta t$.  Then for all $f \in C^\infty(\mathbb{T}^d)$ with $\pi(f) = 0$:
\begin{equation}
\label{eq:covariance_expansion2}
\widetilde{\sigma}^2_{\Delta t}(f) = \sigma^2_{\Delta t}(f) + 2\Delta t^k R_f + o(\Delta t^{k}).
\end{equation}
where
\begin{equation}
\label{eq:covariance_remainder_term2}
R_f = \left\langle \left(-\widetilde{\mathcal{L}}_{\Delta t}\right)^{-1}M_{\Delta t}(-\mathcal{L}_k)\left(-{\mathcal{L}}\right)^{-1}f, M_{\Delta t}f\right\rangle_{\widehat{\pi}_{\Delta t}}.
\end{equation}
If moreover
\begin{equation}
\label{ass:mean_zero_q2}
    \pi\left(\mathcal{L}_k\psi\right) = 0,
\end{equation}
holds for for all $\psi \in C^\infty(\mathbb{T}^d)$, then we can write
\begin{equation}
\label{eq:improved_remainder_term2}
R_f = \left\langle \left(-\mathcal{L}\right)^{-1}(-\mathcal{L}_k)\left(-\mathcal{L}\right)^{-1}f, f\right\rangle_{{\pi}} + o(\Delta t^k).
\end{equation}
\end{prop}

\noindent
The result follows from an argument similar to that of Theorem \ref{thm:correlation_expansion_discrete}.

\section{Asymptotic Bias and Variance Estimates for the splitting scheme} \label{sec:main}
In this section we derive asymptotic bias and variance estimates for the Lie-Trotter splitting scheme (\ref{eq:Lie}) on $\mathbb{T}^d$ by applying the general results derived in Section \ref{sec:set_up1}.  In Section \ref{subsec:bias} we apply Theorem \ref{th:general} to obtain an asymptotic bias estimate for the splitting scheme, while in Section \ref{subsec:variance} we obtain estimates for the asymptotic variance, in the particular case where a MALA scheme is ued to integrate the reversible part of the dynamics.

\subsection{Asymptotic bias of the splitting scheme} \label{subsec:bias}
We now consider the Lie-Trotter scheme (\ref{eq:Lie}) on $\mathbb{T}^d$.  In this section we obtain estimates for the asymptotic bias of the scheme by applying Theorem \ref{th:general}.

\begin{theorem}
\label{thm:bias}
Suppose that the integrator $\Theta_{\Delta t}$ used for the reversible dynamics is invariant with respect to $\pi$ and that  that the deterministic flow ${\Phi}_{\Delta t}$ satisfies a modified backward equation of the form (\ref{eq:bea}) where the vector fields $f_j$ satisfy
\begin{equation}
\label{eq:vector_field_invariance}
	\nabla\cdot\left(f_j(x) \pi(x)\right) = 0,\quad j =1,\ldots, r-1.
\end{equation}
Then, assuming ergodicity, the Lie-Trotter splitting \eqref{eq:Lie} has order $r$ of accuracy for the invariant measure.  More precisely, for all $\phi \in C^2(\mathbb{T}^d)$ and $\Delta t$ sufficiently small
\begin{equation}
\label{eq:bias}
	\int_{\mathbb{T}^{d}} \phi(x) \widehat{\pi}^{\Delta t}(dx)  = \int_{\mathbb{T}^{d}} \phi(x) \pi(dx) + {\Delta t}^r C_{r, \phi} + {\Delta t}^{r+1}R_{\phi, \Delta t},
\end{equation}
where $C_{r,\phi}$ and $R_{\phi, \Delta t}$ are uniformly bounded and
$$
    C_{r,\phi} = \left\langle f_r, (-\mathcal{L})^{-1}(\phi - \pi(\phi))\right\rangle_{\pi}.
$$
\end{theorem}

\begin{remark}
From standard elliptic energy estimates, the remainder term $C_{r,\phi}$ in \eqref{eq:bias} satisfies the a priori bound
$$
|C_{r,\phi}| \leq 2\lVert f_r \rVert_{L^2(\pi)}\lVert \phi \rVert_{L^2(\pi)}.
$$
\end{remark}
\noindent
Theorem \ref{thm:bias} follows from a direct application of Theorem \ref{th:general} and is proved in Section \ref{subsec:bias_proofs}.  Suppose that the nonreversible dynamics is determined by \eqref{eq:non_reversible} where $\gamma(x) = \beta \widetilde{\gamma}(x)$, for $\beta \in \mathbb{R}$ and for some smooth vector field $\widetilde{\gamma}$.  If $\Psi_{\Delta t}$ is an integrator for the flow with error order $r$, then it is straightforward to show that $\Psi_{\Delta t}$ will satisfy a modified backward equation of the form (\ref{eq:bea}) where the vector fields $f_j$ satisfy the scaling $f_j = |\beta|^{j+1} \widetilde{f}_j$, with $\lVert \widetilde{f}_j\rVert_{L^2(\pi)} \sim O(1)$ for $j=0,\ldots, r-1$.  It follows that if the conditions of Theorem \ref{thm:bias} hold, then the leading order term of the bias is of the form $C\Delta t^r |\beta|^{r+1}$, where $C$ is independent of $\Delta t$ and $\beta$.   This estimate provides a rule of thumb for choosing the magnitude of the nonreversible perturbation $\beta$.  Clearly, this should be as large as is possible while maintaining a given tolerance $\epsilon$ for the bias.   To this end, for $\Delta t \ll 1$,  $\beta$ must satisfy
$$
    |\beta| \asymp \epsilon^{\frac{1}{r+1}}\Delta t^{-\frac{r}{r+1}}.
$$
In particular, assuming that $|\beta| \asymp \Delta t^{-\kappa}$ where $\kappa \in \mathbb{R}$, we obtain an upper bound
\begin{equation}
\label{eq:scaling_heuristic}
    \kappa \leq -\frac{1}{r+1}\frac{\log \,\epsilon}{\log\,\Delta t} + \frac{r}{r+1}.
\end{equation}
For $\epsilon \asymp \Delta t$, this rule suggests that $\beta$ should have been chosen to be $O(1)$ with respect to $\Delta t$ if a first order integrator is used to simulate the nonreversible dynamics.  Employing a higher order integrator however, permits larger values of $|\beta|$, in particular $|\beta| \asymp \Delta t^{-0.6}$ for a fourth order scheme as considered in the examples of Section \ref{sec:num}.   We emphasise that unless we have explicit control on the growth of the remainder term in \eqref{flow:bel} as a function of $\beta$, then (\ref{eq:scaling_heuristic}) is only heuristic.  Moreover, we are assuming that the integrator $\Psi_{\Delta t}$ is stable for this parameter regime.  In practice, the stiffness of the ODE (\ref{eq:non_reversible}) would impose additional constraints on $\beta$.

\subsection{Asymptotic variance of the splitting scheme} \label{subsec:variance}
Contrary to Theorem \ref{thm:bias} we shall focus on the case of MALA  for the integrator $\Theta_{\Delta t}$  for which we are able to verify that Assumption \ref{ass:un_erg} holds.
%We have no reason to believe that this Assumption will not hold for other integrators based on RWM and we defer this to further study.
% As we show in the Appendix that  $\Theta_{\Delta t}$ in the case of MALA can be expanded as
% $$
%   \Theta_{\Delta t} \phi  =\phi +\Delta t \mathcal{G}_1\phi + \Delta t^2\mathcal{G}_2\psi + \Delta t^{5/2} R_2(\phi),\quad \phi \in C^\infty(\mathbb{T}^d),
% $$
% where $\mathcal{G}_1$ is the reversible part of $\mathcal{L}$ and $R_1(\phi)\in C^\infty(\mathbb{T}^d)$ is independent of $\Delta t$.
As before we shall assume that the integrator  $\Phi_{\Delta t}$ for the nonreversible flow  satisfies the following expansion
$$
    \Phi_{\Delta t} \phi  =\phi +\Delta t\mathcal{A}_1\phi + \Delta t^2\mathcal{A}_2\phi + \Delta t^3 R_{\phi},\quad \phi \in C^\infty(\mathbb{T}^d),
$$
where  $\mathcal{A}_1 = \gamma(x)\cdot\nabla$ is the antisymmetric part of $\mathcal{L}$ in $L^2(\pi)$ and  $R_{\phi} \in C^\infty(\mathbb{T}^d)$ is bounded independently of $\Delta t$.  We make the following assumption.
\begin{assumptions}
\label{ass:uniform_lipschitz}
The numerical flow $\Phi_h$ is a consistent scheme for (\ref{eq:Lie}) and that there exists $\Delta t_0> 0$ and  $L > 0$ independent of $\Delta t$ such that
\begin{equation}
\label{eq:lipschitz_constant}
    \left\lvert \Phi_{\Delta t}(z_1) - \Phi_{\Delta t}(z_2) \right\rvert \leq L \left\lvert z_1 - z_2 \right\rvert, \qquad z_1, z_2 \in \mathbb{T}^d,
\end{equation}
for all  $\Delta t < \Delta t_0$.
\end{assumptions}

 Provided that \eqref{eq:lipschitz_constant} holds, Theorem \ref{prop:mala_expansion} in the Appendix implies that the reversible integrator $\Theta_{\Delta t}$ satisfies the following perturbation expansion
\begin{equation}
\Theta_{\Delta t} \phi = \phi + \Delta t \mathcal{G}_1 \phi + \Delta t^2 \mathcal{G}_2 \phi + \Delta t^{5/2}R_\phi,\quad \phi \in C^\infty(\mathbb{T}^d),
\end{equation}
where $\mathcal{G}_1 = \mathcal{S}$ is the symmetric part of $\mathcal{L}$ in $L^2(\pi)$ and $\mathcal{G}_2$ is given by (\ref{eq:G2}), and $R_{\phi}$ is a smooth remainder term bounded independently with respect to $\Delta t$.  The following theorem then characterises the asymptotic variance of the Lie-Trotter splitting scheme (\ref{eq:Lie}) when the reversible dynamics are integrated with MALA.  It is a direct application of Theorem \ref{thm:correlation_expansion_discrete} and is proved in Section \ref{subsec:variance_proofs}.

\begin{theorem}
\label{thm:variance}
Consider the Lie-Trotter splitting scheme defined by (\ref{eq:Lie}) where $\Theta_{\Delta t}$ is integrated using MALA and suppose that the nonreversible dynamics preserves the invariant distribution up to order $2$ and satisfies Assumption \ref{ass:uniform_lipschitz}.    Then for all $f \in C^\infty(\mathbb{T}^d)$ we have
\begin{equation*}
%\label{eq:asympt_var_expansion}
    \widehat{\sigma}^2_{\Delta t}(f) = \sigma^2(f) + \Delta t \left\langle (-\mathcal{L})^{-1}(\mathcal{L}^2 - 2\left(\mathcal{A}_2 + \mathcal{G}_1\mathcal{A}_1 + \mathcal{G}_2 \right)(-\mathcal{L})^{-1}(f - \pi(f), f - \pi(f)\right\rangle_{\pi} + o(\Delta t).
\end{equation*}
If moreover, the nonreversible dynamics is integrated using a second order scheme then the $O(\Delta t)$ term can be written as
$$
\left\langle (-\mathcal{L})^{-1}\left( (\mathcal{S}^2 - 2\mathcal{G}_2) + [\mathcal{S}, \mathcal{A}] \right)(-\mathcal{L})^{-1}(f - \pi(f), f - \pi(f)\right\rangle_{\pi},
$$
where $\mathcal{S}$ and $\mathcal{A}$ are the symmetric and antisymmetric parts of $\mathcal{L}$ in $L^2(\pi)$, respectively.
\end{theorem}

From the point of view of tuning the nonreversible Langevin sampler defined by (\ref{eq:Lie}) the main conclusion of Theorem \ref{thm:variance} is that, for $\Delta t$ sufficiently small, the asymptotic variance of (\ref{eq:Lie}) is, to leading order, equal to the asymptotic varaince of the exact dynamics (\ref{eq:sde1}).  In particular, given an observable $f$, this result implies that a  choice of flow $\gamma$ which reduces the variance of a sampler based on (\ref{eq:sde1}) will have a similarly beneficial effect on (\ref{eq:Lie}).  One can thus leverage the theory detailed in \cite{duncan2016variance} and \cite{lelievre2013optimal} to design efficient samplers for a given target distribution $\pi$ and observable $f$.

\section{Gaussian target distributions} \label{sec:gaussian}

In Sections \ref{subsec:bias} and \ref{subsec:variance}, the asymptotic bias and variance for estimators based on Lie-Trotter splitting scheme \eqref{eq:Lie} were characterised in terms of stepsize $\Delta t$ and magnitude of the nonreversible perturbation $\beta$.  This detailed analysis was however restricted to the case of $\mathbb{T}^d$--valued diffusions, as a similar analysis for $\mathbb{R}^d$ would be significantly more involved (see for example \cite{kopec2014weak}).   To demonstrate that analogous expressions for the asymptotic variance and bias can be derived in the $\mathbb{R}^d$ case, in this section we consider the class of linear SDEs given by
\begin{equation} \label{eq:general-linear}
dX_{t}=-AX_{t}dt+ dW_{t}
\end{equation}
where $X_{t} \in \IR^{d}$, $W_{t}$ is a standard $d$-dimensional Brownian motion.
\\\\
In the case where $-A$ is stable the dynamics generated by \eqref{eq:general-linear} are ergodic with respect to $\mathcal{N}(0,\Sigma_{\infty})$ where $\Sigma_{\infty}$ satisfies the Lyapunov equation \cite{GAR85}:
\begin{equation} \label{eq:lyap1}
A\Sigma_{\infty}+\Sigma_{\infty}A^{T}=I .
\end{equation}
We shall consider a vector field $\gamma$  satisfying  (\ref{eq:flow_condition}) which is given by
\[
\gamma(z)=\beta J A z,
\]
where $J$ is a skew symmetric matrix, and $\beta$ is a free parameter. Hence \eqref{eq:sde1} becomes
\begin{equation} \label{eq:sde-gen-linear}
dX_{t}=-(I-\beta J)AX_{t}\,dt + dW_{t}.
\end{equation}
The fact that  equation \eqref{eq:sde-gen-linear} is linear implies that is amenable to very detailed analysis, as for certain classes of numerical schemes, one can find another linear SDE that the numerical method solves exactly in the weak sense. We explain this idea further in Section
\ref{sub:modified}, while in Section \ref{sub:asymp} we extend the formula for the asymptotic variance from ~\cite{duncan2016variance} to linear diffusions with a general positive definite diffusion tensor. This allows the use of the modified equation analysis presented in Section \ref{sub:modified} not just to study the infinite time bias of numerical schemes applied to \eqref{eq:sde-gen-linear}, but also the asymptotic variance.     This is discussed further in Section \ref{sub:ex} in the context of a simple two dimensional example.

\subsection{Exact modified equation}
\label{sub:modified}
Consider a one step method applied to \eqref{eq:sde-gen-linear}
\begin{equation} \label{eq:num-linear}
\widehat{X}^{\Delta t}_{n+1}=B(\Delta t)\widehat{X}^{\Delta t}_{n}+f(\Delta t,\omega), \quad \widehat{X}_0^{\Delta t} = x_0,
\end{equation}
where $f(\Delta t,\omega)$ is the flow map for the noise process and $B(\Delta t) \in \mathbb{R}^{d\times d}$ satisfies $B(0)=I$. For an Euler-Maruyama discretisation of \eqref{eq:sde-gen-linear},
\begin{eqnarray*}
B(\Delta t)&=&I-\Delta t(I-\beta J)A, \\
f(\Delta t,\omega)&=& \sqrt{\Delta t} \xi,
\end{eqnarray*}
where $\xi  \in \mathbb{R}^d$ satisfies $\xi \sim \mathcal{N}(0, I)$. The fact that \eqref{eq:num-linear} remains linear imply that the solution $\widehat{X}^{\Delta t}_{n}$ remains Gaussian at all times, assuming  a deterministic initial condition $x_{0}$. This implies \cite{Zyg11}, that the numerical solution \eqref{eq:num-linear} satisfies exactly in the weak sense at all times the following stochastic differential equation
\begin{equation} \label{eq:mod_exact}
d\widetilde{X}_{t}=\widetilde{B}\widetilde{X}_{t}+\widetilde{\Sigma}^{1/2}dW_{t}
\end{equation}
where $\widetilde{B} \in \mathbb{R}^{d\times d}$ and $\widetilde{\Sigma} \in \mathbb{R}^{d\times d}_{sym}$ are defined by
\begin{subequations} \label{eq:mod_coeff}
\begin{eqnarray}
\widetilde{B} &=& \frac{\log(B(\Delta t))}{\Delta t}, \\
 B(\Delta t)\widetilde{\Sigma}B(\Delta t)^{T}-\widetilde{\Sigma}&=&\widetilde{B}L+L\widetilde{B}^{T},
\end{eqnarray}
\end{subequations}
where $L=\IE(ff^{T})$. For sufficiently small $\Delta t$ one can show that  \eqref{eq:mod_exact} is ergodic with respect to $\mathcal{N}(0,\widetilde{K})$ where $\widetilde{K}$ satisfies a  Lyapunov equation similar to  \eqref{eq:lyap1}. Thus, by solving this equation we can obtain an expression for the invariant measure that the numerical scheme is ergodic with respect to, and hence have an explicit expression for the asymptotic  bias of the numerical method. We study this further in Section \ref{sub:ex}, in the context of a two dimensional example.

\subsection{Asymptotic variance}
\label{sub:asymp}
By extending the results from ~\cite{duncan2016variance} one can calculate the asymptotic variance for \eqref{eq:mod_exact}. In particular if we consider the SDE \eqref{eq:general-linear} our objective is to derive an explicit expression for the asymptotic variance $\sigma^2(f)$ of
$$
I_t = \frac{1}{t}\int_0^t f(\widetilde{X}_s)\,ds,
$$
where $f$ is a function of the form
$$
	f(x) = x\cdot M x + L \cdot x + K,
$$
for some $M \in \mathbb{R}^{d\times d}_{sym}$, $L \in \mathbb{R}^d$ and $K \in \mathbb{R}$. In particular we have the following proposition
\begin{prop}
\label{prop:a_variance}
Consider the linear diffusion defined by the SDE,
$$
dX_t = -AX_t\,dt + \sigma dW_t,
$$
where  $W_t$ is a $m$-dimensional Brownian motion, $\sigma\in \mathbb{R}^{d\times m}$ such that $\Sigma = \sigma\sigma^\top$ is positive definite and $-A$ is stable.  Then, for
\begin{equation} \label{eq:observable}
f(x) = x\cdot M x + L\cdot x + K,
\end{equation}
the asymptotic variance $\sigma^2(f)$ is given by
$$
	\sigma^2(f) = 2\mbox{Tr}\left[\left(\int_0^\infty e^{-A^{\top}t}Me^{-At}\,dt\right)M_\Sigma\right] +  2L_\Sigma\cdot A^{-1}\Sigma_{\infty} A^{-\top}L_\Sigma,
$$
where $M_\Sigma = \Sigma^{1/2} M\Sigma^{1/2}$ and $L_\Sigma = \sqrt{2}\Sigma^{-1/2}L$.
\end{prop}
\begin{proof}
The proof of this proposition can be found in Appendix \ref{app:gaussian_analysis}.
\end{proof}
\subsection{Example}
\label{sub:ex}
We now consider the linear diffusion (\ref{eq:sde-gen-linear}) where
\[
A=\left(\begin{array}{cc}
\alpha & 0 \\
0 & \alpha
\end{array} \right),
\]
for which we know that the stationary covariance satisfies
\[
\Sigma_{\infty}=\left(\begin{array}{cc}
\frac{1}{2\alpha} & 0 \\
0 & \frac{1}{2\alpha}
\end{array} \right)
\]
We now study the properties of integrators where the $\Phi_{\Delta t}$ and $\Theta_{\Delta t}$ in \eqref{eq:Lie} are given by
\begin{subequations} \label{eq:splitting_linear}
\begin{eqnarray}
\Phi_{\Delta t}(z) &=& \left(I+\Delta t \beta J A + \frac{\Delta t^{2}}{2} (\beta J A)^{2}+\cdots+ \frac{\Delta t^{p}}{p!}(\beta J A)^{p}\right)z \\
\Theta_{\Delta t}(z) &=& e^{-A\Delta t}z+\sigma_{\Delta t}\xi
\end{eqnarray}
\end{subequations}
where
\[
\sigma_{\Delta t}\sigma_{\Delta t}^{T}=\int_{0}^{\Delta t}e^{-A(\Delta t-s)}\Sigma e^{-A^{T}(\Delta t-s)}ds = \frac{1}{2\alpha}\left[1 - e^{-2\alpha \Delta t}\right]I.
\]
More precisely we solve the reversible part of the
dynamics exactly, while we apply a Taylor-based method
of order $p$ to the nonreversible part of the
dynamics. We note here that the exact solution of the
reversible part of the dynamics is only possible
because the dynamics are linear. A further consequence of the linearity of the dynamics is that it is possible to conserve the invariant measure for the reversible part with using the $\theta$ method with $\textstyle\theta=\frac{1}{2}$, see  \cite{AVZ13} . Hence we will also consider the integrator $\widetilde{\Theta}_{\Delta t}(z)$ given by
\begin{equation} \label{eq:splitting_linear1}
\widetilde{\Theta}_{\Delta t}(z)=\left(I+\frac{\Delta t}{2}A\right)^{-1} \left[\left(I-\frac{\Delta t}{2}A\right)z+\sqrt{\Delta t}\xi \right]
\end{equation}
The other interesting feature of \eqref{eq:splitting_linear1} is that even though not exact like \eqref{eq:splitting_linear}, when metropolised, proposals generated from by (\ref{eq:splitting_linear1}) will be accepted almost surely.   For nonlinear problems, the reversible dynamics cannot be integrated exactly, and

it is impossible to construct
an exact solution and  \eqref{eq:splitting_linear1} does not  conserve the invariant measure. Hence one  would replace these integrator with one that
conserves the invariant measure by introducing a Metropolisation step, and Theorem \ref{thm:bias} would still hold.

\paragraph{Study of the invariant measure bias}
We now study the properties of the numerical invariant measure using \eqref{eq:mod_coeff}. We use
Mathematica to symbolically calculate the solutions to \eqref{eq:mod_coeff} and then obtain an expression for the
numerical invariant measure, when a first and a second order numerical method is used to solve the nonreversible part of
the diffusion.  In particular, in Tables \ref{tab:exact_invariant},\ref{tab:theta_invariant} we present exact expressions for the
numerical invariant measure based on the Lie-Trotter splitting \eqref{eq:Lie}, for different ordering of the splitting and
different choices of integrators for the reversible and nonreversible part. Furthermore, in Figure \ref{fig:figconvlinear1} we
plot  the $2$-norm of the difference between the covariance matrix of the numerical method and the true covariance matrix $\Sigma_{\infty}$  when
the nonreversible part is solved first and then the $\theta$-method with $\theta=1/2$ is used for the reversible part\footnote{We have not included
any of the other possible combinations of ordering of splitting and numerical integrators for the reversible part as the results are qualitatively the
same}. As we can see the order of convergence is always odd. This was also observed in \cite{AVZ15} and it relates with the fact that for the
deterministic methods used here, the coefficient $f_{p}$ in Theorem \ref{thm:bias} is always zero  when $p$ is even hence giving the extra order of
convergence observed in Figure \ref{fig:figconvlinear1}. Additionally in Figure \ref{fig:figconvlinear2} we plot the asymptotic bias of $\Delta t$ when a
numerical integrator of order $1$ is used to solve the nonreversible part for different values of $\beta$. As we can see, the larger the value of $
\beta$ the larger the asymptotic bias.

\begin{table}
\begin{center}
\begin{tabular}{ r|c|c| }
\multicolumn{1}{r}{}
 &  \multicolumn{1}{c}{Reversible first}
 & \multicolumn{1}{c}{Non reversible first} \\
\cline{2-3}
 $p=1$ & $\frac{\left(1-e^{-2 \alpha \Delta t}\right) \left(1+\alpha^2 \beta^2 \Delta t^2\right)}{2 \alpha \left[1-e^{-2 \alpha \Delta t} \left(1+\alpha^2 \beta^2 \Delta t^2\right)\right]}$ &  $ \frac{1-e^{-2 \alpha \Delta t}}{2 \alpha \left[1-e^{-2 \alpha \Delta t} \left(1+\alpha^2 \beta^2 \Delta t^2\right)\right]}$  \\
\cline{2-3}
$p=2$ & $\frac{\left(1-e^{-2 \alpha \Delta t}\right) \left(\alpha^4 \beta^4 \Delta t^4+4\right)}{2 \alpha \left(e^{-2 \alpha \Delta t} \left(\alpha^4 \beta^4 \Delta t^4+4\right)-4\right)}$ & $\frac{2 \left(1-e^{-2 \alpha \Delta t}\right)}{\alpha \left[4-e^{-2 \alpha \Delta t} \left(4+\alpha^4 \beta^4 \Delta t^4\right)\right]}$  \\
\cline{2-3}
\end{tabular}
\end{center}
\caption{Numerical invariant measure when the reversible part is solved exactly}
\label{tab:exact_invariant}
\end{table}

\begin{table}
\begin{center}
\begin{tabular}{ r|c|c| }
\multicolumn{1}{r}{}
 &  \multicolumn{1}{c}{Reversible first}
 & \multicolumn{1}{c}{Non reversible first} \\
\cline{2-3}
 $p=1$ & $\frac{4+4 \alpha^2 \beta^2 \Delta t^2}{8 \alpha-4 \alpha^2 \beta^2 \Delta t+4 \alpha^3 \beta^2 \Delta t^2-\alpha^4 \beta^2 \Delta t^3}$ &  $\frac{4}{8 \alpha-4 \alpha^2 \beta^2 \Delta t+4 \alpha^3 \beta^2 \Delta t^2-\alpha^4 \beta^2 \Delta t^3}$  \\
\cline{2-3}
$p=2$ &$\frac{4 \left(4+\alpha^4 \beta^4 \Delta t^4\right)}{\alpha \left(32-\alpha^3 \beta^4 \Delta t^3 (2-\alpha \Delta t)^2\right)}$  & $\frac{16}{\alpha(32-\alpha^{3}\beta^{4}\Delta t^{3}(2-\alpha\Delta t)^{2})} $  \\
\cline{2-3}
\end{tabular}
\end{center}
\caption{Numerical invariant measure when the reversible part is solved by $\theta$-method for $\theta=\frac{1}{2}$.}
\label{tab:theta_invariant}
\end{table}

\begin{figure}[bth]
\centering
\includegraphics[scale=0.5]{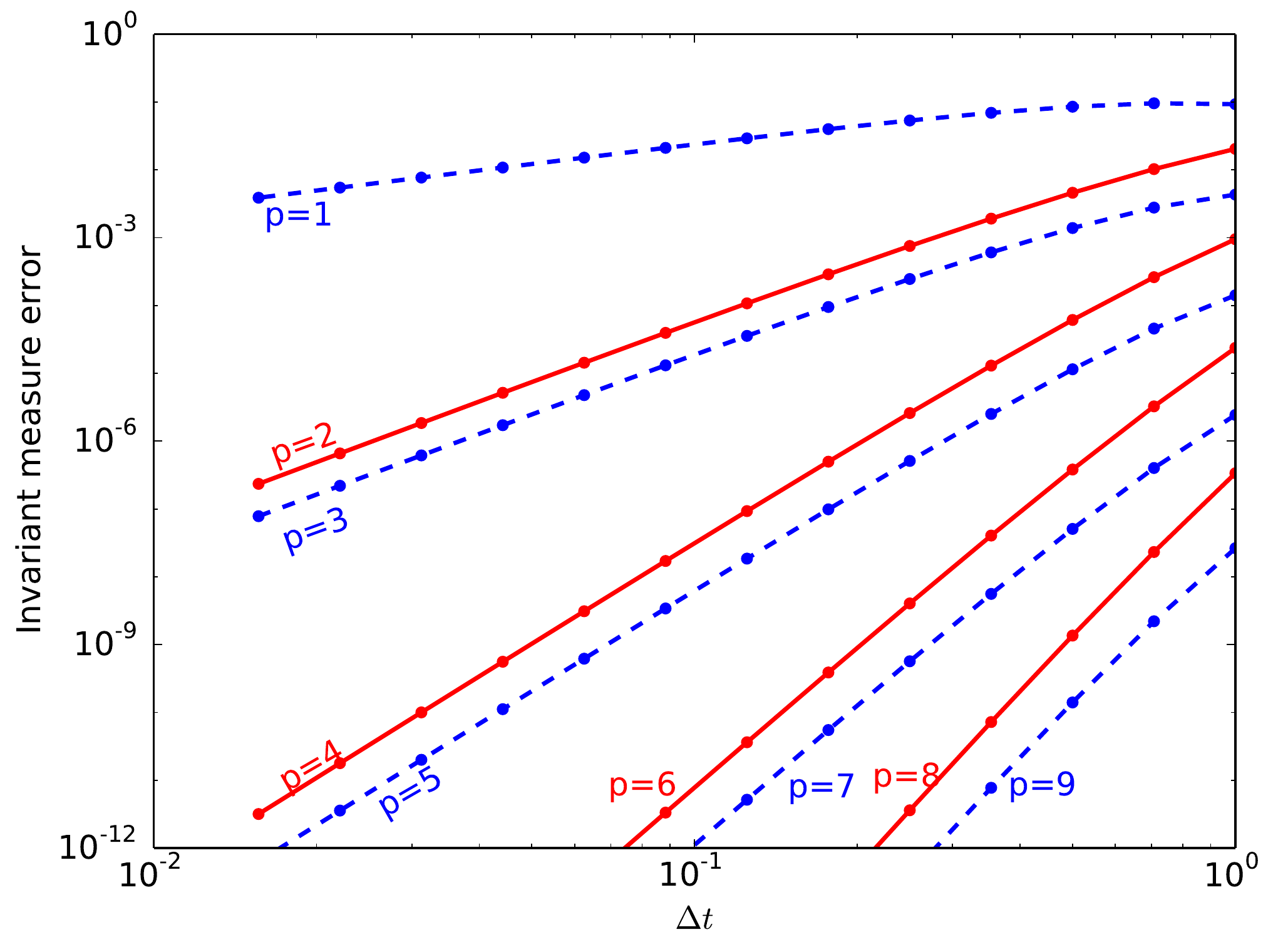}
\caption{Accuracy of the numerical invariant measure (covariance matrix error) of the Lie-Trotter splitting  The lines corresponds to explicit deterministic integrators for the nonreversible part of orders $p=1,2,3,\ldots,9$ (from top to bottom), respectively. The orders of accuracy for the invariant measure are always odd.
\label{fig:figconvlinear1}}
\end{figure}

\begin{figure}[bth]
\centering
\includegraphics[scale=0.5]{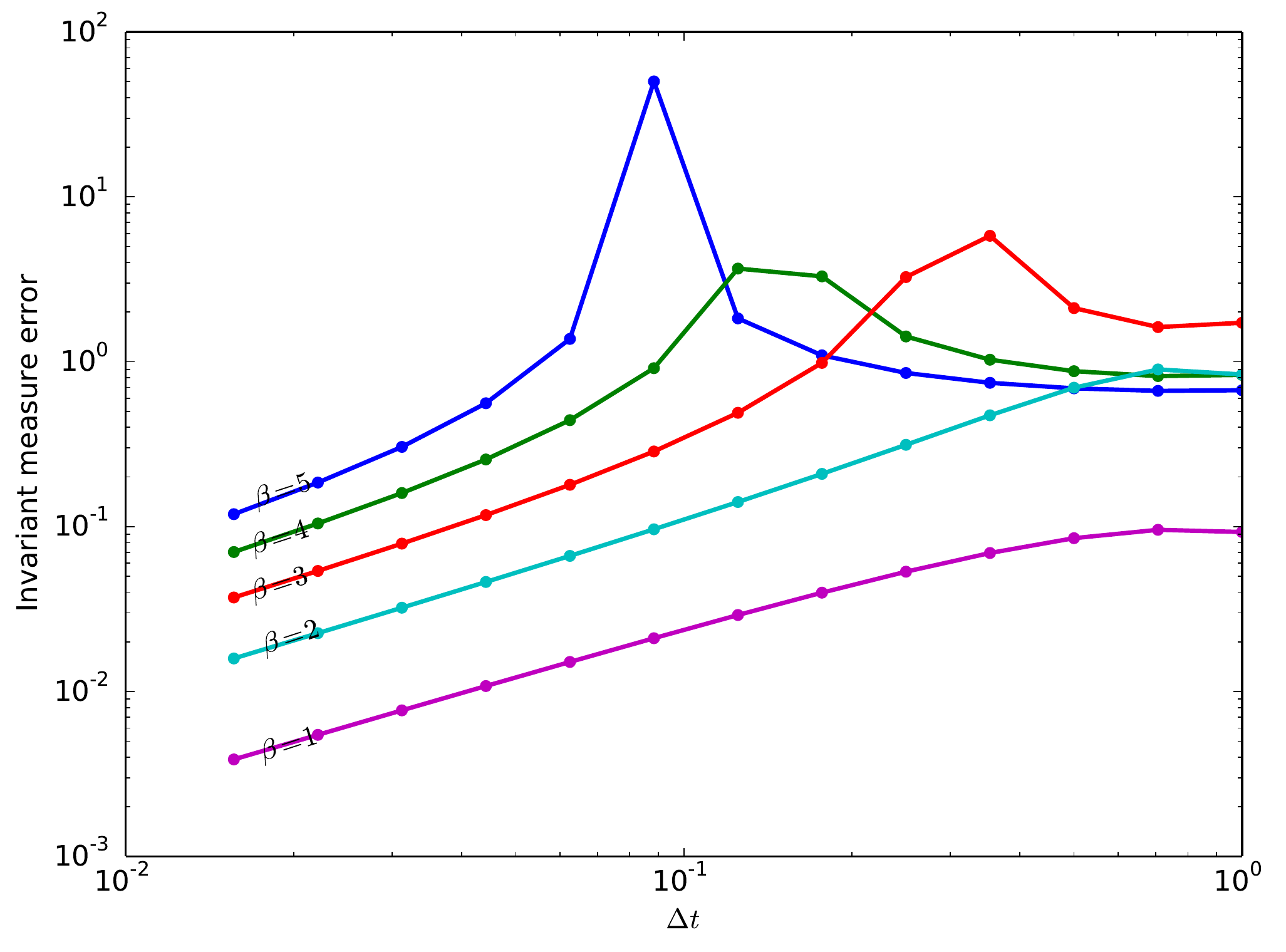}
\caption{Accuracy of the numerical invariant measure (covariance matrix error) of the Lie-Trotter splitting  for different values of $\beta$ when a first order numerical method is used for solving the nonreversible part.
\label{fig:figconvlinear2}}
\end{figure}

\paragraph{Study of the asymptotic variance}
We now study the properties of the asymptotic variance using \eqref{eq:mod_coeff}. In particular the idea is that since our numerical solution
satisfies exactly in the weak sense the corresponding modified equation then it is enough to look at Proposition \ref{prop:a_variance} where $A$ and $\sigma$ are now replaced with the modified coefficients \eqref{eq:mod_coeff}.  Similarly to the case of the invariant measure bias we use
Mathematica  to symbolically calculate the solutions to \eqref{eq:mod_coeff} and then obtain an expression for the
asymptotic variance, when a first and a second order numerical method is used to solve the nonreversible part of
the diffusion. In particular, we take $K=0$,$L=0$ and $M$ the two-by-two identity in \eqref{eq:observable} we find that when the reversible part is solved exactly that  for
$p=1$, we have
\[
\widetilde{\sigma}_{\Delta t}^{2}(f)=\frac{2+\beta^2}{2 \alpha \left(1+\beta^2\right)}+\frac{\left(2 \beta^2+\beta^4+\beta^6\right)\Delta t }{4 \left(1+\beta^2\right)^2}+\mathcal{O}
(\Delta t^{2}),
\]
independently of the ordering  of the splitting, while for $p=2$ we have
\[
\widetilde{\sigma}_{\Delta t}^{2}(f)=\frac{2+\beta^2}{2 \alpha \left(1+\beta^2\right)}-\frac{\alpha \beta^4 \Delta t^2}{6 \left(1+\beta^2\right)^2}+\mathcal{O}(\Delta t^{3}),
\]
again  independently of the ordering of the splitting. The expressions above change to
\[
\widetilde{\sigma}_{\Delta t}^{2}(f)=\frac{2+\beta^2}{2 \alpha \left(1+\beta^2\right)} +\frac{\left(2 \beta^2+\beta^4+\beta^6\right) \Delta t}{4 \left(1+\beta^2\right)^2}+\mathcal{O}(\Delta t^{2}),
\]
for $p=1$, and
\[
\widetilde{\sigma}_{\Delta t}^{2}(f)=\frac{2+\beta^2}{2 \alpha \left(1+\beta^2\right)}-\frac{\alpha(-\beta^{2} +2\beta^4) \Delta t^2}{12 \left(1+\beta^2\right)^2}+ \mathcal{O}(\Delta t^{3}),
\]
when the reversible part of the dynamics is solved by the $\theta$-method for $\theta=1/2$, again independently of the ordering of the splitting.  We note here that these results agree with Proposition \ref{prop:continuous_perturbation}, since for $p=1$ the leading order perturbation in terms of the continuous time variance is $\mathcal{O}(\Delta t)$ while for $p=2$ is $\mathcal{O}(\Delta t^{2})$.

\paragraph{Mean Square Error}
Having obtained analytical expressions for the asymptotic bias of the invariant measure as well as for the asymptotic variance of the corresponding numerical schemes, we combine them in order to study the mean square error.  More precisely, decomposing the MSE into bias and variance,
\[
\IE |\widehat{\pi}_{T}(f) -\pi(f) |^{2}= (\IE\widehat{\pi}_{T}(f)-\pi(f) )^{2}+\IE(\widehat{\pi}_{T}(f)-\IE\widehat{\pi}_{T}(f))=\widehat{\mu}^{2}_{T}+\widehat{\sigma}^{2}_{T},
\]
we approximate $\widehat{\mu}_{T}$ by the invariant measure bias, while on the other hand
\[
\widehat{\sigma}^{2}_{T} \simeq \frac{\widehat{\sigma}^{2}(f)}{T}.
\]
We now plot in Figure \ref{fig:MSE} the MSE when a first and a second order numerical method is used to solve the nonreversible part and the reversible part is solved exactly. In particular, we choose our timestep $\Delta t=10^{-4}, \alpha=1$, $T = 10^3$ and we study the influence of $\beta$ on the MSE. As can be seen in both cases there is a range of values of the parameter $\beta$ for which the MSE is reduced almost to $\frac{1}{2}$ which is the theoretical minimal variance attainable using this choice of dynamics \cite{duncan2016variance}.  Increasing the magnitude of $\beta$ beyond this point, eventually the bias term will dominate the mean-square error which will rapidly increase.   Using a second order integrator for the nonreversible dynamics mitigates this increase in bias, and a significant reduction in MSE is possible for a much wider range of $\beta$.   Indeed, in Figure \ref{fig:MSE} we see that, in this case, the asymptotic bias is $\mathcal{O}(\Delta t^{3})$  and will not dominate the MSE  for a wider range of $\beta$ values.

  \begin{figure}[!ht]
    \subfloat[First order method]{%
      \includegraphics[width=0.45\textwidth]{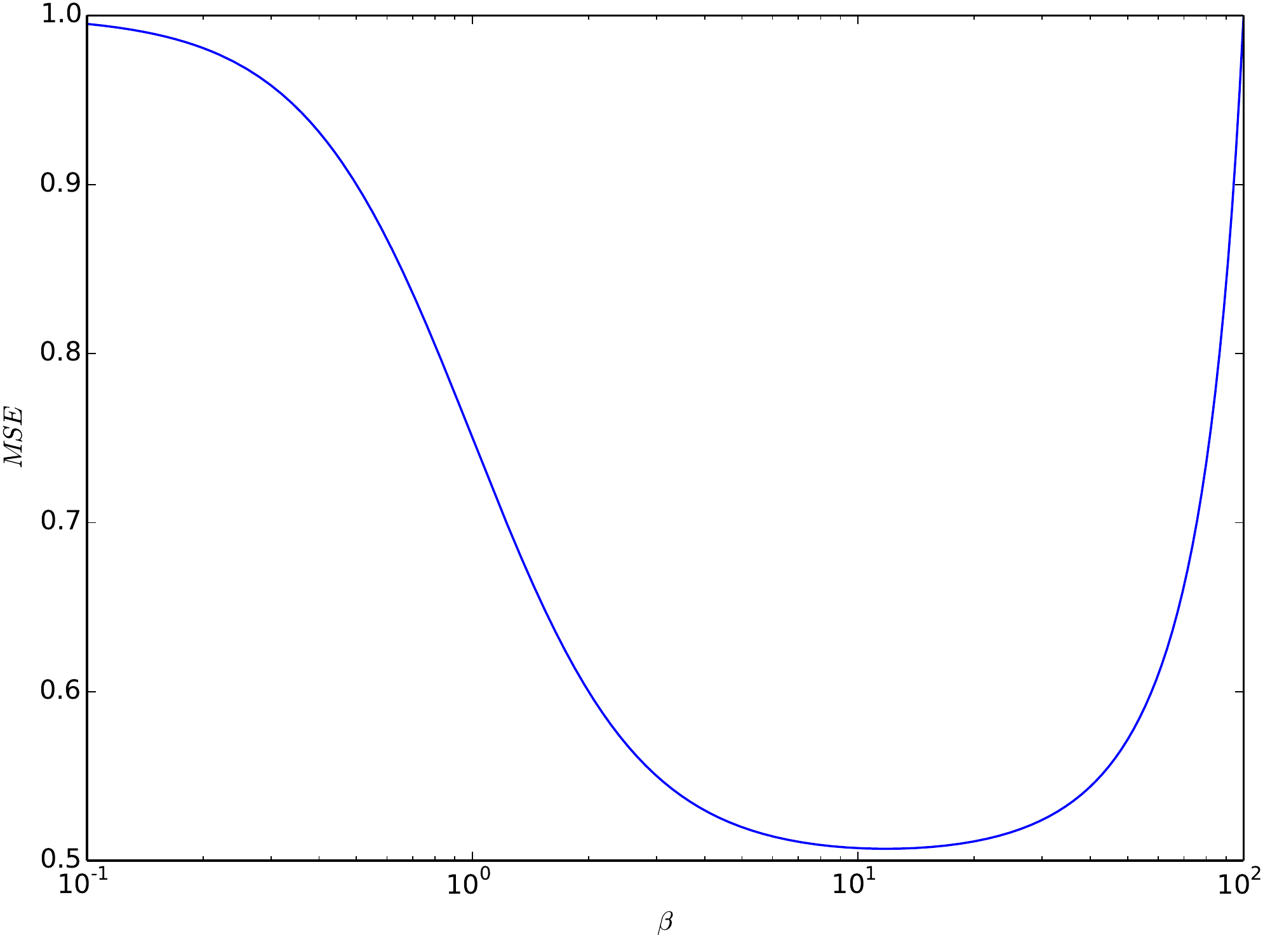}
    }
    \hfill
    \subfloat[Second order method]{%
      \includegraphics[width=0.45\textwidth]{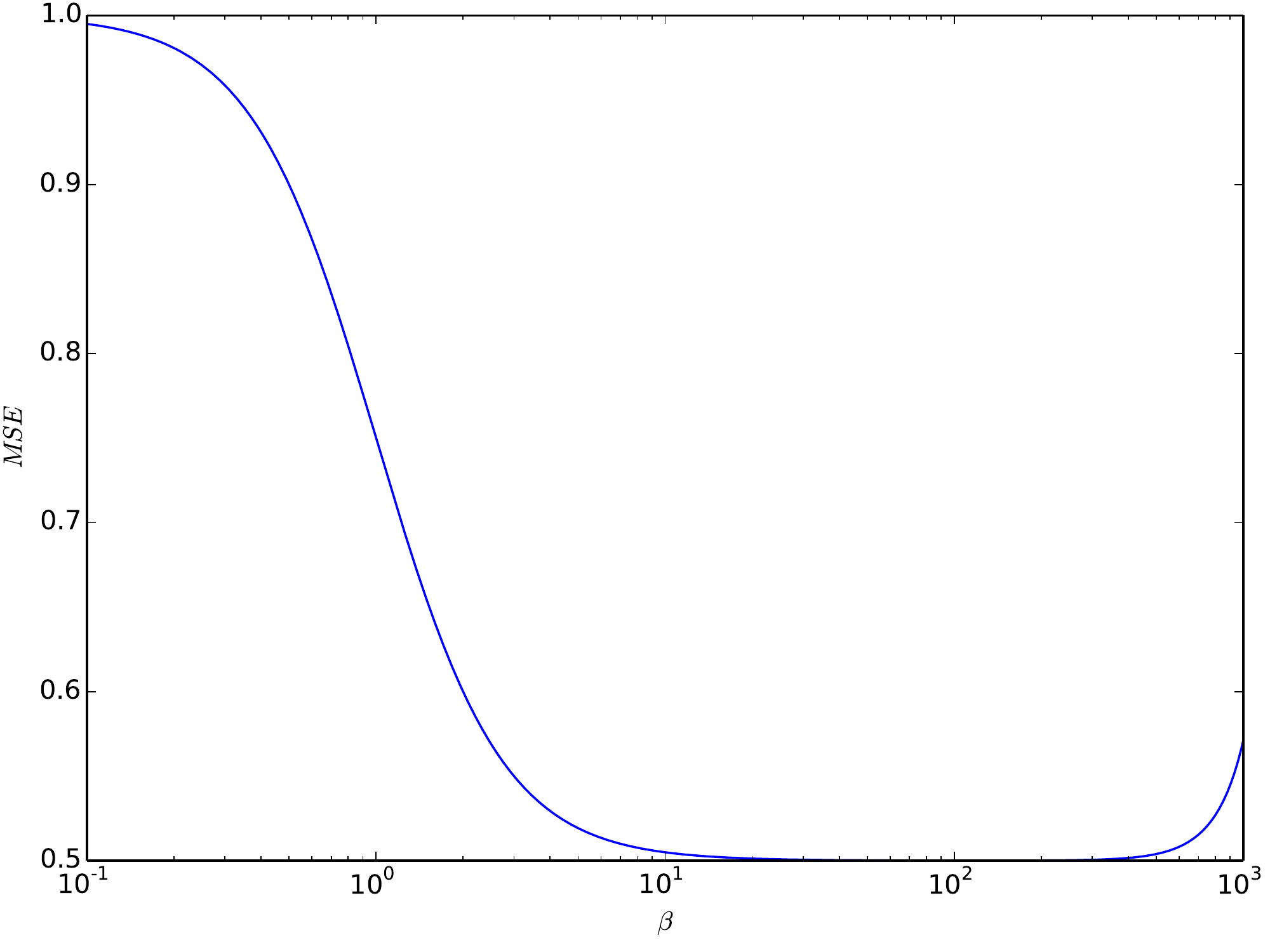}
    }
    \caption{MSE for two different methods applied to the nonreversible part.}
    \label{fig:MSE}
  \end{figure}

\section{Numerical experiments} \label{sec:num}
In this section, we perform a number of different numerical investigations that illustrate the superiority of the nonreversible
Langevin samplers over standard Metropolis-Hastings algorithms for a fixed computational budget.  In particular, we define computational cost here in terms of number of density evaluations which is the dominating cost in high dimensions. To this end we ensure that every comparison is made  for the same computational cost, \emph{i.e.}, same number of density evaluations.

\subsection{Warped Gaussian distribution}
As a first numerical we consider the expectation of an observable with respect to the following two dimensional distribution
\begin{equation} \label{eq:warped}
\pi(x) \propto \exp \left(-\frac{x^{2}_{1}}{100}-(x_{2}+bx^{2}_{1}-100b)^{2}\right)
\end{equation}
where $x=(x_{1},x_{2})$.  The parameter $b>0$ controls the degree of warpedness, and is chosen to be $b = 0.05$.  The log density is plotted in Figure \ref{fig:banana_vanity}a.   Our objective is to estimate $\pi(f)$ where $f(x) = |x|^2$.   The nonreversible flow $\gamma$ is chosen as follows:
\[
\gamma(x) = J\nabla \log \pi(x),\quad J = \left(\begin{matrix} 0 & 1 \\ -1 & 0  \end{matrix}\right)
\]
\begin{figure}[!ht]
    \includegraphics[width=\textwidth]{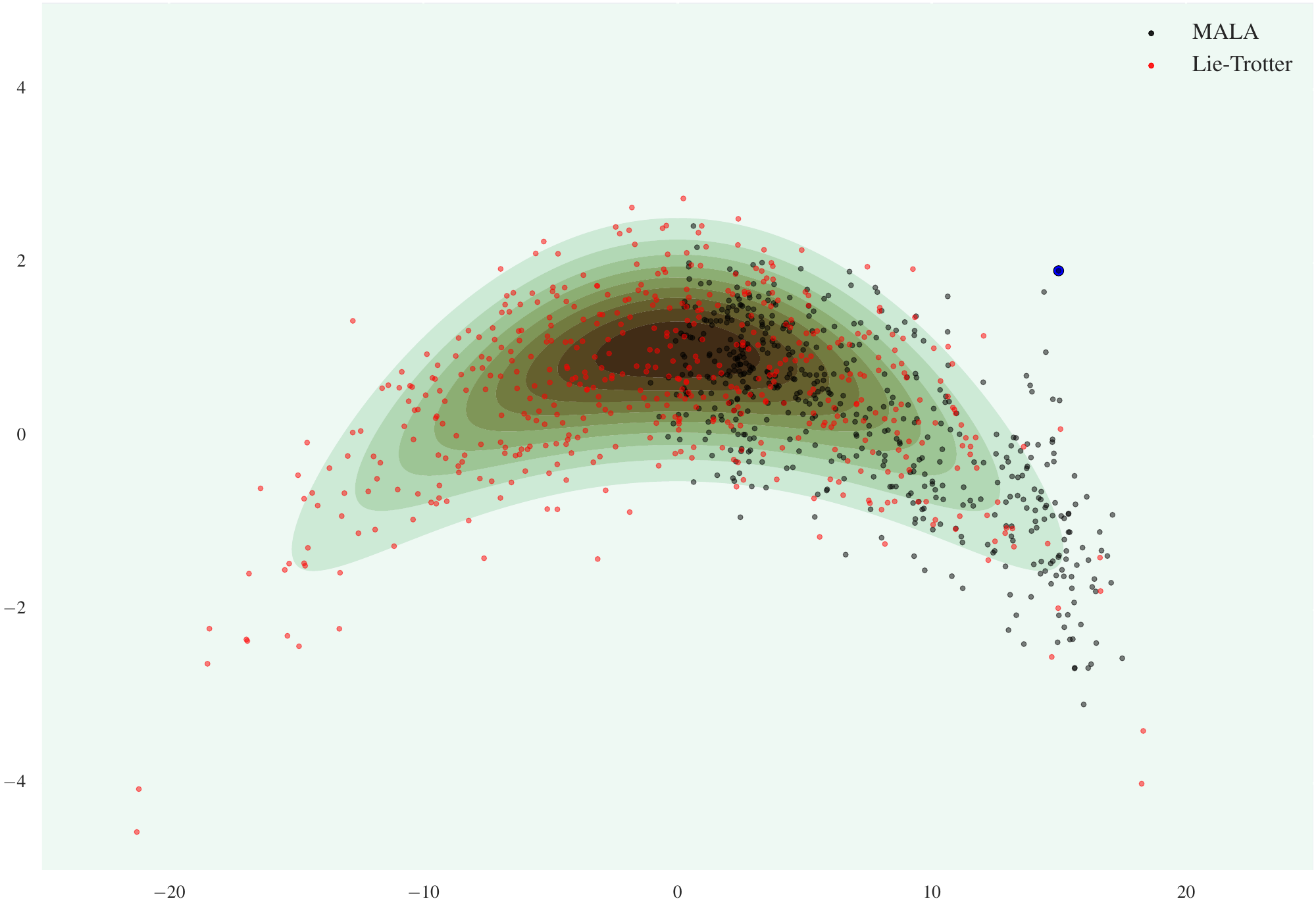}
      \caption{Typical trajectories for MALA and Lie-Trotter splitting scheme applied to the warped Gaussian distribution (\ref{eq:warped}), with computational budget of $3200$ density evaluations.  Both schemes started from $x=(15,2)$ depicted by a blue dot.}
    \label{fig:banana_vanity}
  \end{figure}
In Figure \ref{fig:banana_vanity}, we plot characteristic trajectories of MALA as well its nonreversible counterpart (for $\beta = 25$) starting from the initial point $x = (15,2)$. The figure suggests superior mixing of the nonreversible samplers, which improves further with increasing $\beta$ values.  In Figure \ref{fig:banana_results} the mean-square error is plotted as a function of stepsize for different values of flow strength $\beta$.  The reversible part of the Lie-Trotter scheme is simulated using MALA, RWMH and Barker rule in Figures \ref{fig:banana_results}a,\ref{fig:banana_results}b and \ref{fig:banana_results}c, respectively.   The ``exact'' value of $\pi(f)$ used to compute the MSE is obtained via adaptive Gaussian quadrature, accurate up to $10^{-10}$.    In accordance with the results of Theorems \ref{thm:bias} and \ref{thm:variance}, the MSE is a tradeoff between bias and variance.  For a fixed computational budget as $\Delta t$ decreases, the bias arising from the discretisation of the nonreversible flow decreases.  However, the variance simultaneously increases as the total simulated time $T = N\Delta t$ is reduced.   This competion between bias and variance suggest an optimal choice of timestep $\Delta t$ which minimises the MSE. This tradeoff is further exacerbated  when $\beta$ is increased. Nevertheless, for an appropriate choice of $\beta$ the MSE can be up to an order of magnitude lower than that of MALA, at the same computational cost.

\begin{figure}[!ht]
    \subfloat[MALA for the reversible part]{%
      \includegraphics[width=0.32\textwidth]{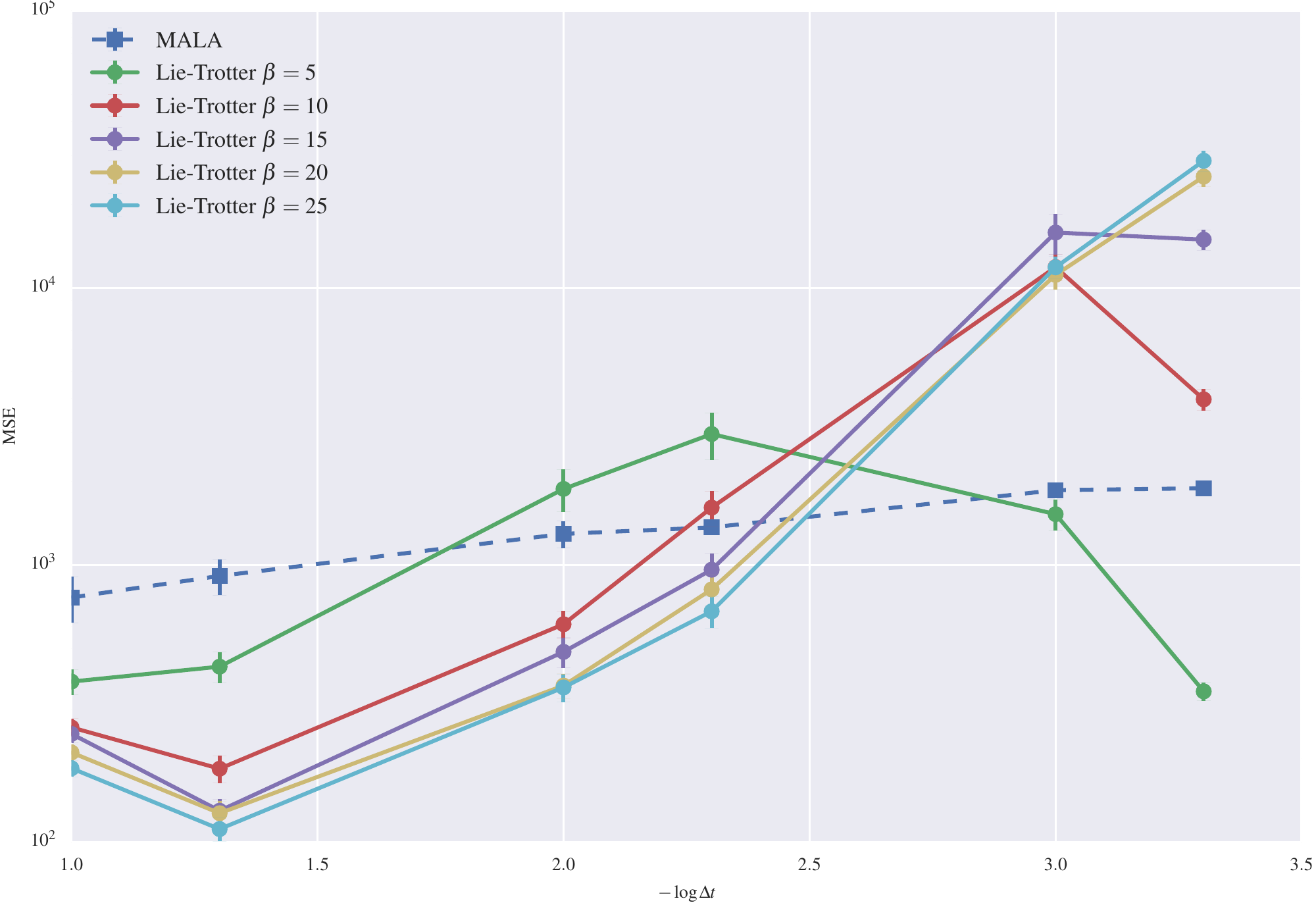}
    }
    \hfill
    \subfloat[RMWH rule for the reversible part]{%
      \includegraphics[width=0.32\textwidth]{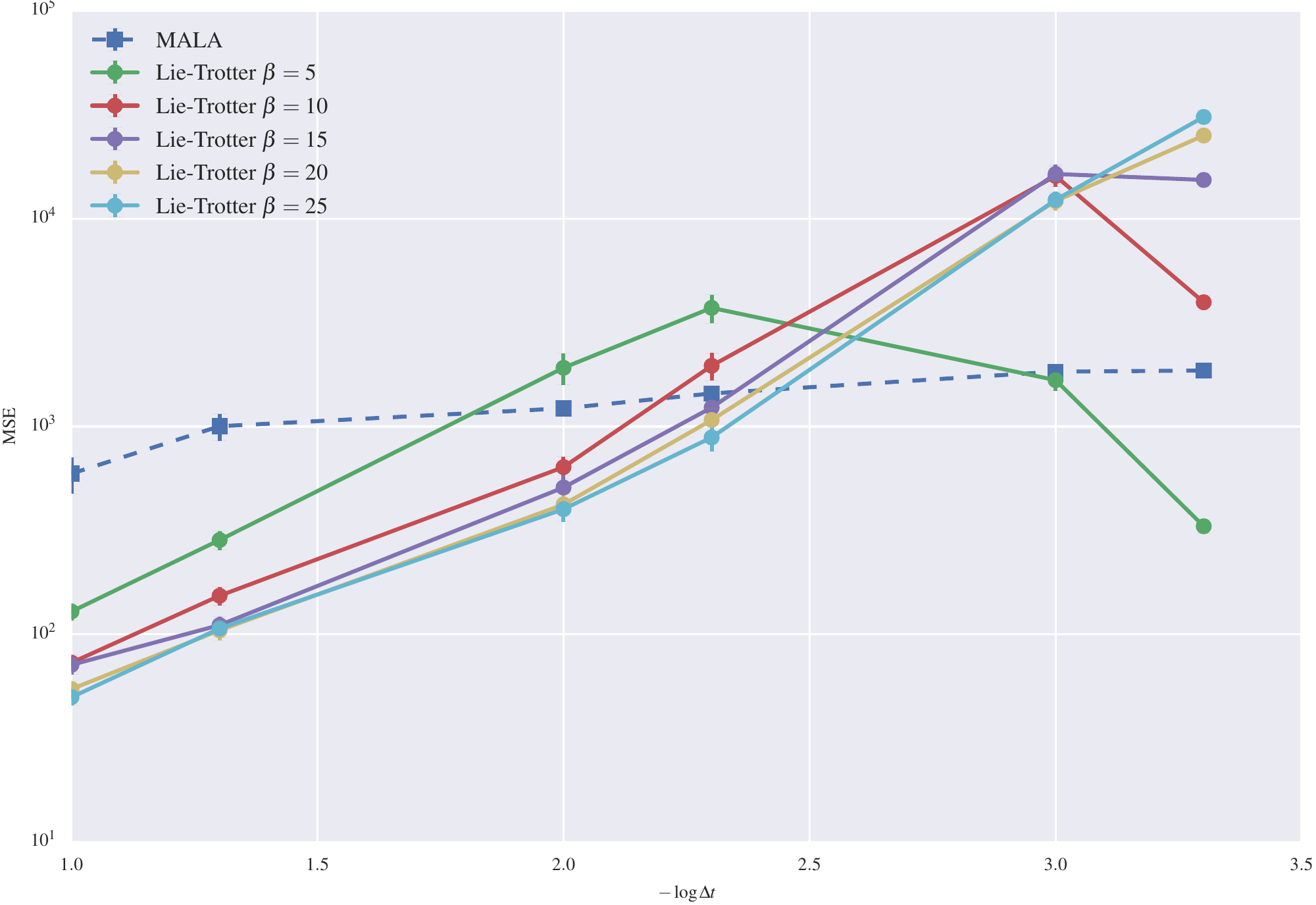}
    }
     \hfill
    \subfloat[Barker Scheme for the reversible part]{%
      \includegraphics[width=0.32\textwidth]{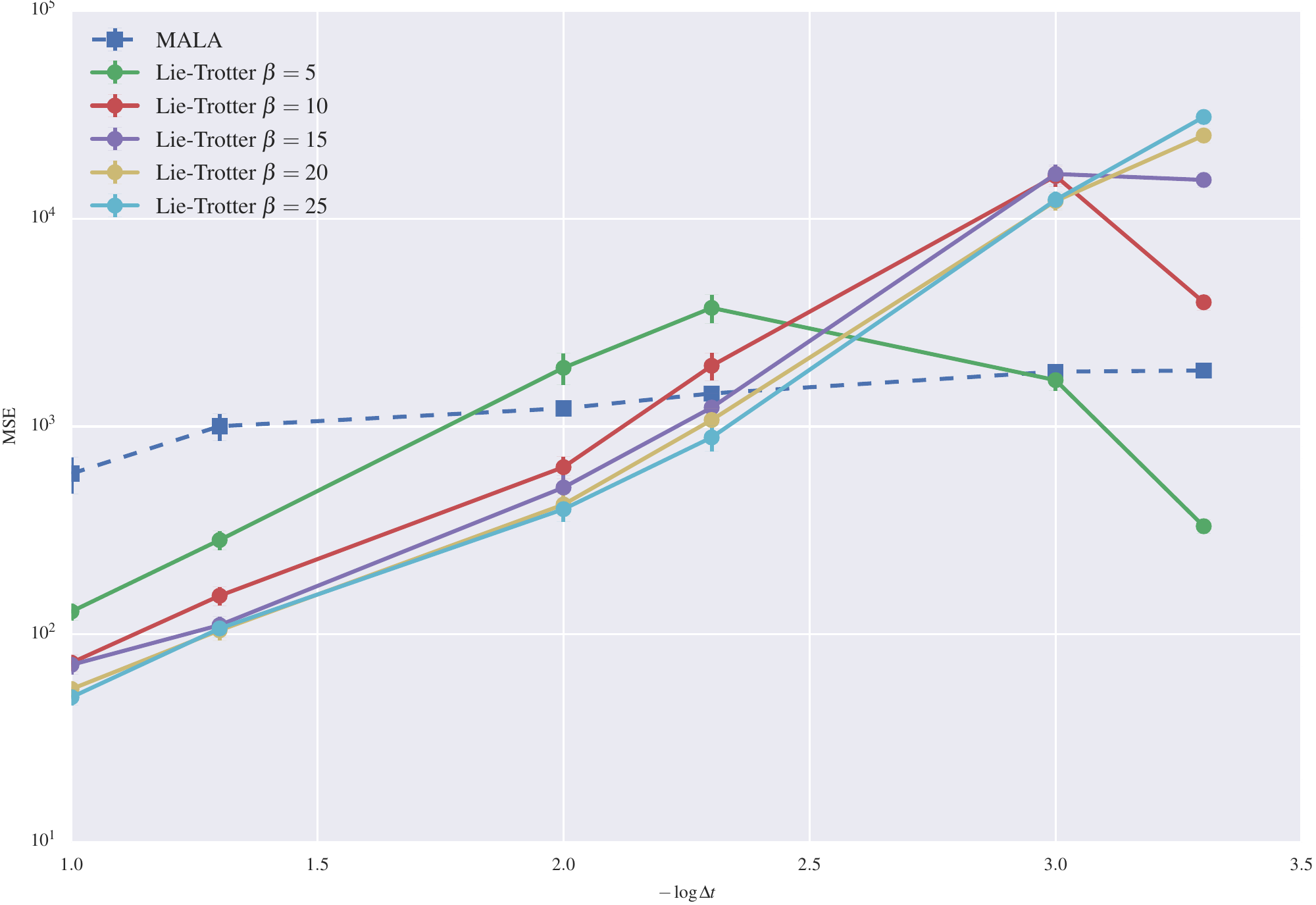}
    }
 \caption{Comparison of the MSE between MALA and different nonreversible samplers applied to the warped Gaussian distribution (\ref{eq:warped}).  The computational budget is set to $N=3.5\cdot10^{3}$ density evaluations, and $4^{th}$ order Runge-Kutta method is used for the nonreversible component.}
    \label{fig:banana_results}
  \end{figure}

\subsection{Logistic Regression}
Let $X$ be a $m \times d$ design matrix comprising $m$ samples with $d$ covariates and a binary response variable $Y \in \{-1,1 \}^{m}$. A Bayesian logistic regression model of the binary response is obtained by the introduction of the regression coefficient $\theta \in \IR^{d}$. For the sake of exposition, we shall assume a Gaussian prior of $\theta$, \emph{i. e.}, $\beta \sim \mathcal{N}(0,\Sigma)$. The posterior distribution $\pi(\theta|X,Y)$ is given by
\begin{equation} \label{eq:log_likelihood_logistic}
\pi(\theta| (X,Y)) \propto \exp \left(\sum_{i=1}^{m} Y_{i}\theta^{T}X_{i}-\log{(1+e^{\theta^{T}X_{i}})}-\frac{1}{2}\theta^{T}\Sigma^{-1} \theta \right)
\end{equation}
In Figure \ref{fig:pima_results} we investigate the use of the Lie Trotter sampler applied to this problem for the Pima indians \footnote{Here $m=768, d=9$.} dataset obtained from the UCI machine learning repository. The skew symmetric matrix $J$ is chosen by generating a random permutation $\sigma(1), \ldots, \sigma(d)$ and setting
\begin{equation*}
J_{\sigma(i), \sigma(i+1)} = 1 \mbox{ and } J_{\sigma(i+1), \sigma(i)} = -1,
\end{equation*}
for $i=1,\ldots, d-1$, and zero elsewhere.  In Figure \ref{fig:pima_results}a we plot the first estimator $\widehat{\pi}^{\Delta t}_{T}(\theta_1)$ with $95\%$ confidence intervals for different values of $\beta$ and stepsize.   Each point in the plot cost $3.5\cdot 10^{3}$ density evaluations.  To provide a comparison against the truth, an optimally tuned  MALA scheme was integrated over $10^{7}$ timesteps.  In Figure \ref{fig:pima_results}b we plot the effective sample size (ESS) of the Lie-Trotter scheme for different values of $\beta$ and $\Delta t$.  The markers denote the median value of the ESS with the markers denoting the $5\%$ and $95\%$ percentiles.   We note however that there typically be a very small number of observables for which the nonreversible scheme offers no advantage.  This agrees with the theory detailed in \cite{duncan2016variance} which characterises the minimum attainable variance reduction in terms of the projection of the observable $f$ on the nullspace of the operator $J\nabla V(x)\cdot\nabla$.   As $J$ is chosen randomly, there will always been a number of observables which are close to this subspace, and thus the nonreversible dynamics offer no advantage.   One possible remedy around this is to periodically resample the nonreversible matrix $J$,  but we do not investigate this here.

\begin{figure}[!ht]
  \hfill
    \subfloat[First covariate vs Step-size]{%
      \includegraphics[width=0.5\textwidth]{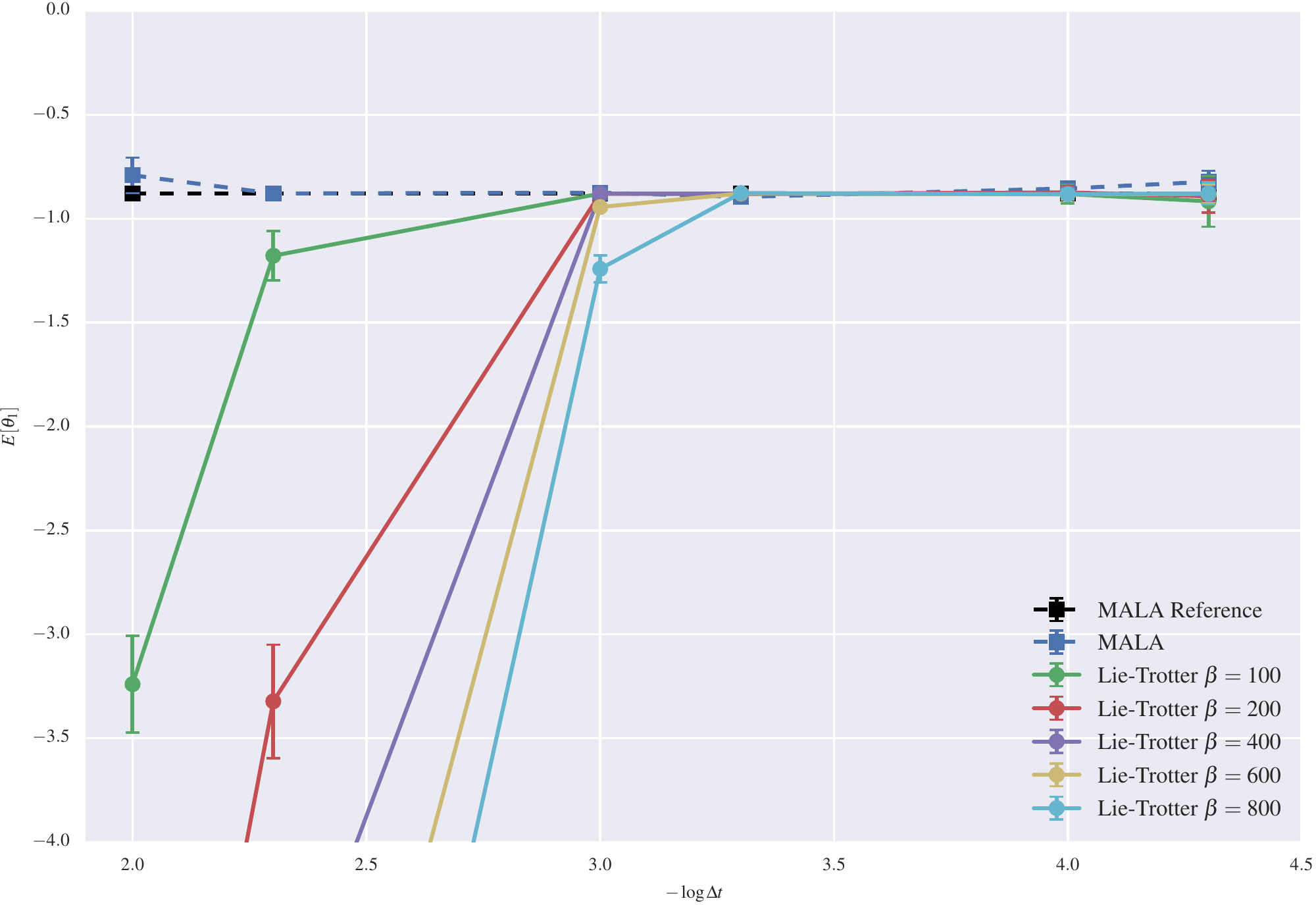}
    }
    \subfloat[ESS vs Step-size]{%
      \includegraphics[width=0.5\textwidth]{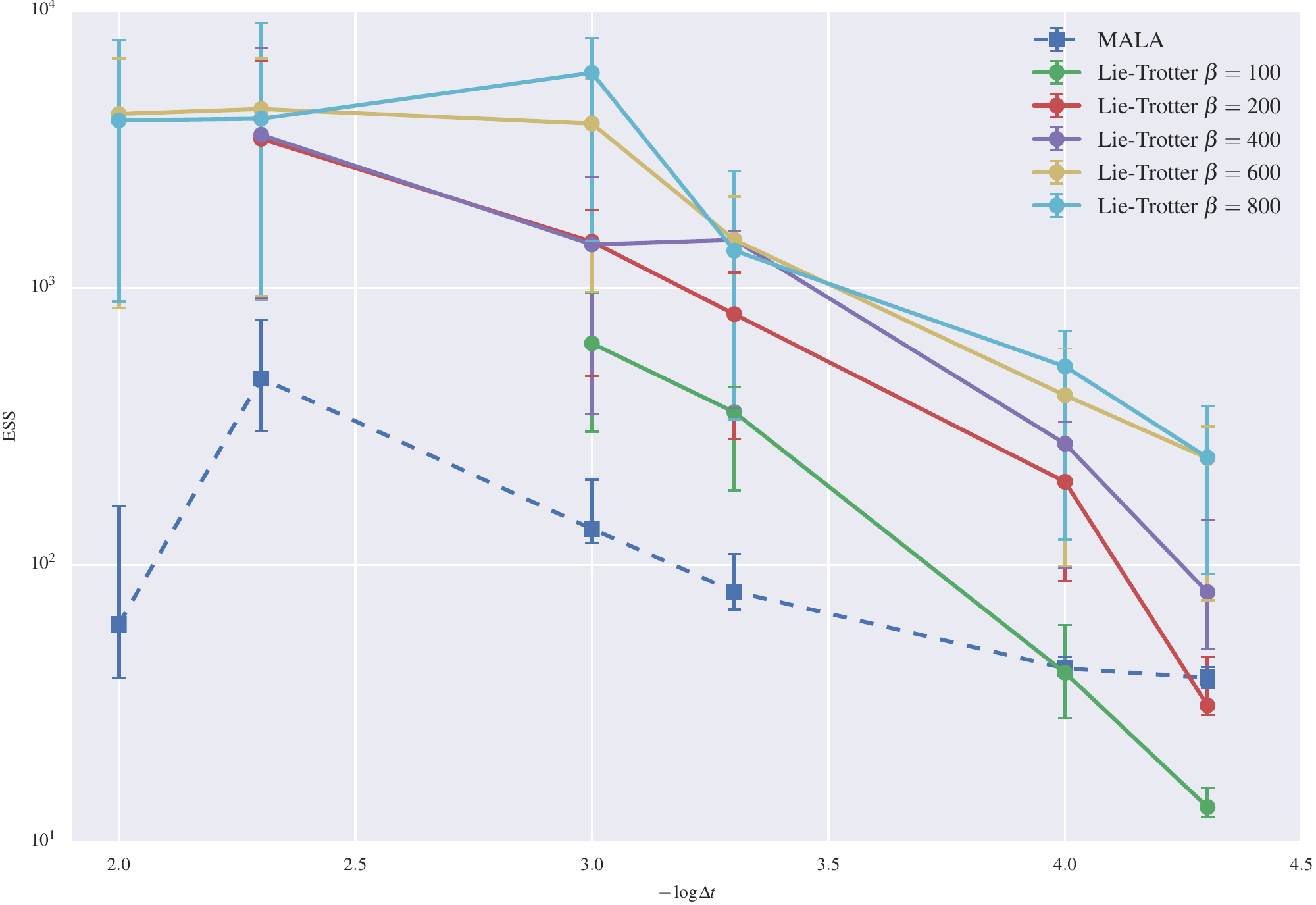}
    }
    \caption{Confidence Interval for the first covariate and ESS for estimators for $\pi\left(\theta_{i}\right)$, $i=1,\ldots,9$ for logistic regression of the Pima Indians data set.  Each data point in these plots is set to $3.5\cdot 10^{3}$ density evaluations.  The results are compared to an optimally tuned MALA simulation run for $10^{7}$ density evaluations.}
    \label{fig:pima_results}
  \end{figure}

  % \begin{figure}[!ht]
  % \hfill
  %   \subfloat[Bias vs Step-size]{%
  %     \includegraphics[width=0.5\textwidth]{sonar_bias.pdf}
  %   }
  %   \subfloat[ESS vs Step-size]{%
  %     \includegraphics[width=0.5\textwidth]{sonar_ess.pdf}
  %   }
  %   \caption{Bias and ESS for the covariates $\left(\theta_{i}\right)_{i=0}^{64}$ for logistic regression of the Sonar data set.   Each data point in these plots is set to $3.5\cdot 10^{3}$ density evaluations}
  %   \label{fig:sonar_results}
  % \end{figure}

\subsection{Spatial model}
We now consider a high dimensional target distribution related to inference for a log-Gaussian Cox point process previously considered in \cite{moller1998log}.  In particular, given the location of $126$ Scots pine saplings in a natural forest in Finland, we wish to infer the average intensity of a corresponding Poisson point process. Following \cite{christensen2005scaling}, we consider a discretised version of the model where the spatial region is discretised to a $64 \times 64$ regular grid. For each $i,j$ $X_{i,j}$ is the random variable counting the number of observations in the $(i,j)$-cell ,and hence the dimension of the problem is  $d=64^{2}=4096$. The observations are assumed to be generated by a Poisson point process with unobserved  intensity $\Lambda_{i,j}, i, j=1,\cdots,64$. Given the $\Lambda_{i,j}$ the random variables $X_{i,j}$ are assumed to be conditional independent with Poisson distributed mean $m \Lambda_{i,j}$, where $m=1/4096$ is the area of a single cell.  We impose a log-Gaussian prior on $\Lambda_{i,j}$, more specifically
\[
\Lambda_{i,j}=\exp{(Y_{i,j})}
\]
 where $Y=(Y_{i,j}, i,j=1,\cdots 64) \sim \mathcal{N}(\mu \mathbf{1}, \Sigma)$ where
 \[
 \Sigma_{i,j,i',j'}=\sigma^{2} \left[-\frac{-\{(i-i')^{2}+(j-j')^{2} \}^{1/2}}{64\beta} \right], \quad i,j,i',j'=1 \cdots, 64.
 \]
The posterior distribution is thus given by
\[
f(y|x) \propto \prod_{i,j =1}^{64} \exp\{(x_{i,j}y_{i,j})-m\exp(y_{i,j}) \} \exp \{ -0.5 (y-\mu \mathbf{1})^{T}\Sigma^{-1}(y-\mu \mathbf{1})\}
\]
Due to the poor scaling of the posterior distribution in \cite{christensen2005scaling} a reparametrization of $y$ is introduced to improve the
mixing of the Metropolis-Hastings scheme. This procedure is expensive with a  computational cost of  $\mathcal{O}(d^3)$. However, in the case of the nonreversible samplers, the nonreversible perturbation compensates for the poor scaling, thus rendering this reparametrisation unnecessary.

In Figures \ref{fig:pine_results} we plot  an estimator of $\IE(\Lambda\,|\,x)$ using MALA and its nonreversible
counterpart respectively. For this computation the skew-symmetric matrix $J$ was generated
randomly as in the logistic regression example.  Due to the large number of covariates, for any given random choice of $J
$, there would be a small number of covariates for which the nonreversible scheme does not offer significant advantage
over MALA, as described in \cite{duncan2016variance}. To better understand the  effect the nonreversible flow for an average covariate, we thus generate $10$ independent random skew-symmetric matrices,  and compute the average ESS over $J$.   The results are presented in Figure \ref{fig:pine_results1}.  In Figure
\ref{fig:pine_results1}c a histogram of the ESS over all covariates is plotted for both MALA and the splitting scheme for specific choices of $\Delta t$ and $\beta$. We observe that the ESS for the nonreversible scheme is orders of magnitude better than MALA. To illustrate the dependence of ESS on timestep, similarly to the case of logistic regression, in Figure \ref{fig:pine_results1}b we plot the median ESS for different choices of timestep.  It is clear that increasing $\beta$ and $
\Delta t$ as much as possible increases the ESS.  However, this comes at the cost of increasing bias as can be observed in Figure \ref{fig:pine_results1}a.  Nonetheless, it is evident that the nonreversible sampler significantly outperforms the MALA scheme.

\begin{figure}[!ht]
  \hfill
    \subfloat[Inferred Poisson intensity for MALA]{%
      \includegraphics[width=0.5\textwidth]{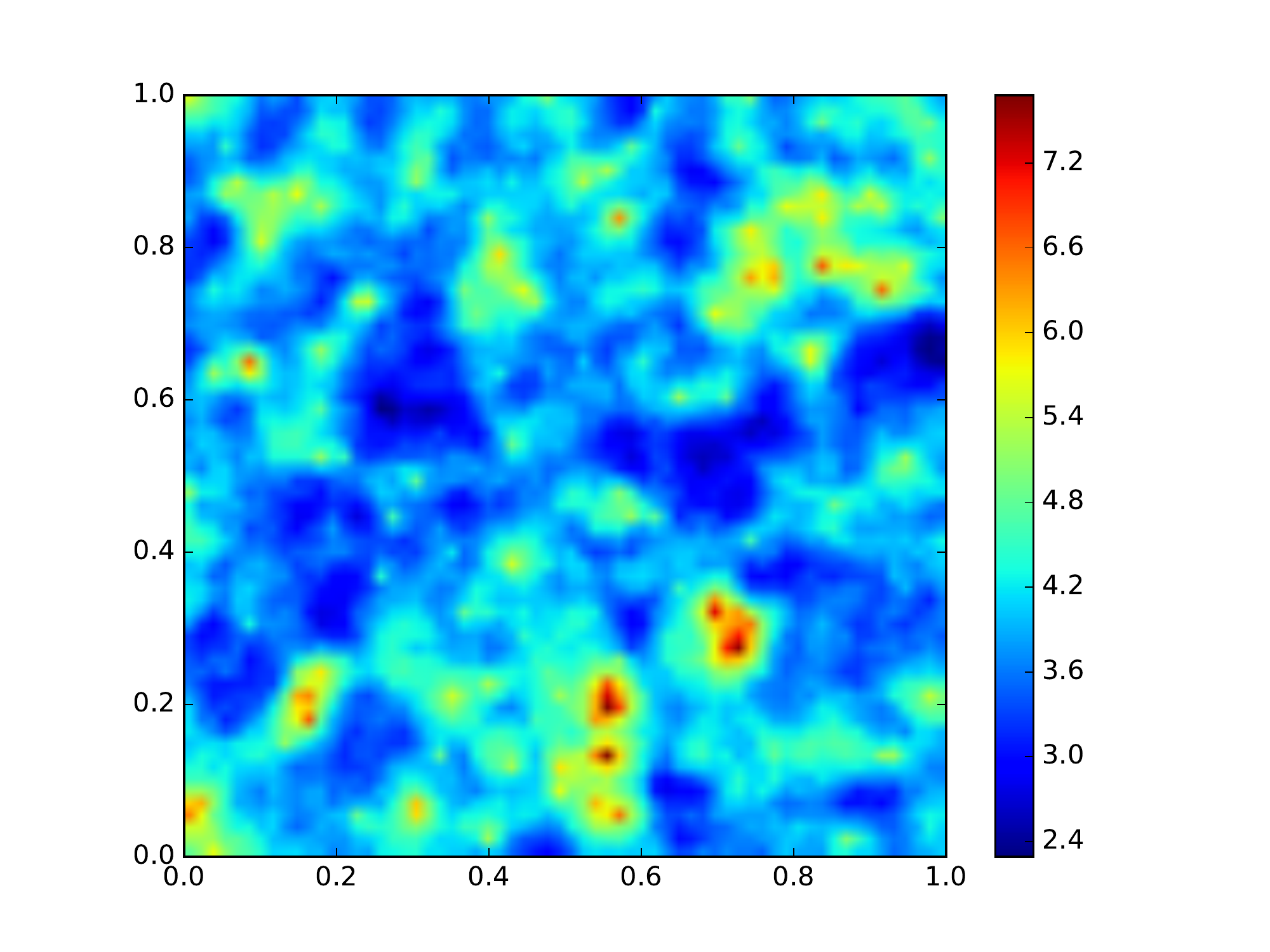}
    }
    \subfloat[Inferred Poisson intensity for Lie-Trotter scheme]{%
      \includegraphics[width=0.5\textwidth]{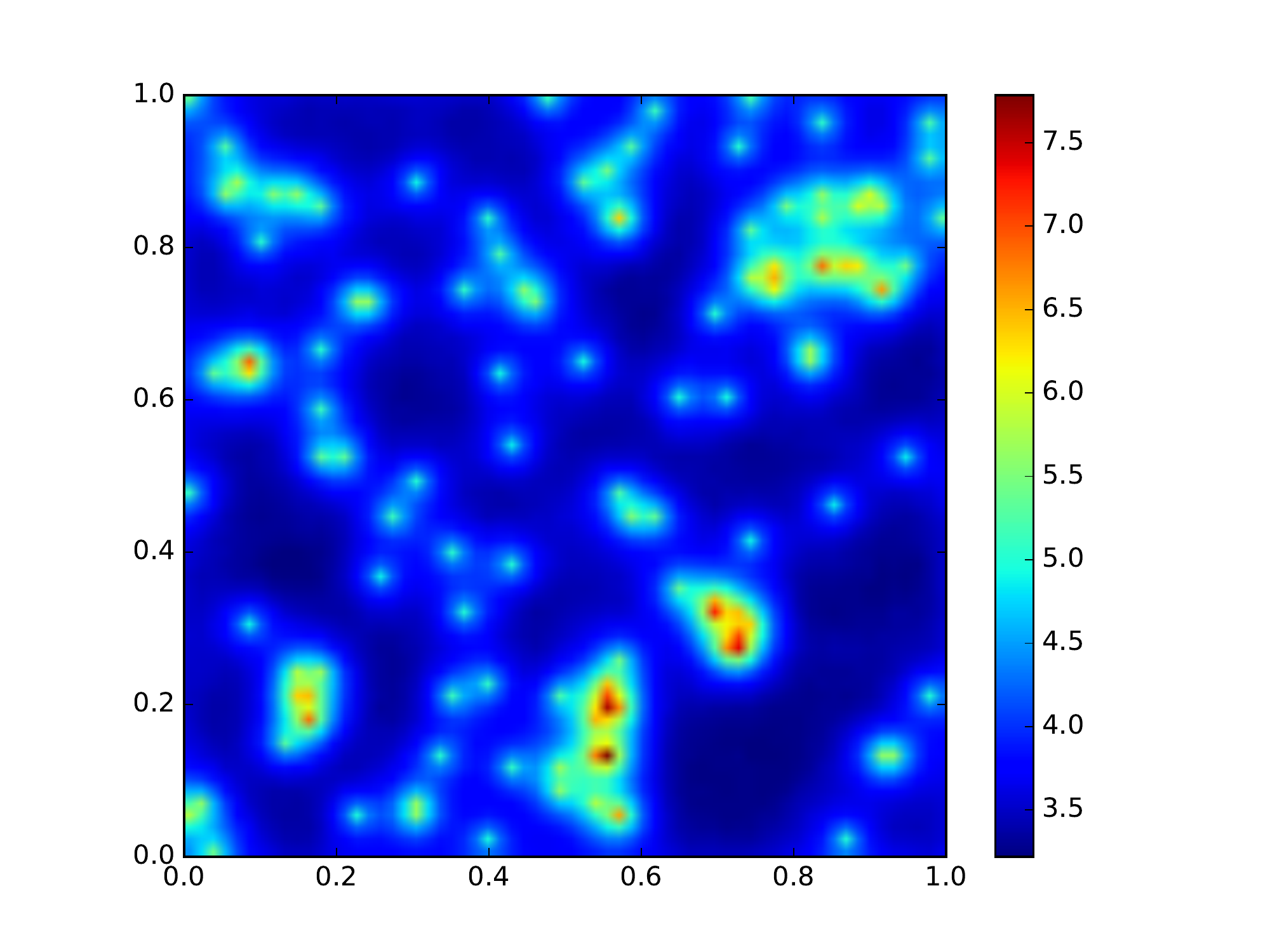}
    }
    \caption{Average inferred Poisson intensity using the different schemes. The computational budget is set to $N=3.5\cdot 10^{3}$ gradient evaluations.}
    \label{fig:pine_results}
  \end{figure}

 \begin{figure}[!ht]
  \hfill
    \subfloat[First covariate vs stepsize]{%
      \includegraphics[width=0.32\textwidth]{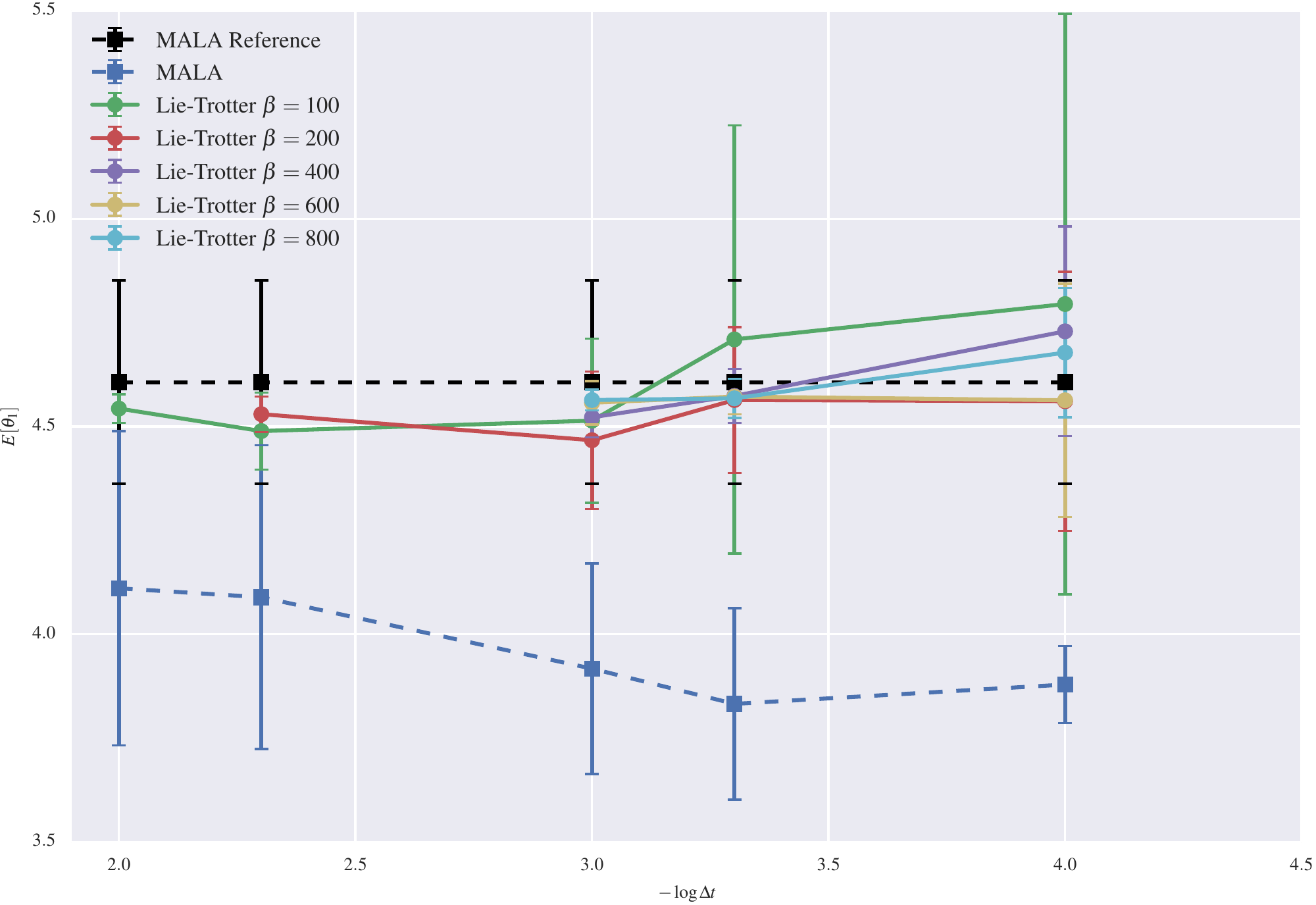}
    }
    \subfloat[ESS vs stepsize]{%
      \includegraphics[width=0.32\textwidth]{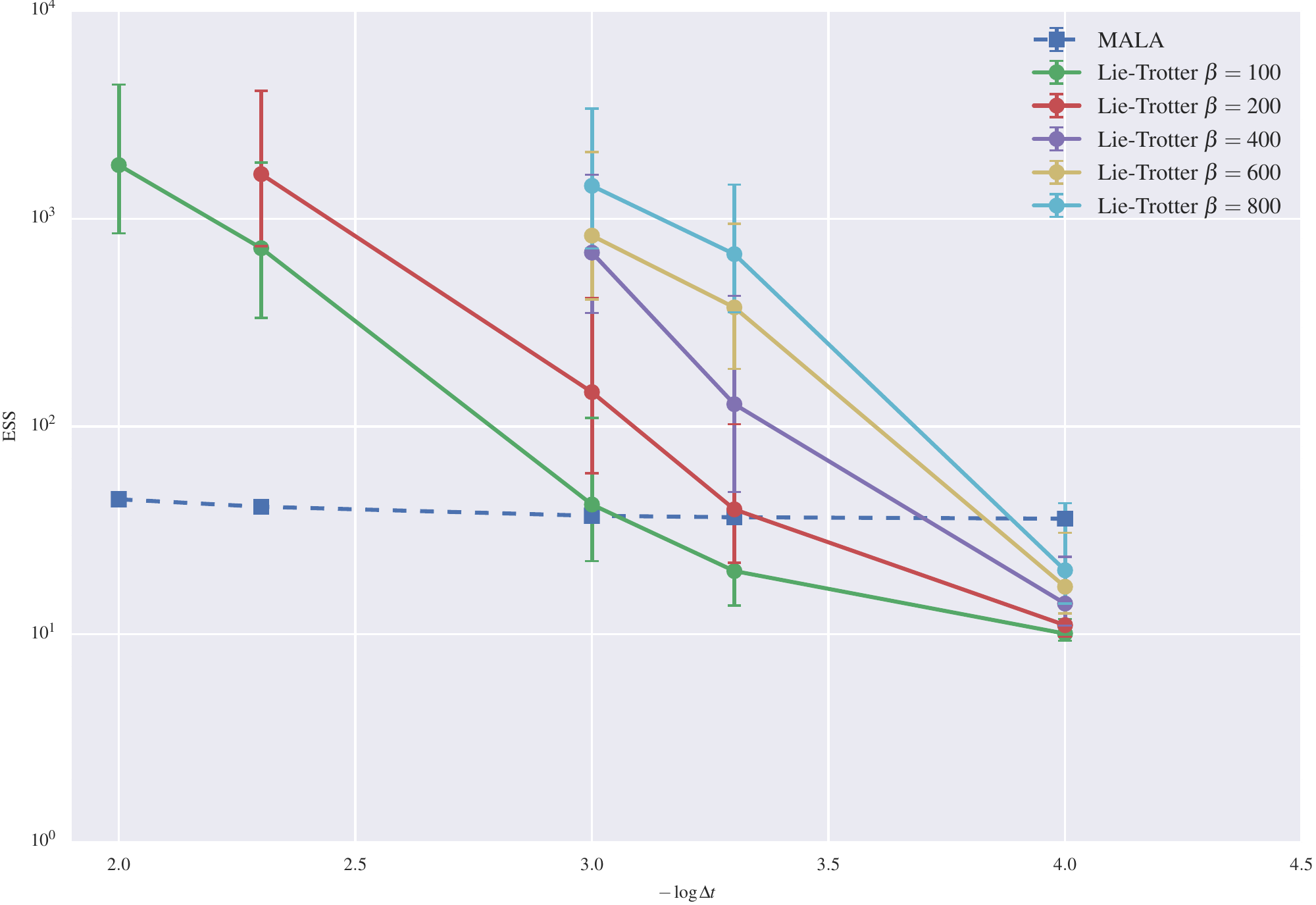}
    }
    \hfill
    \subfloat[Histogram of ESS]{%
      \includegraphics[width=0.32\textwidth]{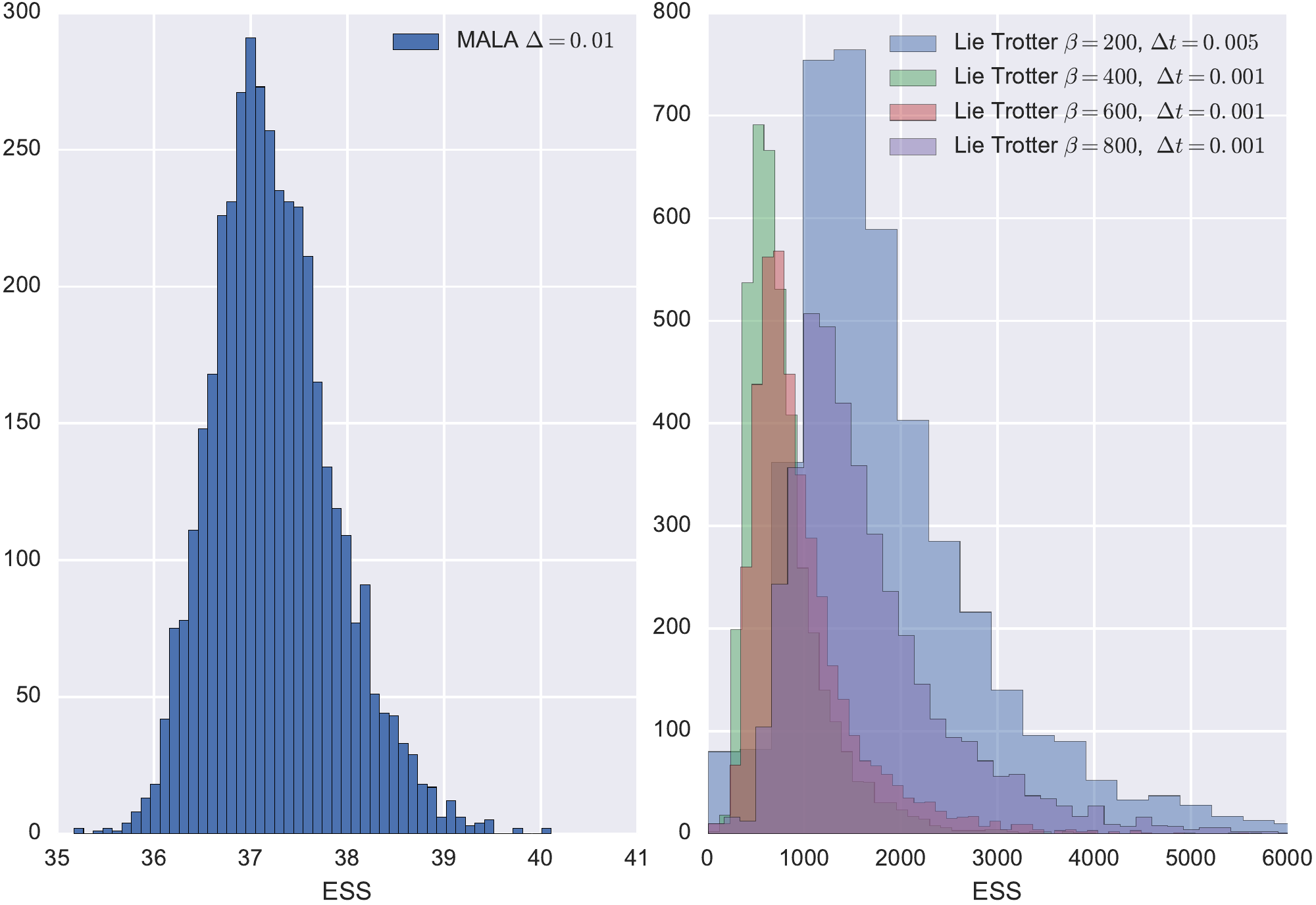}
    }
    \caption{Results for the inference of the log-Gaussian cox process. The computational budget is set to $N=3.5\cdot 10^{3}$ density evaluations.  A reference MALA simulation run for $10^{7}$ density evaluations is provided for comparison.}
    \label{fig:pine_results1}
  \end{figure}

\section{Proofs of the main results} \label{sec:proofs_main}
In this section we prove the main results of the paper. In particular, in Section \ref{subsec:geo_num_proofs} we prove the geometric ergodicity of the splitting scheme \eqref{eq:Lie}, while in Sections \ref{subsec:bias_proofs} and \ref{subsec:variance_proofs} we prove the results related to the asymptotic bias and variance of the splitting method.

\subsection{Ergodicity of the splitting scheme} \label{subsec:geo_num_proofs}
Here we prove the theorems and corollaries stated in Section \ref{subsec:geo_num}.
\begin{proof}[Proof of Theorem \ref{thm:geom_ergodic}]
% \subsubsection{Proof of Theorem \ref{thm:geom_ergodic}}
We verify the criteria for geometric ergodicity formulated in Chapters 15 and 16 of \cite{meyn1993markov}.
\begin{enumerate}
\item  We show that $P_{\Delta t}(x, \cdot)$ is $\pi$-irreducible.  Let  $A\subset \mathbb{R}^d$ such that $\pi(A) > 0$, then
    \begin{align*}
        P_{\Delta t}(x, A) &= \int_{A} q_{\Delta t}(y \, | \Phi_{\Delta t}(y))\alpha(\Phi_{\Delta t}(x), y)\,dy  \\ &\quad + \mathbf{1}_A(x) \int_{\mathbb{R}^d} q_{\Delta t}(z \, | \, \Phi_{\Delta t}(x))(1 - \alpha(z, \Phi_{\Delta t}(x)))\,dz > 0,
    \end{align*}
    which implies that $P_{\Delta t}$ is $\pi$-irreducible.
    \item We now show that every compact set $C$ of positive measure is small.  To this end, let $C$ be such a set and  $B$  a measurable subset of $C$. Then $D = C \cup \Phi_{\Delta t}(C)$ is also a compact set of positive measure.  Since the target density $\pi$ and proposal $q_{\Delta t}(y | x)$ are positive and continuous for all $x, y$, applying \cite[Lemma 1.2]{mengersen1996rates} implies that there exists $\eta > 0$ such that
    $$
        \widetilde{P}_{\Delta t}(x, B)  \geq \eta \pi(B), \quad B \subset D, x \in D.
    $$
    In particular,
    \begin{equation}
        \label{eq:smallness}
        \widehat{P}_{\Delta t}(x, B) = \widetilde{P}_{\Delta t}(\Phi_{\Delta t}(x), B)  \geq \eta \pi(B),\quad B \subset C, x \in C,
    \end{equation}
    so that $C$ is small.  Aperiodicity of the chain follows immediately from (\ref{eq:smallness}).
    \item
    To complete the proof we show that $P_{\Delta t}$ satisfies a Foster-Lyapunov condition for the Lyapunov function $V$. To this end using \eqref{eq:foster_rev}, given $x \in \mathbb{R}^d$:
    \begin{align*}
        \widehat{P}_{\Delta t}V(x) & \leq \lambda V(\Phi_{\Delta t}(x)) + b \mathbf{1}_C(\Phi_{\Delta t}(x)) \\
         &\leq \lambda V(\Phi_{\Delta t}(x)) + b \mathbf{1}_{\Phi^{-1}_{\Delta t}(C)}(x) \\
         &\leq \lambda V(x) + \lambda \left( V(\Phi_{\Delta t}(x)) - V(x)\right) + b \mathbf{1}_{\Phi^{-1}_{\Delta t}(C)}(x).
     \end{align*}
     By Assumptions \ref{ass:nonreversible}(2) and \ref{ass:nonreversible}(3), there exists a compact set $D\subset \mathbb{R}^d$ and $0 < c < 1$ such that $\Phi_{\Delta t}^{-1}(C) \subset D$ and moreover
     $$
        \lambda V(x) + \lambda \left( V(\Phi_{\Delta t}(x)) - V(x)\right) \leq c V(x),  \quad x \in \mathbb{R}^d\setminus D,
     $$
     which implies that
     $$
      \widehat{P}_{\Delta t}V(x) \leq c V(x) + b \mathbf{1}_{D}(x),
     $$
     as required.
    \end{enumerate}
\end{proof}

% \subsubsection{Proof of Corollary \ref{cor:rwm}}
%    By \cite[Theorem 3.2]{roberts1996geometric}, the RWMH chain satisfies a Foster drift condition for the Lyapunov function $V(x) = \pi^{-1/2}(x)$, so that (A2) of Assumptions \ref{ass:nonreversible} holds.   In particular, for $U(x) = -\log\pi(x)$ we have
%     $$
%         \frac{V(\Psi_{\Delta t}(x)) - V(x)}{V(x)} = e^{(U(\Psi_{\Delta t}(x)) - U(x))/2} - 1,
%     $$
%     so that, by (\ref{eq:rwmh_condition}) given $\Delta t> 0$, there exists $M > 0$ such that $(V(\Psi_{\Delta t}(x)) - V(x))/V(x) < \delta$ for  all $ |x|\geq M$.   The conditions of Theorem \ref{thm:geom_ergodic} thus hold, so that $\widehat{X}_{n}^{\Delta t}$ is geometrically ergodic.

\begin{proof}[Proof of Corollary \ref{cor:mala}]
 Provided that the conditions of \cite[Theorem 4.1]{roberts1996exponential} hold, then the MALA chain satisfies a Foster-Lyapunov condition for $V(x) = e^{s|x|}$ for $s > 0$ sufficiently small.  If we  consider
    $$
        \frac{V(\Phi_{\Delta t}(x)) - V(x)}{V(x)} = e^{s(|\Phi(x)|-|x|)} \leq e^{s(|\Phi(x)|-|x|)},
    $$
  then Assumption \ref{ass:nonreversible}(2) follows immediately.  Finally, we note that (\ref{eq:mala_condition}) implies that there exists $K > 0$ such that
  $$
   \left|\Phi_{\Delta t}(x)\right| - K \leq |x| \leq \left|\Phi_{\Delta t}(x)\right| + K,\quad x \in \mathbb{R}^d,
  $$
  from which \ref{ass:nonreversible}(3) follows immediately.  Hence, the conditions of Theorem \ref{thm:geom_ergodic} all hold, and thus the process $\widehat{X}_{n}^{\Delta t}$ is geometrically ergodic.
  \\\\
  Suppose now that $\gamma = J\nabla \pi^{\alpha}$, where $J = -J^\top$ and $\alpha > 0$.  Suppose $\Phi_{\Delta t}(x)$ is an explicit Runge-Kutta discretisation of the nonreversible dynamics having $s$ stages.  Then we can write
  \begin{equation}
  \Phi_{\Delta t}(x) = x + h \sum_{i=1}^{s}b_i k_i(x),
  \end{equation}
  where
  \begin{align*}
    k_1(x) &= \gamma(x) \\
    k_2(x) &= \gamma\left(x + h w_{2,1}k_1(x)\right) \\
    k_3(x) &= \gamma\left(x + h(w_{3,1}k_1(x) + w_{3,2}k_2(x))\right)\\
     & \vdots \\
    k_s(x) &= \gamma\left(x + h\sum_{i=1}^{s-1} w_{s,i}k_i(x)\right),
  \end{align*}
  where $(w_{i,j})$ is the Runge-Kutta matrix associated with the discretisation.  By (\ref{eq:pi_decay_ass}) there exist positive constants $\alpha'$, $K'$ and $K_1$ such that
  $$k_1(x) \leq |\gamma(x)| \leq K' |\nabla \pi^{\alpha}(x)| \leq K_1 \pi^{\alpha'}(x).$$
Suppose now that there exists constants $K_2, \ldots, K_{i-1}$ such that
  $$
    |k_j(x)| \leq K_j \pi^{\alpha'}(x), \quad x \in \mathbb{R}^d, \quad j =1,\ldots, i-1.
  $$
By (\ref{eq:pi_decay_ass}) the matrix $\nabla\gamma = \left(\partial_{x_i}\gamma_j(x)\right)_{i,j}$ has bounded components in $\mathbb{R}^d$ and so applying the mean value theorem to every component of $\gamma$, it follows that
  $$
    |k_{i}(x)| \leq |\gamma(x)| + h\left(\sup_{x\in \mathbb{R}^d}|\nabla \gamma(x)|_{max}\right) \sum_{j=1}^{i-1}|w_{i, j}| |k_j(x)| \leq K_i \pi^{\alpha'}(x),
  $$
  for some constant $K_i$.  It follows by induction that $\left|\Phi_{\Delta t}(x) - x \right| \leq K \pi^{\alpha'}(x)$, for all $x \in \mathbb{R}^d$, which implies (\ref{eq:mala_condition}).  The corresponding result for $\gamma$ given by (\ref{eq:finitely_supported_gamma}) follows similarly.
  \end{proof}

\subsection{Asymptotic variance of numerical integrators}  \label{subsec:as_var_num_proofs}
Here we prove Proposition \ref{prop:correlation_expansion_discrete} and Theorem \ref{thm:correlation_expansion_discrete}  which characterises the error in the asymptotic variance for an arbitrary numerical integrator

\begin{proof}[Proof of Proposition \ref{prop:correlation_expansion_discrete}]
It follows from standard elliptic regularity that the operator $(-\mathcal{L})^{-1}$ is bounded on $L_0^\infty(\pi)$.  Similarly the operator $\Delta t ( I - P_{\Delta t})^{-1}$ is bounded on $L_0^\infty(\pi)$, uniformly with respect to $\Delta t$.
\\\\
Let $\psi \in C^\infty(\mathbb{T}^d)$ with $\pi(\psi) = 0$.  There exists $R_\psi$, smooth and bounded uniformly with respect to $\Delta t$ such that
\begin{equation} \label{eq:useful}
\left(\frac{I - P_{\Delta t}}{\Delta t}\right)\psi = -\mathcal{L}\psi - \frac{\Delta t}{2}\mathcal{L}^2\psi - \frac{\Delta t^2}{6}\mathcal{L}^3 \psi +  \Delta t ^3 R_{\psi},
\end{equation}
provided that $\Delta t$ is sufficiently small.  Hence using \eqref{eq:useful} we obtain
\begin{equation}
\label{eq:poisson_expansion}
\begin{aligned}
(-\mathcal{L})^{-1}\psi &= \left(\frac{I - P_{\Delta t}}{\Delta t}\right)^{-1}\left(\frac{I - P_{\Delta t}}{\Delta t}\right)(-\mathcal{L})^{-1}\psi\\
&= \left(\frac{I - P_{\Delta t}}{\Delta t}\right)^{-1}\psi + \frac{\Delta t}{2} \left(\frac{I - P_{\Delta t}}{\Delta t}\right)^{-1}\mathcal{L}\psi + \frac{\Delta t^2}{6} \left(\frac{I - P_{\Delta t}}{\Delta t}\right)^{-1}\mathcal{L}^2\psi \\
    &\quad + \Delta t^3 \ \left(\frac{I - P_{\Delta t}}{\Delta t}\right)^{-1}{R}_{(-\mathcal{L})^{-1}\psi}.
\end{aligned}
\end{equation}
Since both sides of \eqref{eq:useful} has mean zero and $\pi(\mathcal{L}^i\psi) = 0$ for $i \geq 0$, it follows that $\pi(R_{(-\mathcal{L})^{-1}\psi}) = 0$.  Thus the remainder term in (\ref{eq:poisson_expansion}) is well-defined and  uniformly bounded with respect to $\Delta t$.
\\\\
Now let $f \in C^\infty(\mathbb{T}^d)$, then similar to  (\ref{eq:discrete_as_var}), the asymptotic variance of the estimator $N^{-1}\sum_{n=0}^{N-1} f(X^{\Delta t}_n)$ for the discretized exact process is given by
$$
\sigma^2_{\Delta t}(f) = 2\left\langle \left(\frac{I - P_{\Delta t}}{\Delta t}\right)^{-1}(f - \pi(f)), f - \pi(f)\right\rangle_{\pi} - \Delta t\mbox{Var}_{\pi}[f].
$$
By (\ref{eq:poisson_expansion}) it follows that
\begin{align*}
\sigma^2_{\Delta t}(f) &= 2 \left\langle (-\mathcal{L})^{-1} (f - \pi(f)), f - \pi(f) \right\rangle_{\pi} \\ &\quad + \Delta t \left\langle \left(\frac{I - P_{\Delta t}}{\Delta t}\right)^{-1}(-\mathcal{L})(f-\pi(f)), f-\pi(f)\right\rangle_{\pi} \\ &\quad - \frac{\Delta t^2}{3} \left\langle \left(\frac{I - P_{\Delta t}}{\Delta t}\right)^{-1}(-\mathcal{L})^2(f-\pi(f)), f-\pi(f)\right\rangle_{\pi}  + \Delta t ^3 R_f - \Delta t\mbox{Var}_{\pi}[f],
\end{align*}
where $R_f$ is a remainder term depending on $f$.  Since $f$ is smooth, we can iteratively apply (\ref{eq:poisson_expansion}) to the second term and third terms on the RHS obtaining
\begin{align*}
\sigma^2_{\Delta t}(f) &= 2 \left\langle (-\mathcal{L})^{-1} (f - \pi(f)), f - \pi(f) \right\rangle_{\pi} \\ &\quad + \Delta t \left\langle (-\mathcal{L})^{-1}(-\mathcal{L})(f-\pi(f)), f-\pi(f)\right\rangle_{\pi} - \Delta t\mbox{Var}_{\pi}[f]\\
&\quad + \frac{\Delta t^2}{6}\left\langle (-\mathcal{L})(f-\pi(f)), f-\pi(f)\right\rangle_{\pi}\\
 &\quad + \Delta t^3 R_f \\
    &= \sigma^2(f) + \frac{\Delta t^2}{6}\left\langle (-\mathcal{L})(f-\pi(f)), f-\pi(f)\right\rangle_{\pi} +  \Delta t^3 \widetilde{R}_f,
\end{align*}
as required.
\end{proof}

\begin{proof}[Proof of Theorem \ref{thm:correlation_expansion_discrete}]
The proof of this result follows closely that of \cite[Theorem 2.9]{leimkuhler2013computation}.  To this end, given $f, g \in C^\infty(\mathbb{T}^d)$ such that $\pi(f) = \pi(g)  = 0$, consider
\begin{align*}
\left\langle \left(\frac{I - P_{\Delta t}}{\Delta t} \right)^{-1}f, g \right\rangle_{\pi}.
\end{align*}
Since $\left(\frac{I - P_{\Delta t}}{\Delta t} \right)^{-1}f$ has mean zero with respect to $\pi$, then
\begin{align*}
\left\langle \left(\frac{I - P_{\Delta t}}{\Delta t} \right)^{-1}f, g \right\rangle_{\pi} &=\left\langle \left(\frac{I - P_{\Delta t}}{\Delta t} \right)^{-1}f, M_{\Delta t}g \right\rangle_{\pi} \\
&= \left\langle \left(\frac{I - P_{\Delta t}}{\Delta t} \right)^{-1}f, M_{\Delta t}g \right\rangle_{\widehat{\pi}_{\Delta t}} + \Delta t ^r R_{f,g},
\end{align*}
for a smooth remainder term $R_{f,g}$ bounded uniformly with respect to $\Delta t$.  Using the expansion \eqref{eq:semigroup_expansion} for the semigroup $\widehat{P}_{\Delta t}$:
\begin{equation}
\label{eq:generator_expansion_discrete}
\begin{aligned}
&\left\langle \left(\frac{I - P_{\Delta t}}{\Delta t} \right)^{-1}f, M_{\Delta t}g \right\rangle_{\widehat{\pi}_{\Delta t}} = \left\langle \left(\frac{I - \widehat{P}_{\Delta t}}{\Delta t} \right)^{-1}M_{\Delta t}\left(\frac{I - \widehat{P}_{\Delta t}}{\Delta t} \right)\left(\frac{I - P_{\Delta t}}{\Delta t} \right)^{-1}f, M_{\Delta t}g \right\rangle_{\widehat{\pi}_{\Delta t}}\\
&= \left\langle \left(\frac{I - \widehat{P}_{\Delta t}}{\Delta t} \right)^{-1}M_{\Delta t}\left( \frac{I - P_{\Delta t}}{\Delta t} + \Delta t^{k}\left(\frac{\mathcal{L}^{k+1}}{(k+1)!} - A_{k}\right)\right)\left(\frac{I - P_{\Delta t}}{\Delta t} \right)^{-1}f, M_{\Delta t}g \right\rangle_{\widehat{\pi}_{\Delta t}} \\
&\qquad +\Delta t^{q-1}\left\langle \left(\frac{I - \widehat{P}_{\Delta t}}{\Delta t} \right)^{-1}M_{\Delta t}R_f, M_{\Delta t}g \right\rangle_{\widehat{\pi}_{\Delta t}}\\
&=\left\langle \left(\frac{I - \widehat{P}_{\Delta t}}{\Delta t} \right)^{-1}M_{\Delta t}f, M_{\Delta t} g\right\rangle_{\widehat{\pi}_{\Delta t}}\\
&\qquad + \Delta t^{k} \left\langle \left(\frac{I - \widehat{P}_{\Delta t}}{\Delta t} \right)^{-1}M_{\Delta t}\left(\frac{\mathcal{L}^{k+1}}{(k+1)!} - A_{k}\right) \left(\frac{I - {P}_{\Delta t}}{\Delta t} \right)^{-1}f, M_{\Delta t}g\right\rangle_{\widehat{\pi}_{\Delta t}} \\
&\qquad +\Delta t^{q-1}\left\langle \left(\frac{I - \widehat{P}_{\Delta t}}{\Delta t} \right)^{-1}M_{\Delta t}R_f, M_{\Delta t}g \right\rangle_{\widehat{\pi}_{\Delta t}},
\end{aligned}
\end{equation}
where $R_f$ is a smooth function depending on $f$, bounded uniformly with respect to $\Delta t$.   By Assumption \ref{ass:un_erg}, the coefficients of the $\Delta t^k$ and $\Delta t^{q-1}$ terms are bounded uniformly with respect to $\Delta t$.  Equation (\ref{eq:covariance_expansion1}) then follows immediately, and thus (\ref{eq:covariance_expansion}).   Assume now that (\ref{ass:mean_zero_q}) holds, then by applying (\ref{eq:generator_expansion_discrete}) with
$$
    f = \left(\frac{\mathcal{L}^{k+1}}{(k+1)!} - A_{k}\right)\left(\frac{I - P_{\Delta t}}{\Delta t}\right)^{-1}f,\quad \mbox{ and } \quad g = g,
$$
we obtain
\begin{align*}
R_1(f,g) &= \left\langle \left(\frac{I - \widehat{P}_{\Delta t}}{\Delta t} \right)^{-1}M_{\Delta t}\left(\frac{\mathcal{L}^{k+1}}{(k+1)!} - A_{k}\right) \left(\frac{I - {P}_{\Delta t}}{\Delta t} \right)^{-1}f, M_{\Delta t}g\right\rangle_{\widehat{\pi}_{\Delta t}}\\
&\qquad = \left\langle \left(\frac{I - {P}_{\Delta t}}{\Delta t} \right)^{-1}\left(\frac{\mathcal{L}^{k+1}}{(k+1)!} - A_{k}\right) \left(\frac{I - {P}_{\Delta t}}{\Delta t} \right)^{-1}f, g\right\rangle_{\pi} + \Delta t^{q-1} R_2(f,g),
\end{align*}
for some smooth, uniformly bounded remainder term $R_2$.  We now apply (\ref{eq:generator_expansion_discrete1}) to the discrete generator $\Delta t^{-1}(I - P_{\Delta t})$ to obtain
\begin{align*}
\left\langle \left(\frac{I - {P}_{\Delta t}}{\Delta t} \right)^{-1}\left(\frac{\mathcal{L}^{k+1}}{(k+1)!} - A_{k}\right) \left(\frac{I - {P}_{\Delta t}}{\Delta t} \right)^{-1}f, g\right\rangle_{\pi} &= \left\langle \left(-\mathcal{L} \right)^{-1}\left(\frac{\mathcal{L}^{k+1}}{(k+1)!} - A_{k}\right) \left(-\mathcal{L} \right)^{-1}f, g\right\rangle_{\pi} \\ &\quad +\Delta t R_3(f,g),
\end{align*}
for a smooth bounded remainder term $R_3$, from which (\ref{eq:covariance_remainder_term}) follows.
% where $R_1$, $R_2$, $R_3$ and $R_4$ are smooth remainder terms which are bounded uniformly with respect to $\Delta t$, by Assumption \ref{ass:un_erg}.

%\left\langle \left(\frac{I - \widehat{P}_{\Delta t}}{\Delta t} \right)^{-1}M_{\Delta t}R_f, M_{\Delta t}g \right\rangle_{\widehat{\pi}_{\Delta t}},
\end{proof}

\subsection{Asymptotic bias of the splitting scheme}
\label{subsec:bias_proofs}
Here we prove the results from Section \ref{subsec:bias}

\begin{proof}[Proof of Theorem \ref{thm:bias}]
Assume that the transition semigroup associated with $\widehat{X}^{\Delta t}$ satisfies the expansion \eqref{eq:semigroup_expansion}.  In order to prove the first part of Theorem \ref{thm:bias}  it is enough to show that
\begin{equation} \label{eq:qed}
 A^{*}_{j} \pi=0 \quad \text{for} \quad j=1, \cdots r-1,\qquad
 A^{*}_r  \pi = \Ddiv(f_{r} \pi).
\end{equation}
The result then follows immediately from Theorem \ref{th:general} using the identity
\begin{equation} \label{eq:ide}
\int_{\mathbb{T}^{d}}A_{r}\psi(z)\pi(z)dz=-\int_{\mathbb{T}^{d}}\psi(z)\Ddiv(f_{r}(z) \pi(z)) dz.
\end{equation}
We now start with the calculation of $A_{j}$. In particular, given $\phi\in C^\infty(\mathbb{T}^d)$ and $x\in\IR^{N}$, using the semigroup property of the Markov process we have
\begin{equation}
\label{semi_g}
\IE\left[\phi\left(\widehat{X}_1^{\Delta t}\right) |\widehat{X}_0^{\Delta t} = x\right] =
\IE\left[\phi\left(\Phi_{\Delta t}  \circ \Theta_{\Delta t}  (x)\right)\right] =
 e^{\Delta t\mathcal{L}_{S,num}} {(\phi \circ \Phi_{\Delta t})}(x),     \end{equation}
%where $\mathcal{L}_S$ is the generator of \eqref{eq:stochastic_part}
where $e^{\Delta t\mathcal{L}_{S,num}}\phi$ denotes the numerical flow
generated by the numerical method applied to the reversible part of the dynamics \eqref{eq:reversible}.   We next recall the generator \eqref{flow:bel} of the truncated modified equation \eqref{eq:beatr} of the integrator $\Phi_{\Delta t}$,
$$
\widetilde{\mathcal{L}}_D \phi = F_0 + \Delta tF_1 \phi + \ldots + \Delta t^r F_r \phi + \Delta t^{r+1}R_{\phi},
$$
where $R_{\phi}$ is a smooth remainder term bounded uniformly with respect to $\Delta t$ and where we define the differential operators $F_j\phi=f_j\cdot \nabla\phi$ (with $f_0=f$).
% {We have, using \eqref{semi_g}, applying Remark \ref{rem:talay} for $\mathcal{L}_{S,num}$, and using \eqref{flow:be} with $M=s=r$,}
We then have
\begin{eqnarray*}
\IE\left[\phi\left(\widehat{X}_1^{\Delta t}\right) |\widehat{X}_0^{\Delta t} = x \right] &=&
\left(\sum_{k=0}^{r} \frac{\Delta t^k\mathcal{L}_{S,num}^k}{k!}  \right)
\left(\sum_{k=0}^{r} \frac{\Delta t^k\widetilde{\mathcal{L}}_D^k}{k!} \right) \phi(x)+ \Delta t^{r+2}R'_{\phi} \\
&=&\phi(x) + \Delta t\mathcal{L}\phi(x)+\sum_{k=1}^{r} \Delta t^{k+1} A_k  \phi(x) + \Delta t^{r+2} R'_{\phi},
\end{eqnarray*}
for a smooth remainder term $R'_{\phi}$ and where
$$
A_k = \sum_{j=0}^{k+1}  \mathcal{L}_{S,num}^{k+1-j}
\Big(
\sum_{\tiny\begin{array}{c}1\leq i\leq j \\ n_{1}+n_{2}+\cdots+n_{i}=j-i\end{array}}  \frac1{i!(k+1-j)!}{F_{n_{1}}}\cdots {F_{n_{i}}}\Big),
$$
where the second sum above is over integers $n_1,\ldots,n_i \geq 0$ and is equal to the identity $I$ when  $j=0$.
%with $|n_{i}|$ we denote the number of times that the index $n_{i}$ is taken into account, e.g if we have $B_{n_{i}}B_{n_{i}}$ appearing in the sum then $|n_{i}|=2$.
We obtain for all $k\geq 1$,
$$
A_k^*\pi =
\sum_{j=0}^{k+1}
\Big(
\sum_{\tiny\begin{array}{c}1\leq i\leq j \\ n_{1}+n_{2}+\cdots+n_{i}=j-i\end{array}}  \frac1{i!(k+1-j)!}{F_{n_{i}}^*\cdots {F_{n_{1}}^*}}\Big)(\mathcal{L}^*_{S,num})^{k+1-j} \pi.
$$
Now since the integrator applied to the reversible part preserves the invariant measure we have  $\mathcal{L}_{S,num}^* \pi = 0$ which together with  $F^{*}_{i}\pi=0, \ i=1,\cdots r-1$ implies that
 for $k\leq r$, the only possibly non-zero term in the above sum is obtained for $j=r+1,k=r,i=1,$ i.e.,
$F^{*}_r  \pi=\Ddiv(f_r \pi)$. Hence, we deduce \eqref{eq:qed} which permits to conclude the proof.
\end{proof}

\subsection{Asymptotic variance of the splitting scheme}
\label{subsec:variance_proofs}
Here we prove the results from Section \ref{subsec:variance}.

\begin{proof}[Proof of Theorem \ref{thm:variance}]
Clearly, Assumption \ref{ass:un_erg} holds immediately from Theorem \ref{thm:unif_ergodic} in the Appendix.  Consider the one step semigroup $\widehat{P}_{\Delta t} = \Theta_{\Delta t}\widehat{\Phi}_{\Delta t }$ be the one-step semigroup corresponding to the Lie-Trotter splitting scheme \eqref{eq:Lie}, where $\Theta_{\Delta t}$ is the one-step semigroup integrated by MALA.  By \ref{prop:mala_expansion} one obtains
$$
    \widehat{P}_{\Delta t} \phi = \phi + \Delta t A_0 \phi + \Delta t^2 A_1 \phi + \Delta t^{5/2}R_{\phi},
$$
where
\begin{align*}
    A_0 &= \mathcal{A}_1 + \mathcal{G}_1 = \mathcal{L}, \\
    A_1 &= \mathcal{A}_2 + \mathcal{G}_1\mathcal{A}_1 + \mathcal{G}_2,
\end{align*}
and where $R_{\phi}$ is a smooth remainder term, bounded uniformly with respect to $\Delta t$.  Since the integrator $\widehat{\Phi}_{\Delta t}$ is assumed to preserve the invariant distribution up to order $2$, and $\widehat{\Theta}_{\Delta t}$ preserves $\pi$ it follows that
$$
    \pi\left(\left(\mathcal{A}_2 + \mathcal{G}_1\mathcal{A}_1+ \mathcal{G}_2\right)\phi\right)= 0, \quad \phi \in C^\infty(\mathbb{T}^d).
$$
Applying Theorem \ref{thm:correlation_expansion_discrete}, it follows that for $f \in C^\infty(\mathbb{T}^d)$,
$$
    \widehat{\sigma}^2_{\Delta t}(f) = \sigma^2_{\Delta}(f) + \Delta t R_f + o(\Delta t),
$$
where
$$
    R_f = 2\left\langle (-\mathcal{L})^{-1}(\mathcal{L}^2/2 - \left(\mathcal{A}_2 + \mathcal{G}_1\mathcal{A}_1 + \mathcal{G}_2 \right)(-\mathcal{L})^{-1}(f - \pi(f), f - \pi(f)\right\rangle_{\pi}.
$$
Finally, invoking Theorem \ref{th:general} we obtain
$$
    \widehat{\sigma}^2_{\Delta t}(f) = \sigma^2(f) + \Delta t R_f + o(\Delta t),
$$
as required.
% We verify the assumptions of Theorem \ref{thm:correlation_expansion_discrete}.  Assumption \ref{ass:un_erg} holds immediately from Theorem \ref{thm:unif_ergodic} in the appendix. Applying Proposition \ref{prop:mala_expansion} to the reversible dynamics, the  one-step semigroup associated with (\ref{eq:Lie}) is given by $\widehat{P}_{\Delta t} = \Theta_{\Delta t}\widehat{\Phi}_{\Delta t }$, so that for $\phi \in C^\infty(\mathbb{T}^d)$,
% $$
%     \frac{I-\widehat{P}_{\Delta t}}{\Delta t}\phi = \frac{I - \widehat{\Phi}_{\Delta t}}{\Delta t}\phi + \frac{I - \Theta_{\Delta t}}{\Delta t}\widehat{\Phi}_{\Delta t}\phi = \mathcal{L} \phi + \Delta t \mathcal{A}_2 \phi  + \Delta t \mathcal{G}_1\mathcal{A}_1\phi + \Delta t \mathcal{G}_2\phi + \Delta t^{3/2}R_{\phi},
% $$
% where $R_{\phi} \in C^\infty(\mathbb{T}^d)$ is bounded uniformly with respect to $\Delta t$.

% Comparing the coefficients of $\Delta t$ with the expansion of $\widehat{P}_{\Delta t}$ in (\ref{eq:semigroup_expansion}) we obtain the expression for the $O(\Delta t)$ term in (\ref{eq:asympt_var_expansion}).  Finally, it is straightforward to check that
% \begin{align*}
% \mbox{Var}_{\widehat{\pi}^{\Delta t}}\left[f\right] =  \mbox{Var}_{\pi}[f] + \Delta t^r R_4(f),
% \end{align*}
% where  $R_4(f)$ is a smooth function depending on $f$ bounded uniformly with respect to $\Delta t$, thus completing the proof.
\end{proof}

\section{Discussion}
\label{sec:discussion}
In this paper sampling methods based on  nonreversible diffusions have been proposed and evaluated on a range of different inference problems.  The development of these methods is an attempt to improve on existing MCMC methodology  in the case of target densities  that might be of high dimension and exhibit strong correlations. The key idea behind these samplers is the
exploitation of the irreversibility of an underlying diffusion process, which leads to reduced asymptotic variance. This becomes possible through a careful discretisation of the underlying SDE that introduces a controllable bias, but more importantly mimics the reduced asymptotic variance of the nonreversible diffusion.
\\\\
From a practical point of view, the careful balancing of the bias and variance achieved by the nonreversible samplers leads to much more efficient sampling than MALA. In particular, across all our experiments  we observe improvements of two orders of magnitude in terms of effective sample size. Moreover, all our comparisons are being made on the basis of the same number of density evaluations used in the nonreversible samplers and MALA. Furthermore, in the case of the log- Gaussian Cox model the nonreversible samplers are able to achieve this dramatic improvement in terms of the ESS without the need of an expensive  $\mathcal{O}(d^{3})$ reparametrisation, which is also the computational bottleneck in high dimensions  for more sophisticated sampling algorithms such
as  MMALA \cite{girolami2011riemann}.
\\\\
There exist a number of different directions that one could extend this work. In particular, when dealing with the nonreversible part of the dynamics further computational benefits may be achieved with the use of adaptive integration.
Furthermore, one could replace the Metropolis-Hasting scheme used for simulating the reversible part of the dynamics by appropriate numerical schemes \cite{AVZ13} that preserve the invariant measure  to high order. In this situation one would expected the results of our analysis to still hold which is important as the corresponding nonreversible samplers would allow for greater flexibility in the presence of big data, where traditional MCMC methods might become prohibitively expensive.

\section*{Acknowledgements}
G.A. Pavliotis is supported by the Engineering and Physical Sciences Research Council of the UK through Grants Nos.
EP/L020564, EP/L024926 and EP/L025159.  A. B Duncan  ackowledges the EPSRC for support under EP/J009636/1 ,
EP/L020564/1,  EP/K009788/2 (Network on Computational Statistics and Machine Learning).  K. C. Zygalakis was
supported by a grant from the Simons Foundation and by the Alan Turing Institute under the EPSRC grant EP/
N510129/1. Part of this work was done during the author's stay at the Newton Institute for the program Stochastic
Dynamical Systems in Biology: Numerical Methods and Applications. The authors would like to thank Mark Girolami and Andrew Stuart for various discussions about the paper.

%In summary, the careful discretisation of nonreversible diffusion
%provides novel MCMC algorithms whose
%performance has been assessed on a range of statistical models
%and in all cases has been shown to be superior to similar MCMC
%methods.

\appendix
\section{Expansions for the Generator of the Reversible Dynamics}
\label{sec:expansions}
In this section we present the expansion of the generator for a variety of different Metropolised integrators
\subsection{Expansion of the Generator for MALA}
% \begin{equation}
% \begin{equation}
% \label{eq:diffusion_P}
% dX_t = -\nabla V(X_t)\,dt + \sqrt{2}\,dW_t,
% \end{equation}
% where $W_t$ is a $\mathbb{R}^d$ Brownian motion.  The process $X_t$ is ergodic with unique invariant distribution
% $$
% \pi(x) = \frac{1}{Z}e^{-V(x)},
% $$
% where $Z = \int_{\mathbb{R}^d} e^{-V(x')}\,dx'$ is a normalization constant.  Given the current state $q \in \mathbb{R}^d$ and step--size $\Delta t$, the proposal distribution is
Consider the MALA scheme with proposal distribution\footnote{Here $U(x)=-\log \pi(x)$.} $\mathcal{N}\left(x - \nabla U(x)\,\Delta t, 2\Delta t\right), $ having density
$$
	q_{\Delta t}(x' \, | \, x) \propto \exp\left[-\frac{\left\langle x' - \left(x - \nabla U(x)\Delta t\right), \left(x' - \left(x - \nabla U(x)\Delta t \right)\right) \right\rangle}{4\Delta t}\right],
$$
where $\Delta t$ is the stepsize. The acceptance probability is given by
$$
	\alpha(x', x) = \min\left(1, r(x', x)\right),
$$
where
\begin{align*}
	r(x', x) = \frac{\pi(x')q_{\Delta t}(x \, | \, x')}{\pi(x)q_{\Delta t}(x' \, | \, x)} &= e^{-U(x') + U(x) - \frac{1}{4\Delta t}\left\langle x - x' + \nabla U(x')\Delta t, x - x' + \nabla U(x')\Delta t\right\rangle + \frac{1}{4\Delta t}\left\langle x' - x + \nabla U(x)\Delta t, x' - x + \nabla U(x)\Delta t\right\rangle} \\
  &= e^{-U(x') + U(x) + \frac{\Delta t}{4} \lvert \nabla U(x)\rvert^2 - \frac{\Delta t}{4} \lvert \nabla U(x')\rvert^2 + \frac{1}{2}\langle x' - x, \nabla U(x) + \nabla U(x') \rangle} \\
  &= e^{-\lambda(x' , x)}.
\end{align*}
where
$$
	\lambda(x', x) = U(x') - U(x) - \frac{\Delta t}{4}\lvert \nabla U(x)\rvert^2 + \frac{\Delta t}{4} \lvert \nabla U(x')\rvert^2 - \frac{1}{2}\left\langle x' - x, \nabla U(x) + \nabla U(x')\right\rangle
$$
We now Taylor expand $U(x')$ around $x$ up to fourth order, using integral remainders, to obtain
\begin{align*}
U(x') - U(x) = \left\langle \nabla U(x), x' - x\right\rangle &+ \frac{1}{2}\langle x' - x, \nabla\nabla U(x)(x' - x)\rangle\\  &+ \frac{1}{6}\nabla\nabla\nabla U(x):(x' - x)^{\otimes 3} \\ &+ \frac{1}{6}\int_0^1 (1-t)^3 \nabla^4 U((1-t)x + t x'):(x' - x)^{\otimes 4}\,dt,
\end{align*}
and similarly
\begin{align*}
\nabla U(x') = \nabla U(x) + \left\langle \nabla\nabla U(x), x' - x\right\rangle
 &+ \frac{1}{2}\left\langle \nabla\nabla\nabla U(x), x' - x\right\rangle\\
 &+ \frac{1}{2}\int_0^1 (1-t)^2 \nabla^{4}U((1-t)x + t x'):(x' - x)^{\otimes 3}\,dt.
\end{align*}
 Substituting the above expansions in $\lambda(x', x)$, that a number of terms cancel out, leaving
\begin{align*}
\lambda(x', x) = &\frac{1}{6}\nabla\nabla\nabla U(x):(x'-x)^{\otimes 3} + \frac{1}{6}\int_0^1 (1-t)^{3}\nabla^4 U((1-t)x + tx'):(x' - x)^{\otimes 4}\,dt \\
& -\frac{1}{2}\Bigg\langle x' - x, \frac{1}{2}\nabla^{3} U(x):(x' - x)^{\otimes 2}  + \frac{1}{2}\int_0^1 (1-t)^2 \nabla^4 U(x)((1-t)x + tx'):(x'-x)^{\otimes 3}\,dt \Bigg\rangle \\
	&- \frac{\Delta t}{4}\left\lvert \nabla U(x) \right\rvert^2 \\
	&+ \frac{\Delta t}{4}\Bigg\lvert \nabla U(x) + \nabla\nabla U(x):(x'-x) + \frac{1}{2}\nabla^3 U(x):(x'-x)^{\otimes 2} \\ &\qquad + \frac{1}{2}\int_0^1 (1-t)^{2}\nabla^4 U((1-t)x + tx'):(x'-x)^{\otimes 3}\,dt\Bigg\rvert^{2}.
\end{align*}
Our objective is to  obtain explicit expressions for the leading terms in the expansion of $\lambda(x', x)$, in the specific case where
$$
	x' = x - \nabla U(x)\Delta t+ \sqrt{2 \Delta t}G,
$$
where $G \sim \mathcal{N}(0,I)$ and $\Delta t$ is small. Indeed, we have that
\begin{align*}
\lambda\left(x - \nabla U(x)\Delta t+ \sqrt{2\Delta t}G, x\right) &= \Delta t^{3/2}T(x, G) + \Delta t^{2}\xi(x, G),
\end{align*}
where
$$
	T(x, G) = -\frac{\sqrt{2}}{6} \nabla^3 U(x):G^{\otimes 3} + \frac{1}{\sqrt{2}}\left\langle \nabla U(x), \nabla\nabla U(x) G\right\rangle,
$$
and $\xi(x, G)$ collects all terms of order $\Delta t^2$ are higher.  Note that, since $\nabla^k U(q)$ is bounded for all $k \geq 0$, we have:
$$
	|\xi(x, G)| \leq C(1 + |G|^6),
$$
for some constant $C$ independent of $q$ and $G$ and uniformly on $0 \leq \Delta t\leq 1$.  We use this lemma from  \cite{fathi2014error}.
\begin{lemma}
\label{lem:lipschitz_lemma}
	For $a \in \mathbb{R}$, let $a_{+} = 0 \vee a$. Then we have the following relationship:
	\begin{equation}
	\label{eq:taylor_bound1}
		x_{+} - \frac{x^2_{+}}{2} \leq 1 - 1 \wedge e^{-x} \leq x_{+}.
	\end{equation}
	% and similarly
	% \begin{equation}
	% \label{eq:taylor_bound2}
	% 	x_{+} - \frac{x^2_{+}}{2} + \frac{x^3_{+}}{6} - \frac{x^4_{+}}{24} \leq 1 - 1\wedge e^{-x} \leq x_{+} - \frac{x_+^2}{2} + \frac{x_+^3}{6}
	% \end{equation}
\end{lemma}
% \begin{proof}
% When $x \leq 0$ both inequalities hold trivially.  Suppose that $ x > 0$, then
% $$
% 	1 - 1 \wedge e^{-x} = 1 - e^{-x} \leq x,
% $$
% since both sides equal $0$ when $x = 0$, the derivatives satisfy $e^{-x} \leq 1$.  Similarly, if
% $$
% 	a(x) = x - x^2/2 \mbox{ and } b(x) = 1 - e^{-x}
% $$
% then $a(0) = b(0) = 0$, and $a'(0) = b'(0) = 0$ and $a''(x) \leq b''(x)$, so that $a(x) \leq b(x)$.  This proves (\ref{eq:taylor_bound1}).
% % The proof of (\ref{eq:taylor_bound2}) follows in a similar manner.
% \end{proof}
\noindent
As a consequence of this lemma, we have that
\begin{equation}
\label{eq:mala_acceptance_asympt}
\begin{aligned}
\alpha(x - \nabla U(x)\Delta t+ \sqrt{2\Delta t}G, x) &= \min\left(1, e^{-\alpha(x', x)}\right)	\\
		 &= 1 - \Delta t^{3/2}T_{+}(x, G) + \Delta t^2 \widetilde{\xi}(x, G),
\end{aligned}
\end{equation}
where $|\widetilde{\xi}(x, G)| \leq \widetilde{C}(1 + |G|^{12})$.  Given a smooth observable $\psi$, we now consider the effect of the semigroup on $\psi$ over a short time $\delta$.  First note that the transition density of the MALA chain is given by
$$
	p_{MALA}(x' \, | x) = \alpha(x', x) + \delta_{x}(x')\int (1 - \alpha(z, x))\,dz,
$$
and so the semigroup for a single step of size $\Delta t$ is given by
$$
	\widetilde{P}_{\Delta t}\psi(x) = \psi(x) + \mathbb{E}_{G\sim \mathcal{N}(0,I)}\left[\alpha(x - \nabla U(x)\Delta t+ \sqrt{2\Delta t}G, x)\left(\psi(x - \nabla U(x)\Delta t+ \sqrt{2\Delta t}G) - \psi(x)\right)\right].
$$
We split the dynamics into two parts, a part which arises from the proposal, and a part which arises from the acceptance/rejection term.  What we shall observe is that the second term does not contribute to the leading order term.  Indeed, the accept/reject part only has an $O(\Delta t^{2})$ contribution.
\begin{align*}
 &\mathbb{E}_{G\sim \mathcal{N}(0, I)}\left[ \alpha(x - \Delta t\nabla U(x) + \sqrt{2\Delta t}G, x)\left(\psi(x - \Delta t\nabla U(x) + \sqrt{2\Delta t}G) - \psi(x)\right)\right] \\
 &= \underbrace{\mathbb{E}_{G\sim \mathcal{N}(0, I)}\left[\psi(x - \Delta t\nabla U(x) + \sqrt{2\Delta t}G) - \psi(x)\right]}_{A}\\
&\quad + \underbrace{\mathbb{E}_{G\sim \mathcal{N}(0, I)}\left[ \left(\alpha(x - \Delta t\nabla U(q) + \sqrt{2\Delta t}G, q) - 1\right)\left(\psi(q - \Delta t\nabla U(x) + \sqrt{2\Delta t}G) - \psi(x)\right)\right]}_{B}.
\end{align*}
For the first term, we obtain after Taylor expansion of $\psi$:
\begin{align*}
A &= \Delta t\left[-\nabla U(x)\cdot\nabla \psi(x) +  \Delta \psi(x)\right]\\
  & + \Delta t^2\left[\frac{1}{2}\nabla U(x)\cdot\nabla\nabla \psi(x)\nabla U(x) - \frac{1}{3}\nabla U(x)\nabla\Delta \psi +  2 \Delta^2\psi(x) \right] + r_1(x),
\end{align*}
where we use the fact that,
\begin{align*}
	\mathbb{E}_{G \sim \mathcal{N}(0,I)}\left[\nabla U(x)\cdot \nabla\nabla\nabla\psi(x):GG\right] &= \mathbb{E}_{g \sim \mathcal{N}(0,I)}\left[\partial_{x_j} U(x) \partial_{jkl}\psi(x)G_k G_l\right] \\
	&= \partial_{x_j} U(x)\partial_{jkk}\psi(x)\\
	& = \nabla U(x)\cdot \nabla\Delta \psi(x),
\end{align*}
and
\begin{align*}
	\mathbb{E}_{G \sim \mathcal{N}(0,I)}\left[\nabla^4\psi(x):G^{\otimes 4}\right] &= \mathbb{E}_{G \sim \mathcal{N}(0,I)}\left[\partial_{jklm}\psi(x)G_j G_k G_l G_m\right] \\
	&= 12\sum_{j,k} \partial_{jjkk}\psi(x) \\
	&= 12 \Delta^2 \psi(x),
\end{align*}
and where $|r_1(x)| \leq C\Delta t^{5/2}$.  For the second term, using $x' = x - \nabla U(x)\Delta t+ \sqrt{2\Delta t}G$,
\begin{align*}
B &= \mathbb{E}_{G\sim \mathcal{N}(0,I)}\Bigg[(-\Delta t^{3/2}T_+(x, G) + \Delta t^2 \widetilde{\xi}(x, G))\\
&\left(\left\langle -\Delta t\nabla U(x) + \sqrt{2\Delta t}G, \nabla \psi(x)\right\rangle + \int_0^1 (1-t)\nabla\nabla\psi((1-t)x + tx'):(x' - x)^{\otimes 2}\,dt\right)\Bigg] \\
 &= -\sqrt{2}\Delta t^2 \int \left(1\wedge T(x, g)\right) \langle \nabla \psi(x), g\rangle \frac{e^{-|g|^2/2}}{\sqrt{(2\pi)^d}}\,dg + \Delta t^{5/2}r(x)
\end{align*}
Therefore we have that
$$
	P_{\Delta t}\psi(q) - \psi(q) = \Delta t\mathcal{G}_1\psi(q) + \Delta t^2 \mathcal{G}_2 \psi(q) + \Delta t^{5/2} r(q),
$$
where $\mathcal{G}_1$ is the reversible part of the infinitesimal generator (\ref{eq:generator1}), i.e.
\begin{equation}
\label{eq:G1}
    \mathcal{G}_1 = -\nabla U\cdot\nabla + \Delta,
\end{equation}
and
\begin{equation}
\label{eq:G2}
\begin{aligned}
	\mathcal{G}_2 \psi(q) &= \frac{1}{2}\nabla U(q)\cdot\nabla\nabla \psi(q)\nabla U(q) - \frac{1}{3}\nabla U(q)\cdot \nabla\Delta \psi(q) \\ &+  2 \Delta^2 \psi   - \sqrt{2}\int \left(1\wedge T(q, g)\right) \langle \nabla \psi(q), g\rangle \frac{e^{-g\cdot g/2}}{\sqrt{(2\pi)^d }}\,dg,
\end{aligned}
\end{equation}
and $|r(q)| \leq C$, uniformly in $0 < \Delta t\leq 1$.
\\\\
The conclusion of the above is summarised in the following proposition.
\begin{prop}
\label{prop:mala_expansion}
Let $P_{\Delta t}$ denote the evolution operator corresponding to the MALA scheme, then for all smooth $\psi:\mathbb{T}^d \rightarrow \mathbb{R}$:
$$
	(P_{\Delta t} - I)\psi(q) = \Delta t\mathcal{G}_1\psi(q) + \Delta t^2 \mathcal{G}_2 \psi + \Delta t^{5/2}r(q),
$$
where $\mathcal{G}_1$, $\mathcal{G}_2$ are given by (\ref{eq:G1}) and (\ref{eq:G2}), respectively and where $r(q)$ are as given above.  In particular, the infinitesimal generator $\mathcal{G}_{\Delta t}$ corresponding to the MALA scheme satisfies
$$
	{\mathcal{G}}_{\Delta t}\psi(q) = \mathcal{G}_1\psi(q) + \Delta t\mathcal{G}_2 \psi + \Delta t^{3/2}r(q), \quad \forall \psi \in C^\infty(\mathbb{T}^d).
$$
\end{prop}

\section{Analysis for Gaussian Distributions}
\label{app:gaussian_analysis}
In this section we will study the specific example where the dynamics (\ref{eq:sde1}) are linear and of the form
\begin{equation}
\label{eq:linear_diffusion}
	dX_t = -A X_t\,dt + \sqrt{2}\sigma\,dW_t,
\end{equation}
where $W_t$ is a standard $m$-dimensional Brownian motion, $A \in \mathbb{R}^{d\times d}$ and $\sigma \in \mathbb{R}^{d\times m}$ such that $\Sigma = \sigma\sigma^{\top}$ is positive definite.  Provided that $-A$ is stable, and $\Sigma$ is positive definite, $X_t$ is ergodic with unique invariant distribution $\pi(x) \propto \exp(-x\cdot \Sigma_{\infty} x/2)$, where the stationary covariance $\Sigma_{\infty}$ is the solution of the Lyapunov equation
$$
	A \Sigma_{\infty} + \Sigma_{\infty}A^\top = 2\Sigma,
$$
which can be written explicitly as
$$
	\Sigma_{\infty} = 2\int_0^\infty e^{-A s}\Sigma e^{-A^\top s}\,ds
$$
Our objective is to derive an explicit expression for the asymptotic variance $\sigma^2(f)$ of
$$
I_t = \frac{1}{t}\int_0^t f(X_s)\,ds,
$$
where $f$ is a function of the form
$$
	f(x) = x\cdot M x + L \cdot x + K,
$$
for some $M \in \mathbb{R}^{d\times d}_{sym}$, $L \in \mathbb{R}^d$ and $K \in \mathbb{R}$.  Taking a different approach to \cite{duncan2016variance}, we shall obtain this expression via the Green-Kubo formula, i.e.
\begin{equation}
\label{eq:green_kubo}
\sigma^2(f) = 2\int_0^{\infty} \langle P_t f - \pi(f), f -\pi(f) \rangle_{\pi} \,dt = 2\int_0^{\infty} \left(\int_{\mathbb{R}^d} P_t (f-\pi(f))(x) (f(x)-\pi(f))\,\pi(dx)\right)\,dt,
\end{equation}
where $P_t$ is the semigroup corresponding to (\ref{eq:linear_diffusion}).  We note that for $l > 0$, the process $X_t$ satisfies the Foster-Lyapunov condition (\ref{eq:lyapunov_condition}) with Lyapunov function $V_l(x) = 1+|x|^{2l}$.  In particular, by Proposition \ref{prop:CLT}, a CLT for the estimator $\pi_T(f)= T^{-1}\int_0^T f(X_t)\,dt$ will hold for all observables $f$ having algebraic growth, and moreover (\ref{eq:green_kubo}) is well defined and finite.  We shall first prove the result for $\Sigma = I$, and then obtain the general case via a simple linear transformation.  In this case $\Sigma_{\infty} = C^{-1}$, where $C=\mbox{Sym}[A] = \frac{1}{2}\left[A + A^\top\right]$.  To obtain this result we shall make use of the following form of Mehler's formula.

\begin{lemma}
\label{lem:mehler}
Let $P_t$ be the semigroup corresponding to the dynamics
$$
	dX_t = -AX_t\,dt + \sqrt{2}dW_t,
$$
where $W_t$ is a standard $d$--dimensional Brownian motion.  Then, for all $f\in L^2(\pi)$ we have
$$
P_t f(x) = \mathbb{E}\left[f(e^{-A t}x + \sqrt{I - e^{-2Ct}}Z)\right] = \int_{\mathbb{R}^d} f(e^{-At}x + C^{-1}\sqrt{I-e^{-2Ct}}z)\,e^{-|z|^2/2}dz
$$
where $Z \sim \mathcal{N}(0, C^{-1})$.
\end{lemma}
% \begin{proof}
% Define $G(x, t) = e^{A t}x$.  Applying It\^{o}'s formula for $G(X_t)$ we have that
% $$
% 	dG(X_t) = \left(A e^{At}X_t - A e^{At}X_t\right)\,dt + \sqrt{2}e^{A t}\,dW_t,
% $$
% so that
% $$
% X_t = e^{-At}X_0 + \sqrt{2}\int_0^t e^{-A(t-s)}dW_s.
% $$
% The martingale term $\sqrt{2}\int_0^t e^{-A(t-s)}\,dW_s$ has distribution $\mathcal{N}(0, \Sigma_t)$,
% where
% \begin{align*}
% 	\Sigma_t &= 2\int_0^t e^{-A(t-s)}e^{-A^\top(t-s)}\,ds \\
% 		  &= 2e^{-2Ct}\int_0^t e^{2C s}\,ds  \\
% 		  &= e^{-2Ct}C^{-1}(e^{2Ct}-I)\\
% 		  &= C^{-1}(I - e^{-2Ct}),
% \end{align*}
% Noting that for $Z \sim \mathcal{N}(0,C^{-1})$, then $\sqrt{(I-e^{-2Ct})}Z$ has distribution $\mathcal{N}(0,\Sigma)$, we have that
% $$
% P_t f(x) = \mathbb{E}\left[f(e^{-A t}x + \sqrt{I-e^{-2Ct}}Z)\right],
% $$
% as required.
% \end{proof}
% $ $
% \\\\
 \noindent
First consider the observable $f_1(x) = x\cdot M x$ for $M \in \mathbb{R}_{sym}^{d \times d}$, then
\begin{align*}
P_t f_1(x) &= x \cdot e^{-A^\top t}M e^{-At}x + \mbox{Tr}\left[C^{-1}\sqrt{I - e^{-2Ct}}M\sqrt{I - e^{-2Ct}}\right]\\
&= x \cdot e^{-A^\top t}M e^{-At}x + \mbox{Tr}\left[C^{-1}M\left({I - e^{-2Ct}}\right)\right].
\end{align*}
Now, using the fact that $\pi(f_1) = \mbox{Tr}[C^{-1}M]$, we can write
$$
	\int_0^\infty P_t[f_1(x) - \pi(f_1)]\,dt = x\cdot\left[\int_0^\infty e^{-A^\top t}M e^{-At}\,dt\right]x -  \frac{1}{2}\mbox{Tr}[C^{-1}MC^{-1}].
$$
Similarly, if $f_2(x) = L\cdot x$, then $\pi(f_2) = 0$ and
$$
	P_t f_2(x) = x\cdot e^{-A^\top t}L,
$$
so that
$$
	\int_0^\infty P_t f_2(x)-\pi(f_2) = x\cdot \left(A^{\top}\right)^{-1} L.
$$
Thus, it follows that the unique, mean-zero solution of the Poisson equation
$$
	-Ax\cdot\nabla \phi(x) + \Delta \phi(x) = f(x) - \pi(f),
$$
is given by
$$
	\phi(x) = x\cdot\left[\int_0^\infty e^{-A^\top t}M e^{-At}\,dt\right]x - \mbox{Tr}\left[C^{-1}\int_0^\infty e^{-A^\top t}Me^{-At}\,dt\right]+ x\cdot \left(A^{\top}\right)^{-1} L.
$$
We now use the following Green-Kubo type formula to compute the asymptotic variance for $f = f_1 + f_2$:
\begin{align*}
\frac{1}{2}\sigma^2(f) &= \int_{\mathbb{R}^d}\int_0^\infty \left[P_t f(x) - \pi(f)\right]\left[f(x) - \pi(f)\right] \,dt\, \pi(dx)\\
		&= \int \left[\left(x\cdot\Pi x\right)\left(x\cdot M x\right) - \mbox{Tr}[M C^{-1}]x\cdot \Pi x   + \left(x\cdot A^{-\top}L\right)^2\right] \pi(dx) \\
		&= 2\mbox{Tr}\left[C^{-1}\Pi C^{-1} M\right]  + L\cdot A^{-1}C^{-1}A^{-\top}L,
	\end{align*}
	where
	$$
	\Pi = \int_0^\infty e^{-A^\top t}M e^{-At}\,dt.
	$$
In the case when  $A = I + \alpha J$, for $J^\top = -J$, so that $C = I$, we obtain
$$
\frac{1}{2}\sigma^2(f) = 2\,\mbox{Tr}\left[\left(\int_0^\infty e^{-A^\top t}M e^{-At}\,dt\right)M\right] + L\cdot (I + \alpha^2 JJ^\top)^{-1} L,
$$
which is precisely the formula derived in \cite{duncan2016variance} using a different approach.
\\\\
Now suppose that $\Sigma \neq I$, so that the Poisson equation we must solve becomes
$$
	-Ax\cdot\nabla\phi(x) + \Sigma:\nabla\nabla \phi(x) = f(x) - \pi(f).
$$
Writing $\phi(x) = \psi(\Sigma^{-1/2}x)$, then we have that
$$
	-Ax\cdot \Sigma^{-1/2}\nabla\psi(\Sigma^{-1/2}x) + \Delta \psi(\Sigma^{-1/2}x) = f(x) - \pi(f),
$$
so that
\begin{equation}
\label{eq:linear_operator_poisson}
	-\Sigma^{-1/2}A\Sigma^{1/2}x\cdot\nabla \psi(x) + \Delta \psi(x)=  f(\Sigma^{1/2}x)-\pi(f).
\end{equation}
The linear operator defined on the left hand side of \eqref{eq:linear_operator_poisson} corresponds to a linear diffusion with stationary distribution $\widetilde{\pi} = \mathcal{N}(0, \widetilde{\Sigma}_{\infty})$ where
$$
\widetilde{\Sigma}_{\infty} = \int_0^\infty e^{A_\Sigma t}e^{A_\Sigma^\top t}\,dt = \Sigma^{-1/2}\Sigma_{\infty}\Sigma^{-1/2}.
$$
where $A_\Sigma = \Sigma^{-1/2}A\Sigma^{1/2}$.
%  Now since we have
% $$
% A\Sigma_{\infty} + \Sigma_{\infty}A^{\top} = 2\Sigma,
% $$
% and since $\widetilde{\Sigma}_{\infty}$ satisfies
% $$
% \Sigma^{-1/2}A \Sigma^{1/2} \widetilde{\Sigma}_{\infty} + \widetilde{\Sigma}_{\infty}\Sigma^{1/2}A^\top \Sigma^{-1/2} = 2I,
% $$
% this implies that $\Sigma^{-1/2}\Sigma_{\infty}\Sigma^{-1/2} = \widetilde{\Sigma}_{\infty}$.
Computing the asymptotic variance $\sigma^2(f)$:
\begin{align*}
2\int_{\mathbb{R}^d} \phi(x) (f(x)-\pi(f))\pi(x)\, &= 2\int_{\mathbb{R}^d} \psi(\Sigma^{-1/2}x)(f(x)-\pi(f))\pi(x) \, dx \\ &= 2\int_{\mathbb{R}^d} \psi(x)[f(\Sigma^{1/2}x)-\pi(f)]\widetilde{\pi}(x)\,dx.
\end{align*}
Applying the previous result, it follows that
\begin{align*}
\frac{1}{2}\sigma^2(f) = 2\mbox{Tr}\left[\left(\int_0^\infty e^{-A_\Sigma^{\top}t}M_\Sigma e^{-A_\Sigma t}\,dt\right)M\right] +  L\cdot A_\Sigma^{-1}\widetilde{\Sigma}_{\infty}^{-1}A_\Sigma^{-\top}L,
\end{align*}
where  $M_\Sigma = \Sigma^{1/2}M\Sigma^{1/2}$.  Noting that $e^{-A_\Sigma t} = \Sigma^{-1/2}e^{-At}\Sigma^{1/2}$, it follows that
$$
 2\mbox{Tr}\left[\left(\int_0^\infty \Sigma^{1/2}e^{-A^{\top}t}Me^{-At}\Sigma^{1/2}\,dt\right)M\right] = 2\mbox{Tr}\left[\left(\int_0^\infty e^{-A^{\top}t}Me^{-At}\,dt\right)M_\Sigma\right]
$$
and moreover,
$$
 L\cdot \Sigma^{-1/2}A^{-1}\Sigma^{1/2}\widetilde{\Sigma}_{\infty}\Sigma^{1/2}A^{-\top}\Sigma^{-1/2}L = L\cdot \Sigma^{-1/2}A^{-1}\Sigma_{\infty} A^{-\top}\Sigma^{-1/2}L= L_\Sigma\cdot A^{-1}\Sigma_{\infty} A^{-\top}L_\Sigma,
$$
where $L_\Sigma = \Sigma^{-1/2}L$.   In summary we have the following result.

\begin{prop}
Consider the linear diffusion defined by the SDE,
$$
dX_t = -AX_t\,dt + \sqrt{2}\sigma\,dW_t,
$$
where $W_t$ is a $m$-dimensional Brownian motion, $\sigma \in \mathbb{R}^{d\times m}$ such that $\Sigma = \sigma\sigma^{\top}$ is positive definite and $-A$ is stable.  Then, for
$$
f(x) = x\cdot M x + L\cdot x + K,
$$
the asymptotic variance $\sigma^2(f)$ is given by
$$
	\frac{1}{2}\sigma^2(f) = 2\mbox{Tr}\left[\left(\int_0^\infty e^{-A^{\top}t}Me^{-At}\,dt\right)M_\Sigma\right] +  L_\Sigma\cdot A^{-1}\Sigma_{\infty} A^{-\top}L_\Sigma,
$$
where $M_\Sigma = \Sigma^{1/2}M\Sigma^{1/2}$ and $L_\Sigma = \Sigma^{-1/2}L$.
\end{prop}
\begin{remark}
Note that there is no impediment to deriving the asymptotic variance for observables involving higher powers, e.g. a third order tensor of the form $\sum_{i,j,k}K_{i,j,k}x_ix_jx_k$, but we only provide the result up to second order for the sake of clarity.   A more general approach would potentially be possible by considering the decomposition of an observable $f$ with respect to the eigenbasis of the Ornstein Uhlenbeck operator $\mathcal{L}$, which can be shown to be Hermite polynomials \cite{metafune2002spectrum}.
\end{remark}
%\textbf{[AD]: NOTE THAT THE NOISE IS ASSUMED TO BE NON-DENEGERATE...WILL NEED TO THINK ABOUT DEGENERATE NOISE IF WE GO DOWN THAT ROUTE, MAYBE FOLLOW ERB \& ARNOLD 2014 PAPER.}

\section{Spectral Gap estimate for the Splitting Scheme}

In this section we shall focus specifically on the splitting scheme where the reversible component is simulated using MALA where we show that Assumption \ref{ass:un_erg} holds in this case.  The approach we follow is strongly based on   arguments found in \cite{fathi2014error,bou2012nonasymptotic}.  The method depends strongly on the fact that the proposal of the MALA scheme is a first order approximation of the corresponding SDE.   In this section, we shall assume that Assumption \ref{ass:uniform_lipschitz} holds.
\\\\
Define $Q_{\Delta t}(x, y)$ to be the transition kernel for the exact dynamics
\begin{equation}
\label{eq:exact_dynamics}
    dY_t = \left(\nabla \log \pi(Y_t) + \gamma(Y_t)\right)\,dt + \sqrt{2}\,dW_t,
\end{equation}
and $\widetilde{P}_{\Delta t}(x,y)$ the transition kernel of the unadjusted Lie-Trotter scheme defined by
\begin{equation}
\label{eq:unadjusted}
    \widetilde{Z}_{n+1} = \Phi_{\Delta t}(\widetilde{Z}_{n}) + \nabla\log \pi(\Phi_{\Delta t}(\widetilde{Z}_n)){\Delta t} + \sqrt{2{\Delta t}}\,\xi_n,
\end{equation}
where $\xi_n \sim \mathcal{N}(0, I)$, and  $\widetilde{Q}_{\Delta t}(x,y)$ to be the transition kernel of the Euler-Maruyama discretisation of (\ref{eq:exact_dynamics}), i.e.
\begin{equation}
\widetilde{Y}_{n+1} = \widetilde{Y}_n + \nabla \log\pi(\widetilde{Y}_n){\Delta t} + \gamma(\widetilde{Y}_n){\Delta t} + \sqrt{2{\Delta t}}\,\xi_n.
\end{equation}

\begin{lemma}
\label{lem:error1}
Consider a coupling of $\widetilde{Z}_{n}$ and $\widetilde{Y}_{n}$ such that they are driven by the same noise $W_t$, and  $\widetilde{Z}_0 = \widetilde{Y}_0 = x$.  Then, for $t = n {\Delta t}$ there exists a constant $C(t) > 0$ independent of ${\Delta t}$ such that
\begin{equation}
\mathbb{E}_x\left\lvert \widetilde{Y}_n - \widetilde{Z}_n \right\rvert \leq C(t){\Delta t}
\end{equation}
for ${\Delta t}$ sufficiently small.
\end{lemma}
\begin{proof}
Using the fact that $\Phi_{\Delta t}(x) = x + \gamma(x){\Delta t} + K_1(x){\Delta t}^2$
for some function $K_1$ bounded uniformly on $\mathbb{T}^d$ for ${\Delta t}$ sufficiently small, we have that
\begin{align*}
\mathbb{E}_x \lvert \widetilde{Y}_{n+1} - \widetilde{Z}_{n+1} \rvert &\leq \mathbb{E}_x \lvert \Phi_{\Delta t}(\widetilde{Z}_{n}) + \nabla \log\pi(\Phi_{\Delta t}(\widetilde{Z}_n)){\Delta t} - \widetilde{Y}_n - \nabla \log\pi(\widetilde{Y}_n){\Delta t} -\gamma(\widetilde{Y}_n){\Delta t}\rvert\\
&\leq \mathbb{E}_x \lvert \widetilde{Z}_{n} - \widetilde{Y}_n \rvert + {\Delta t}\mathbb{E}_x\lvert \gamma(\widetilde{Y}_n) - \gamma(\widetilde{Z}_n)\rvert +  {\Delta t}\mathbb{E}_x\lvert \nabla \log\pi(\widetilde{Y}_n) - \nabla \log\pi(\widetilde{Z}_n)\rvert + C {\Delta t}^2 \\
&\leq (1 + K {\Delta t})\mathbb{E}_x \lvert \widetilde{Z}_{n} - \widetilde{Y}_n \rvert + C {\Delta t}^2,
\end{align*}
where $C > 0$ is a constant, from which the result follows.
\end{proof}
\begin{remark}
It follows automatically from (\ref{lem:error1}) and standard estimates for Euler-Maruyama discretisation of SDEs with additive noise that, for ${\Delta t}$ sufficiently small, there exists $C(t) > 0$ such that
$$
    \mathbb{E}_x \left\lvert Z_n - Y_{n {\Delta t}} \right\rvert \leq C(t){\Delta t},
$$
where $t = n{\Delta t}$.
\end{remark}
\begin{theorem}
\label{thm:unif_ergodic}
Suppose that Assumption \ref{ass:uniform_lipschitz} holds,  then the Lie-Trotter scheme (\ref{eq:Lie}) posesses a unique invariant distribution $\pi_{\Delta t}$ and moreover, there exists ${\Delta t}_0 > 0$, $C, \lambda > 0$ such that, for any $0 < {\Delta t} \leq {\Delta t}_0$ and all $n \in \mathbb{N}$
$$
    \lVert P_{{\Delta t}}^{n} f  - \pi_{\Delta t}(f)\rVert_{\infty} \leq C e^{-\lambda n {\Delta t}}\lVert f - \pi_{\Delta t}(f) \rVert_{\infty}, \quad f \in L^{\infty}(\mathbb{T}^d).
$$
As a consequence, there exists a constant $K > 0$, independent of ${\Delta t}$ such that
\begin{equation}
\label{eq:discrete_poisson_bound2}
    \left\lVert \frac{I - P_{\Delta t}}{\Delta t} \right\rVert_{L^\infty_0(\pi_{\Delta t})} < K,
\end{equation}
for ${\Delta t}$ sufficiently small.
\end{theorem}
\begin{proof}
Denote by $P_{\Delta t}^{n}(x,y)$ the transition kernel density corresponding to (\ref{eq:Lie}).  Since the domain is compact, we need only verify that we have the following uniform minorisation condition
$$
    \left\lVert P_{\Delta t}^{\lfloor 1/{\Delta t} \rfloor}(x, \cdot) - P_{\Delta t}^{\lfloor 1/{\Delta t} \rfloor}(y, \cdot)\right\rVert_{TV} \leq 2 (1 - \gamma),\qquad x, y \in \mathbb{T}^d,
$$
for some $\gamma \in (0,1)$, uniformly on ${\Delta t}$.  Then
\begin{align*}
\sup_{x,y\in \mathbb{T}^d}\left\lVert P_{\Delta t}^{\lfloor 1/{\Delta t} \rfloor}(x, \cdot) - P_{\Delta t}^{\lfloor 1/{\Delta t} \rfloor}(y, \cdot)\right\rVert_{TV} \leq  2&\sup_{x\in \mathbb{T}^d}\left\lVert P_{\Delta t}^{\lfloor 1/{\Delta t} \rfloor}(x, \cdot) - Q_{\Delta t}^{\lfloor 1/{\Delta t} \rfloor}(x, \cdot)\right\rVert_{TV}  \\
                            + &\sup_{x,y\in \mathbb{T}^d}\left\lVert Q_{\Delta t}^{\lfloor 1/{\Delta t} \rfloor}(x, \cdot) - Q_{\Delta t}^{\lfloor 1/{\Delta t} \rfloor}(y, \cdot)\right\rVert_{TV}.
\end{align*}
Now, since the domain is compact, and the diffusion process $Z_t$ is uniformly elliptic,  we know that the transition kernel $Q_{{\Delta t}}(\cdot, \cdot)$ satisfies a minorisation condition, and thus there exists $\gamma_1 > 0$ such that
\begin{equation}
\label{eq:minorization}
    \sup_{x,y\in \mathbb{T}^d}\left\lVert Q_{\Delta t}^{\lfloor 1/{\Delta t} \rfloor}(x, \cdot) - Q_{\Delta t}^{\lfloor 1/{\Delta t} \rfloor}(y, \cdot)\right\rVert_{TV} \leq 2(1- \gamma_1).
\end{equation}
Following \cite{bou2012nonasymptotic,fathi2014error} we introduce the transition kernel of the un-adjusted splitting scheme (\ref{eq:unadjusted}), denoted by $\widetilde{P}_{\Delta t}(x, y)$.  Then we have
\begin{equation}
\label{eq:minorization2}
\begin{aligned}
\sup_{x\in \mathbb{T}^d}\left\lVert P_{\Delta t}^{\lfloor 1/{\Delta t} \rfloor}(x, \cdot) - Q_{\Delta t}^{\lfloor 1/{\Delta t} \rfloor}(x, \cdot)\right\rVert_{TV}  \leq &\sup_{x\in \mathbb{T}^d}\left\lVert \widetilde{P}_{\Delta t}^{\lfloor 1/{\Delta t} \rfloor}(x, \cdot) - Q_{\Delta t}^{\lfloor 1/{\Delta t} \rfloor}(x, \cdot)\right\rVert_{TV} \\ + &\sup_{x\in \mathbb{T}^d}\left\lVert \widetilde{P}_{\Delta t}^{\lfloor 1/{\Delta t} \rfloor}(x, \cdot) - P_{\Delta t}^{\lfloor 1/{\Delta t} \rfloor}(x, \cdot)\right\rVert_{TV}
\end{aligned}
\end{equation}
To control the second term on the RHS of (\ref{eq:minorization2}) we apply a coupling argument, identical to that of \cite[Lemma 3.2]{bou2012nonasymptotic} and \cite[Lemma 2]{fathi2014error}.  To this end,  consider the processes $Z_n$ and $\widetilde{Z}_n$ defined by (\ref{eq:Lie}) and (\ref{eq:unadjusted}) respectively, and assume that they are driven by the same noise process, starting from $Z_0 = \widetilde{Z}_0 = x$.  Using the coupling characterization of total variation
$$
\sup_{x\in \mathbb{T}^d}\left\lVert \widetilde{P}_{\Delta t}^{n}(x, \cdot) - P_{\Delta t}^{n}(x, \cdot)\right\rVert_{TV}\leq  2 \mathbb{P}\left[Z_n \neq \widetilde{Z}_n \, | Z_0 = \widetilde{Z}_0 = x\right] \leq 2 \sum_{i=1}^{n} \mathbb{E}\left[( 1 - \alpha(\Phi_h(Z_i), Z_{i+1})) \,|\, Z_0 = x \right],
$$
where $\alpha(x,y)$ is the probability of the standard MALA scheme of accepting a transition from $x$  to $y$.   From \eqref{eq:mala_acceptance_asympt}, (see also \cite[Lemma 1]{fathi2014error}), there exists $C_1 > 0$ such that
$$
  \mathbb{E}\left[( 1 - \alpha(\Phi_{\Delta t}(Z_i), Z_{i+1})) \,|\, Z_0 = x \right] \leq C_1 {\Delta t}^{3/2},
$$
for ${\Delta t}$ sufficiently small.  We can rewrite the first term on the RHS of (\ref{eq:minorization2}) as
\begin{align*}
\sup_{x\in \mathbb{T}^d}\left\lVert \widetilde{P}_{\Delta t}^{\lfloor 1/{\Delta t} \rfloor}(x, \cdot) - Q_{\Delta t}^{\lfloor 1/{\Delta t} \rfloor}(x, \cdot)\right\rVert_{TV}  \leq &\sup_{x\in \mathbb{T}^d}\left\lVert \widetilde{P}_{\Delta t}^{\lfloor 1/{\Delta t} \rfloor}(x, \cdot) - \widetilde{P}_{\Delta t}\circ Q_{\Delta t}^{\lfloor 1/{\Delta t} \rfloor - 1}(x, \cdot)\right\rVert_{TV} \\
+ &\sup_{x\in \mathbb{T}^d}\left\lVert  \widetilde{P}_{\Delta t}\circ Q_{\Delta t}^{\lfloor 1/{\Delta t} \rfloor - 1}(x, \cdot) - Q_{\Delta t}^{\lfloor 1/{\Delta t} \rfloor}(x, \cdot) \right\rVert_{TV}.
\end{align*}
Noting that
\begin{align*}
\left\lVert \widetilde{P}_{\Delta t}^{\lfloor 1/{\Delta t} \rfloor}(x, \cdot) - \widetilde{P}_{\Delta t}\circ Q_{\Delta t}^{\lfloor 1/{\Delta t} \rfloor - 1}(x, \cdot)\right\rVert_{TV} = \mathbb{E}_x \left\lVert \widetilde{P}_{\Delta t}(\widetilde{Z}_{\lfloor 1/{\Delta t}\rfloor-1}, \cdot) - \widetilde{P}_{\Delta t}(Y_{\lfloor 1/{\Delta t}\rfloor-1}, \cdot)\right\rVert_{TV}.
\end{align*}
From Pinkser's inequality,we get
\begin{align*}
\left\lVert \widetilde{P}_{\Delta t}(x, \cdot) - \widetilde{P}_{\Delta t}(y, \cdot)\right\rVert_{TV} &\leq \frac{\left\lvert \Phi_{\Delta t}(x) + \nabla \log\pi(\Phi_{\Delta t}(x)){\Delta t} - \Phi_{\Delta t}(y) - \nabla \log\pi(\Phi_{\Delta t}(y)){\Delta t}\right\rvert}{\sqrt{2{\Delta t}}} \\
&\leq \frac{1}{\sqrt{2{\Delta t}}}\left\lvert \Phi_{\Delta t}(x)  - \Phi_{\Delta t}(y)\right\rvert + \sqrt{\frac{{\Delta t}}{2}}\left\lvert  \nabla \log \pi(\Phi_{\Delta t}(x)) -  \nabla \log\pi(\Phi_{\Delta t}(y))\right\rvert\\
&\leq K\left({\Delta t}^{1/2} + {\Delta t}^{-1/2}\right)\left\lvert x-y\right\rvert,
\end{align*}
using (\ref{eq:lipschitz_constant}) and the fact that $\nabla\nabla \log\pi$ is bounded uniformly on $\mathbb{T}^d$.  Therefore, from the remark following Lemma \ref{lem:error1},
$$
 \mathbb{E}_x \left\lVert \widetilde{P}_{\Delta t}(\widetilde{Z}_{\lfloor 1/{\Delta t}\rfloor-1}, \cdot) - \widetilde{P}_{\Delta t}(Y_{\lfloor 1/{\Delta t}\rfloor-1}, \cdot)\right\rVert_{TV} \leq K{\Delta t}^{-1/2} \mathbb{E}_x\left\lVert \widetilde{Z}_{\lfloor 1/{\Delta t}\rfloor-1} - Y_{\lfloor 1/{\Delta t}\rfloor-1}\right\rVert \leq C {\Delta t}^{1/2},
$$
where $C$ is a constant independent of ${\Delta t}$.  We introduce an intermediate continuous time process $\widetilde{Y}_t$ defined on $[0,\Delta t]$ by
$$
d\widetilde{Y}_t = \left(\gamma(x) + \nabla \log\pi(x)\right)\,dt + \sqrt{2}\,dW_t,\quad \widetilde{Y}_0 = x,
$$
with corresponding transition kernel $\widetilde{Q}_t(x, \cdot)$.  Then we have
\begin{equation}
\label{eq:tv_bound3}
\begin{aligned}
\left\lVert  \widetilde{P}_{\Delta t}\circ Q_{\Delta t}^{\lfloor 1/{\Delta t} \rfloor - 1}(x, \cdot) - Q_{\Delta t}^{\lfloor 1/{\Delta t} \rfloor}(x, \cdot) \right\rVert_{TV} = &\mathbb{E}_x\left\lVert \widetilde{P}_{\Delta t}(Y_{\lfloor 1/{\Delta t}\rfloor-1},\cdot) - Q_{\Delta t}(Y_{\lfloor 1/{\Delta t}\rfloor-1}, \cdot)\right\rVert_{TV} \\
&\quad \leq \mathbb{E}_x\left\lVert \widetilde{Q}_{\Delta t}(Y_{\lfloor 1/{\Delta t}\rfloor-1},\cdot) - Q_{\Delta t}(Y_{\lfloor 1/{\Delta t}\rfloor-1}, \cdot)\right\rVert_{TV} \\
&\qquad + \mathbb{E}_x\left\lVert \widetilde{P}_{\Delta t}(Y_{\lfloor 1/{\Delta t}\rfloor-1},\cdot) - \widetilde{Q}_{\Delta t}(Y_{\lfloor 1/{\Delta t}\rfloor-1}, \cdot)\right\rVert_{TV}.
\end{aligned}
\end{equation}
From Pinsker's inequality we have
\begin{align*}
\left\lVert \widetilde{P}_{\Delta t}(x,\cdot) - \widetilde{Q}_{\Delta t}(x, \cdot)\right\rVert_{TV} \leq \frac{\left|\Phi_{\Delta t}(x) + \nabla \log\pi(\Phi_{\Delta t}(x)){\Delta t} - x - \nabla \log\pi(x){\Delta t} - \gamma(x){\Delta t}\right|}{\sqrt{2{\Delta t}}}.
\end{align*}
Using the fact that
$$
\Phi_{\Delta t}(x) = x + \gamma(x){\Delta t} + K_1(x){\Delta t}^2,
$$
for some function $K_1$ bounded uniformly on $\mathbb{T}^d$ for ${\Delta t}$ sufficiently small, and applying Taylor's theorem for $\nabla \log \pi$, we obtain the bound
\begin{align*}
\left\lVert \widetilde{P}_{\Delta t}(x,\cdot) - \widetilde{Q}_{\Delta t}(x, \cdot)\right\rVert_{TV} \leq K {\Delta t}^{3/2},
\end{align*}
for some constant $K$ independent of ${\Delta t}$.  Denote by $\widetilde{\mathbb{Q}}_{\Delta t}$ and $\mathbb{Q}_{\Delta t}$ the path measures on $C[0,{\Delta t}]$ induced by the processes $\widetilde{Y}_t$ and $Y_t$, respectively.  Then by Girsanov's theorem \cite[Ch. 3, Cor. 5.16]{karatzas2012brownian} we obtain
\begin{align*}
\frac{d\mathbb{Q}_{\Delta t}}{d\widetilde{\mathbb{Q}}_{\Delta t}}(\widetilde{Y}_t) &= \exp\Bigg(\frac{1}{2}\int_0^{\Delta t} \left\langle \gamma(\widetilde{Y}_s) + \nabla \log\pi(\widetilde{Y}_s) - \gamma(x) - \nabla \log\pi(x), d\widetilde{Y}_s \right\rangle  \\
&\qquad - \frac{1}{4}\int_0^{\Delta t}\lvert \gamma(\widetilde{Y}_s) + \nabla \log\pi(\widetilde{Y}_s) \rvert^2 - \lvert \gamma(x) + \nabla \log\pi(x)\rvert^2 \,ds \Bigg).
\end{align*}
By Pinsker's inequality, it follows that
\begin{align*}
\lVert \widetilde{Q}_{\Delta t}(x,\cdot) - Q_{\Delta t}(x, \cdot)\rVert_{TV}^2 &\leq \frac{1}{2}\mathbb{E}_x\int_0^{\Delta t} \left\vert \gamma(\widetilde{Y}_s) + \nabla \log\pi(\widetilde{Y}_s) - \gamma(x) - \nabla \log\pi(x) \right\rvert^2 \,ds \\
&\leq  \frac{1}{2}\mathbb{E}_x\int_0^{\Delta t} \left\lvert \gamma(\widetilde{Y}_s) + \nabla \log\pi(\widetilde{Y}_s) - \gamma(x) - \nabla \log\pi(x) \right\rvert^2 \,ds \\
&\leq  C\int_0^{\Delta t} \mathbb{E}_x\left\lvert \widetilde{Y}_s - x \right\rvert^2 \,ds \\
&\leq C \int_0^{\Delta t} |\gamma(x) + \nabla \log\pi(x)|^2 s^2 + s \,ds \\
&\leq C \left(|\gamma(x) + \nabla \log\pi(x)|^2 \frac{{\Delta t}^3}{3} + \frac{{\Delta t}^2}{2}\right),
\end{align*}
and so, there exists a constant $K > 0$ independent of ${\Delta t}$ such that
\begin{equation}
\label{eq:bound3}
\mathbb{E}_x\left\lVert \widetilde{Q}_{\Delta t}(Y_{\lfloor 1/{\Delta t}\rfloor-1},\cdot) - Q_{\Delta t}(Y_{\lfloor 1/{\Delta t}\rfloor-1}, \cdot)\right\rVert_{TV} \leq K {\Delta t}.
\end{equation}
Collecting the terms together, it follows that for ${\Delta t}$ sufficiently small, condition (\ref{eq:minorization}) holds.   The bound (\ref{eq:discrete_poisson_bound2}) then follows immediately.
% Consider
% $$
% \sup_{x\in \mathbb{T}^d}\left\lVert  \widetilde{P}_h\circ Q_h^{\lfloor 1/h \rfloor - 1}(x, \cdot) - Q_h^{\lfloor 1/h \rfloor}(x, \cdot) \right\rVert_{TV} =\sup_{x\in \mathbb{T}^d}\left\lVert  \widetilde{P}_h\circ Q_h^{\lfloor 1/h \rfloor - 1}(x, \cdot) - Q_h^{\lfloor 1/h \rfloor}(x, \cdot) \right\rVert_{TV}
% $$
% Consider
% $$
% \sup_{x\in \mathbb{T}^d}\left\lVert \widetilde{P}_h^{\lfloor 1/h \rfloor}(x, \cdot) - \widetilde{P}_h\circ Q_h^{\lfloor 1/h \rfloor - 1}(x, \cdot)\right\rVert_{TV} \leq \sup_{x\in \mathbb{T}^d}\mathbb{E}_{x}\left\lVert \widetilde{P}_h(Z^x_{[1/h]-1}, \cdot) - \widetilde{P}_h(Y^x_{([1/h]-1)h}, \cdot) \right\rVert_{TV},
% $$
% where $Z_{n}^{x}$ and $Y_t^{x}$ denote the unadjusted MALA scheme (\ref{eq:unadjusted}) and exact scheme, respectively, both starting from $x$
\end{proof}

%\bibliography{complete1,complete,abd_biblio,refs}
%\bibliography{../bibtex_files/mybib}
\bibliographystyle{plain}

%\bibliography{refs}

\end{document}